%% file: markov_random_fields.tex
\documentclass[a4paper,onecolumn, superscriptaddress,10pt,shorttitle=papers]{compositionalityarticle}

 \pdfoutput=1
\usepackage[utf8]{inputenc}
\usepackage[english]{babel}
\usepackage[T1]{fontenc}
\usepackage{amsmath}

\newcommand{\compositionalityIssue}{1}
\newcommand{\compositionalityVolume}{8}
\newcommand{\paperdoinumber}{10.46298/compositionality-\compositionalityVolume-\compositionalityIssue}
\newcommand{\paperdoi}{https://doi.org/\paperdoinumber}
\newcommand{\receiveddate}{2024-10-23}
\newcommand{\accepteddate}{2025-12-20}
\usepackage[sort&compress]{natbib}

\usepackage[immediate]{silence}

\usepackage{tikz-cd}

\input{math_commands.tex}
\usepackage{hyperref}

\hypersetup{colorlinks, linkcolor = blue, citecolor=magenta, linktocpage=true}

\usepackage[font=small,labelfont=bf]{caption, subcaption} % Figures look better

\begin{document}

\title{Abstract Markov Random Fields}
\date{}
\author{Leon Lang}
\email{l.lang@uva.nl}
\thanks{Main contributing author.}
\orcid{0000-0002-1950-2831}
\affiliation{Informatics Institute, University of Amsterdam}
\author{Clélia de Mulatier}
\email{c.m.c.demulatier@uva.nl}
\affiliation{Informatics Institute, University of Amsterdam}
\affiliation{Institute for Theoretical Physics, University of Amsterdam}
\orcid{0000-0003-3578-5453}
\author{Rick Quax}
\email{r.quax@uva.nl}
\affiliation{Informatics Institute, University of Amsterdam}
\orcid{0000-0002-0299-0074}
  \author{Patrick Forré}
\email{p.d.forre@uva.nl}
\affiliation{Korteweg-de Vries Institute for Mathematics, University of Amsterdam}
\orcid{0000-0003-4663-3842}

\maketitle

\begin{abstract}
  Markov random fields are known to be fully characterized by properties of their information diagrams, or $\Ent$-diagrams. 
  In particular, for Markov random fields, regions in the $\Ent$-diagram corresponding to disconnected vertex sets in the graph vanish. 
  Recently, $\Ent$-diagrams have been generalized to $\CR$-diagrams, for a larger class of functions $\CR$ satisfying the chain rule beyond Shannon entropy, such as Kullback-Leibler divergence and cross-entropy. 
  In this work, we generalize the notion and characterization of Markov random fields to this larger class of functions $\CR$ and investigate preliminary applications.

  We define $\CR$-independences, $\CR$-mutual independences, and $\CR$-Markov random fields and characterize them by their $\CR$-diagram. 
  In the process, we also define $\CR$-dual total correlation and prove that its vanishing is equivalent to $\CR$-mutual independence. 
  We then apply our results to information functions $\CR$ that are applied to probability mass functions. 
  We show that if the probability distributions of a set of random variables are Markov random fields for the same graph, then we formally recover the notion of an $\CR$-Markov random field for that graph. 
  We then study the Kullback-Leibler diagrams on specific Markov chains, leading to a visual representation of the second law of thermodynamics and a simple explicit derivation of the decomposition of the evidence lower bound for diffusion models.
\end{abstract}

%\newpage

\tableofcontents

%\newpage

\section{Introduction}\label{sec:Introduction}

Entropy, mutual information, and higher-order information terms between several random variables can be visualized in information diagrams, also known as $I$-diagrams.
This is known as Hu's Theorem~\citep{Hu1962,Yeung1991}.
These diagrams become especially interesting when the visualized variables obey conditional independences, which then implies that some regions in these diagrams vanish.
A series of papers~\citep{Kawabata1992,Yeung2002a,Yeung2019} exploited this idea to study Markov random fields, which according to the global Markov property satisfy a set of conditional independences determined by an underlying graph~\citep{Preston1976,Spitzer1971,Hammersley1971}.
This then implies that all intersections in the $I$-diagram corresponding to \emph{disconnected vertex sets} in the graph vanish, a result we visualize for simple graphs in Figure~\ref{fig:three_variables_mrfs}.

Recently, Hu's theorem has been generalized from entropy $I$ to more general functions $\CR$ on commutative, idempotent monoids satisfying a chain rule~\citep{Lang2022}. 
Examples for $\CR$ are Kullback-Leibler divergence, cross-entropy, Tsallis entropy, and even Kolmogorov complexity. 
The resulting $\CR$-diagrams show structurally the same relations as $I$-diagrams, and thus allow higher-order $\CR$-terms to be visualized and reasoned about in a unified way --- see Figure~\ref{fig:venn_diagram_comparison}.

Functions such as the cross-entropy or Kullback-Leibler divergence are important in the context of statistical modeling of multivariate data, in which one aims to find a probabilistic model able to reproduce the information structure of the data. 
For instance, the $\CR$-diagram for cross-entropy allows us to visualize how the cross-entropy between a model probability distribution and the data distribution is decomposed into higher-order terms.  
\cite{Cocco2012} used these higher-order terms (which they called cluster (cross)-entropies) in their adaptive cluster expansion approach to statistical modeling of data with Ising models. 
Kullback-Leibler divergence has been studied in the context of decompositions of joint entropy and information~\citep{Amari2001} and is ubiquitous in machine learning. 
In other contexts, for example Kolmogorov complexity, the precise meaning of the higher-order terms is not yet clear.

In this work, we take the generalization of $I$-diagrams to $\CR$-diagrams as a motivation to generalize the results from~\citet{Kawabata1992,Yeung2002a,Yeung2019}.
We define $\CR$-independences by vanishing $\CR$-terms of degree two, similar to how the probabilistic independence of random variables is characterized by vanishing mutual information.
We then define $\CR$-mutual independences and $\CR$-dual total correlation.
We show that an $\CR$-mutual independence is characterized by a vanishing $\CR$-dual total correlation.
We then define $\CR$-Markov random fields by the global Markov property and fully characterize them in terms of the $\CR$-diagram:
The global Markov property holds if and only if regions corresponding to disconnected vertex sets in the graph vanish.

We then apply this theory to the case where $\CR$ is a function like entropy, Kullback-Leibler divergence, or cross-entropy that is applied to probability mass functions. 
We show that when applying $\CR$ to sets of probability distributions that form a Markov random field with respect to the same graph, then the underlying random variables form an $\CR$-Markov random field. 
In particular, when applying our $\CR$-diagram characterization of $\CR$-Markov random fields, this implies that regions in the $\CR$-diagram corresponding to disconnected vertex sets disappear.
We then apply this result to the case of two joint distributions that form a Markov chain with equal transition probabilities from one time-step to the next, and the specific case that $\CR$ is Kullback-Leibler divergence.
This leads to a degeneracy of the Kullback-Leibler diagram in which the Kullback-Leibler divergence progressively shrinks ``over time'' --- a diagrammatic visualization of a weak version of the second law of thermodynamics.
Finally, we study the loss function of diffusion models --- the evidence lower bound --- and show how its explicit decomposition can be derived using a Kullback-Leibler diagram over a Markov chain.

\subsection*{Notation}

For $i, k \in \N$, we set $\range{i}{k} \coloneqq \{i, i+1, \dots, k\}$ if $i \leq k$ and $\range{i}{k} = \emptyset$, else.
As a special case, we set $\start{k} \coloneqq \range{1}{k}$.
If $I$ is a set and $i \in I$, we often write $I \setminus i$ for $I \setminus \{i\}$.
For a set $\Sigma$, we denote by $2^{\Sigma}$ its power set, i.e., the set of its subsets.
If $W_i$ are sets indexed with $i \in I$, then $W_I$ denotes $\bigcup_{i \in I} W_i$.
If $X_i, i \in I$ are elements of a commutative monoid and $A \subseteq I$, then we set $X_A \coloneqq \prod_{a \in A} X_a$.
If $X_1, \dots, X_n$ are elements of a separoid that form a Markov chain, then we will write $X_1 \to X_2 \to \dots \to X_n$.
We will denote the trivial measurable space by $\Asterisk$, which contains precisely one element denoted $\ast \in \Asterisk$.
The Shannon entropy of a random variable $X$ is denoted by $I(X)$ or $I_1(X)$, deviating from the typical notation of $H(X)$; The aim is to emphasize more strongly how entropy is embedded in the collection of (higher-order) Shannon information functions like mutual information and interaction information.

\section{Background and Outline}\label{sec:tech_back}

In this section, we introduce important background in multivariate information theory and its abstract generalizations, precisely state Yeung's characterization of Markov random fields in terms of $I$-diagrams and our generalizations of those results, and outline the rest of the paper. 
The aim is for this section to be self-contained and to provide sufficient context to appreciate the general results that then follow.

In Section~\ref{sec:ent_mi_ii}, we introduce Shannon entropy, mutual information, and interaction information in the abstract setting from~\citet{Lang2022}, highlighting the structure of monoids acting on abelian groups.
In Section~\ref{sec:background}, we then state the generalized Hu theorem from~\citet{Lang2022} and explain how it gives rise to the well-known $I$-diagrams from~\citet{Yeung1991} when specializing to the case of Shannon entropy.
In Section~\ref{sec:graph_terminology} we introduce some graph terminology necessary in the theory of Markov random fields.
In Section~\ref{sec:yeungs_characterization}, we introduce Markov random fields in separoids and show Yeung's characterization of those in the probabilistic context in terms of the $I$-diagram~\citep{Yeung2002a}.
In Section~\ref{sec:outline}, we motivate and state our generalization of this characterization to $F$-diagrams, making use of the generalized Hu theorem from~\citet{Lang2022}, and then outline the rest of the work, which will contain the proofs and applications of this result, in Section~\ref{sec:outline_coming}.

\subsection{Entropy, Mutual Information, and Interaction Information}\label{sec:ent_mi_ii}

In this section, we introduce the well-known notions of entropy, mutual information, and interaction information of higher degrees in the precise formal framework from~\citet{Lang2022} that reveals the monoid structure that we make use of in our work.
Let $\samp$ be a countable discrete measurable space.
We do \emph{not} fix a probability mass function on $\samp$, which is useful for obtaining a monoid action later in Definition~\ref{def:averaged_conditioning}.
When we speak of \emph{random variables}, then we mean functions $X: \samp \to \vs{X}$ with a \emph{finite and discrete} value space $\vs{X}$.

We write the space of probability mass functions $P: \samp \to [0, 1]$ as $\Sim{\samp}$.
We equip $\Delta(\samp)$ with the smallest $\sigma$-algebra that makes all evaluation maps
\begin{equation*}
  \text{ev}_x: \Delta(\samp) \to \R, \ \ \ P \mapsto \text{ev}_x(P) \coloneqq P(x)
\end{equation*}
for all elements $x \in \samp$ measurable.
In the finite case, this equals the Borel $\sigma$-algebra under the embedding $\Delta(\samp) \subseteq \R^{\samp} \cong \R^{|\samp|}$.
This measurable structure is standard in the context of the Giry monad~\citep{Giry1982}.
Measurability is not strictly necessary in our discrete setting, but would become important if one were to generalize our results to a non-discrete domain.

For a probability mass function $P \in \Sim{\samp}$, we write the distributional law with respect to a random variable $X: \samp \to \vs{X}$ as
\begin{equation*}
  \law{X}{P} \in \Delta(\vs{X}), \quad \law{X}{P}(x) \coloneqq P\big( X^{-1}(x)\big) = \sum_{\omega \in X^{-1}(x)} P(\omega).
\end{equation*}
This is also a probability mass function.
Furthermore, for $x \in \vs{X}$ with $P_X(x) \neq 0$, we define the conditional probability mass function $\cond{P}{X = x} \in \Sim{\samp}$ by
\begin{equation*}
  \cond{P}{X=x} \in \Delta(\samp), \quad \cond{P}{X=x}(\omega) \coloneqq
  \frac{P\big( \{\omega\} \cap X^{-1}(x)\big)}{\law{X}{P}(x)}.
\end{equation*}

If $X: \samp \to \vs{X}$ and $Y: \samp \to \vs{Y}$ are two random variables, then their joint variable is given by
\begin{equation*}
  XY: \samp \to \vs{X} \times \vs{Y}, \quad \omega \mapsto \big(X(\omega), Y(\omega)\big).
\end{equation*}
If the random variable is clear from the context, we write $P(x)$ for $\law{X}{P}(x)$.
Similarly, we may write $P(x \mid y)$ for $\law{X}{\big(\cond{P}{Y=y}\big)}(x)$
and $P(x, y)$ for $\law{XY}{P}(x, y)$.

We want to impose the structure of a monoid on collections of random variables.
Recall that a commutative, idempotent monoid $\mon{M} = (\mon{M}, \cdot, \one)$ consists of a set $\mon{M}$ together with a multiplication rule $\cdot: \mon{M} \times \mon{M} \to \mon{M}$ that is associative, commutative, has $\one$ as its neutral element, and is idempotent: $X \cdot X = X$ for all $X \in \mon{M}$.
For $X, Y \in \mon{M}$, write $X \precsim Y$ if $X \cdot Y = Y$.
With this definition, $(\mon{M}, \precsim)$ becomes a \emph{join-semilattice}, which is an equivalent description of a commutative, idempotent monoid.
Intuitively, it is often useful to think of $\precsim$ as the inclusion of sets, and of the product of elements in $\mon{M}$ as a union, which will be made precise in Theorem~\ref{thm:hu_kuo_ting_generalized}.
Join-semilattices were used in the development of the theory of separoids in~\cite{Dawid2001}, but we will not make use of that viewpoint and will throughout use the structure of monoids.
We also note that we do \emph{not} have a meet operator in our framework, neither in the context of random variables, nor in the generalization to monoids.

To impose the structure of a monoid on collections of random variables, we need to identify \emph{equivalent} random variables.
For two random variables $X, Y$ on $\samp$, we define $X \precsim Y$ if $X$ is a deterministic function of $Y$, meaning there exists a function $f_{XY}: \vs{Y} \to \vs{X}$ such that $X = f_{XY} \circ Y$, i.e., $X(\omega) = f_{XY}(Y(\omega))$ for all $\omega \in \samp$.
We write $X \sim Y$ if $X \precsim Y$ and $Y \precsim X$, which is an equivalence relation.

In the following, we will then write $X$ for both the random variable and its equivalence class, and we denote by $\Asterisk$ the trivial measurable space that contains precisely one element $\ast \in \Asterisk$.
One obtains the following:

\begin{proposition}
  [The Monoid of Random Variables]
   \label{pro:monoid_of_rvs}
   Equivalence classes of random variables, together with the multiplication given by the join operation $X \cdot Y \coloneqq XY$ and the neutral element given by $\one: \samp \to \Asterisk$, form a commutative, idempotent monoid.
\end{proposition}

\begin{proof}
  For a proof in this precise framework, see~\citet[Section 2.3]{Lang2022}.
  In the framework of lattices, this was formulated by~\citep{Dawid2001}. 
  We note that the collection of equivalence classes of random variables on $\samp$ is indeed a \emph{set} instead of a proper class, as an equivalence class $[X]$ can be identified with the partition $\big\lbrace X^{-1}(x) \mid x \in \vs{X} \big\rbrace$ on $\samp$. 
\end{proof}

Clearly, any subset of equivalence classes of random variables --- if it is closed under multiplication and contains a constant random variable --- then also forms a commutative, idempotent monoid.
It is useful to work with monoids of random variables since they have useful structure, and since the information functions we are concerned with do not depend on the representative of an equivalence class.
Recall that an abelian group $\gro{G} = (\gro{G}, +, 0)$ consists of a set $\gro{G}$ together with an addition rule $+: \gro{G} \times \gro{G} \to \gro{G}$ that is associative and commutative, has $0$ as its neutral element, and has ``negative'' elements: $g + (-g) = (-g) + g = 0$.
Now, define $\Meas(\Delta(\samp), \R)$ as the abelian group of measurable functions from $\Delta(\samp)$ to $\R$, where we define $(F + F')(P) \coloneqq F(P) + F'(P)$.

\begin{definition}[Shannon Entropy]\label{def:shannon_entropy} 
  Let $\ln: (0, \infty) \to \R $ be the natural logarithm 
  and $X: \Omega \to \vs{X}$ be a random variable. 
  The \emph{Shannon entropy of $X$ with respect to $P \in \Delta(\Omega)$} is given by
  \begin{equation*}
    I(X; P) \coloneqq I(P_X) = - \sum_{x \in E_X} P_X(x) \ln P_X(x) \in \R.
  \end{equation*}
  The \emph{entropy function} or \emph{Shannon entropy} of $X$ is the measurable function
  \begin{equation*}
    I(X) \in \Meas(\Delta(\Omega), \R), \ \ \ \big[I(X)\big](P) \coloneqq I(X; P)
  \end{equation*}
  defined on probability mass functions. 
  This function does not depend on the representative of the equivalence class $X$.
\end{definition}

\begin{definition}[Averaged Conditioning]\label{def:averaged_conditioning}
  Let $F \in \Meas(\Delta(\samp), \R)$.
  For a random variable $X: \samp \to \vs{X}$, define the \emph{averaged conditioning} of $F$ by $X$ as
  \begin{equation}
    (X.F)(P) \coloneqq \sum_{x \in \vs{X}} P_X(x) F(P|_{X = x}).
    \label{eq:sum_converges_unconditionally}
  \end{equation}
  Then $X.F \in \Meas(\Delta(\samp), \R)$.
  This definition does not depend on the representative of the equivalence class of $X$.
\end{definition}

Let $M$ be a monoid and $G$ an abelian group.
Then an additive monoid action (or monoid action for short) is a function $.: \mon{M} \times \gro{G} \to \gro{G}$ for which $\one \in \mon{M}$ acts trivially ($\one.g = g$), which is associative ($X.(Y.g) = (X \cdot Y).g$), and which is additive ($X.(g + h) = X.g + X.h$) ---  which also implies $X.0 = 0$.
Additive monoid actions generalize the conditioning operation of information functions on random variables, as the following proposition, whose proof we leave to the readers, shows:

\begin{proposition}
  \label{pro:monoid_action}
  Let $M$ be a monoid of equivalence classes of random variables and $G = \Meas(\Delta(\samp), \R)$.
  The averaged conditioning $. \colon M \times G \to G$ from Definition~\ref{def:averaged_conditioning} is a monoid action.
\end{proposition}

This viewpoint of the averaged conditioning was perhaps first studied in the context of information cohomology~\citep{Baudot2015a,Vigneaux2019a,Baudot2019}.
The proof of the following well-known chain rule of Shannon entropy is also left to the reader to prove:

\begin{proposition}[Chain Rule]
\label{cor:cocycle_condition_entropy}
  The following chain rule
  \begin{equation*}
    I(XY) = I(X) + X.I(Y) 
  \end{equation*}
  holds for arbitrary random variables $X: \Omega \to \vs{X}$ and $Y: \Omega \to \vs{Y}$.
\end{proposition}

We remark that $X.I(Y)$ is typically written as $I(Y \mid X)$ in other literature on information theory.
In our context, the notation $X.I(Y)$ is more natural, as it emphasizes that it comes from a monoid action.
For the following definition, set $I_1 \coloneqq I$, which then embeds Shannon entropy into the set of higher-order information functions:

\begin{definition}[Mutual Information, Interaction Information]\label{def:interaction_information}
    Let $q \in \N$ and assume that $I_{q-1}$ is already defined.
    Assume also that $Y_1, \dots, Y_q$ are $q$ random variables on $\Omega$.
    Then we define $I_q(Y_1;\dots;Y_q) \in \Meas(\Delta(\samp), \R)$, the \emph{interaction information of degree $q$}, as the function
    \begin{equation*}
      I_q(Y_1; \dots ; Y_q ) \coloneqq I_{q-1}(Y_1; \dots; Y_{q-1}) - Y_q.I_{q-1}(Y_1; \dots ; Y_{q-1}).
    \end{equation*}
    $I_2$ is also called mutual information.
    This definition does not depend on the representatives of the equivalence classes of $Y_1, \dots, Y_q$.
    For a specific probability mass function $P \in \Delta(\samp)$, we set $I_q(Y_1; \dots ; Y_q; P) \coloneqq \big[ I_q(Y_1; \dots ; Y_q) \big](P) \in \R$, and, when there is another random variable $X$ on $\samp$ on which we condition,
    $X.I_q(Y_1; \dots ; Y_q; P) \coloneqq \big[ X.I_q(Y_1; \dots Y_q) \big](P)$.
\end{definition}

  We now summarize the abstract properties of interaction information $I_q$. 
  Let $M$ be a commutative, idempotent monoid of (equivalence classes of) random variables as in Proposition~\ref{pro:monoid_of_rvs}.
  By abuse of notation, we do not distinguish between random variables and their equivalence classes, i.e., we write $Y$ instead of $\ec{Y}$.
  Denote by $G \coloneqq \Meas\big(\Delta(\Omega), \R\big)$ the abelian group of measurable functions from $\Delta(\Omega)$ to $\R$.
  Averaged conditioning $.: M \times G \to G$ is a well-defined monoid action.

  We can view $I_q$ as a function
  $I_q: M^q \to G$ that is defined on tuples of \emph{equivalence classes} of random variables.
  By Proposition~\ref{cor:cocycle_condition_entropy}, entropy $I_1$ satisfies the equation
  \begin{equation*}
    I_1(XY) = I_1(X) + X.I_1(Y)
  \end{equation*}
  for all $X, Y \in M$, where $X.I_1(Y)$ is the result of the action of $X \in M$ on $I_1(Y) \in G$ via averaged conditioning.
  Finally, by Definition~\ref{def:interaction_information}, for all $q \geq 2$ and all $Y_1, \dots, Y_q \in M$, one has
  \begin{equation*}
    I_q(Y_1; \dots ; Y_q) = I_{q-1}(Y_1; \dots ; Y_{q-1}) - Y_q.I_{q-1}(Y_1; \dots ; Y_{q-1}).
  \end{equation*}

\subsection{The Generalized Hu Theorem and \texorpdfstring{$\CR$}--Diagrams}\label{sec:background}

We now work towards a presentation of the generalized Hu theorem from~\cite{Lang2022}, which generalizes~\cite{Hu1962} and the $\Ent$-diagrams from~\cite{Yeung1991}.
These diagrams show in one overview the (higher-order) information functions of a set of variables and how they additively compose each other.
Fix an abelian group $\gro{G}$ (generalizing $\Meas(\Delta(\samp), \R)$ from above) and a commutative, idempotent monoid $\mon{M}$ (generalizing a monoid of equivalence classes of random variables). 
We also fix an additive monoid action $. \colon M \times G \to G$, generalizing the averaged conditioning.
For a set $\Sigma$, denote by $2^{\Sigma}$ its power set, i.e., the set of its subsets.

\begin{definition}[($\gro{G}$-Valued) Measure]\label{def:group-valued_measure_new}
  Let $\gro{G}$ be an abelian group and $\Sigma$ a set.\footnote{We only make use of the case $\Sigma = \set{X}$ as defined below in Equation~\eqref{eq:Sigma_definition}.}
  A \emph{$\gro{G}$-valued measure} is a function $\meas: 2^\Sigma \to \gro{G}$ with the property 
  \begin{equation*}
    \meas(A_1 \cup A_2) = \meas(A_1) + \meas(A_2)
  \end{equation*}
  for all disjoint $A_1, A_2 \subseteq \Sigma$.
  One automatically obtains $\meas(\emptyset) = 0$, and $\meas$ turns arbitrary finite disjoint unions into the corresponding finite sums.
\end{definition}

We now fix elements $X_1, \dots, X_n \in \mon{M}$, $n \geq 0$.
These are the elements for which we want to obtain an $\CR$-diagram later on.
Since $\mon{M}$ is commutative, every product of these elements (of arbitrary order and multiplicity) can be reordered such that all $X_i$ with the same index $i$ are next to each other.
Then, since $\mon{M}$ is idempotent, we can reduce the product further until each $X_i$ appears maximally once. 
This means that general products of the $X_i$ are of the form 
\begin{equation}
  \label{eq:general_product_reduction}
  X_I \coloneqq \prod_{i \in I}X_i \coloneqq X_{i_1}X_{i_2} \cdots X_{i_q}
\end{equation}
for some possibly empty subset $I = \{i_1 < i_2 < \dots < i_q\} \subseteq \start{n} = \{1, \dots, n\}$.
Furthermore, we have $X_IX_J \coloneqq X_I \cdot X_J = X_{I \cup J}$.

\paragraph{Definition of $\set{X}$}
We set
\begin{equation*}
  \set{X} \coloneqq \set{X}(n) \coloneqq 2^{[n]} \setminus \{\emptyset\}.
\end{equation*}
For $\emptyset \neq I \subseteq \start{n}$, we will write $p_I \coloneqq I$.
We can then also write
\begin{equation}\label{eq:Sigma_definition}
  \set{X} = \Big\lbrace \atom{I} \ \big| \   \emptyset \neq I \subseteq \start{n} \Big\rbrace.
\end{equation}

\begin{figure}
  \centering
  \includegraphics[width=\textwidth]{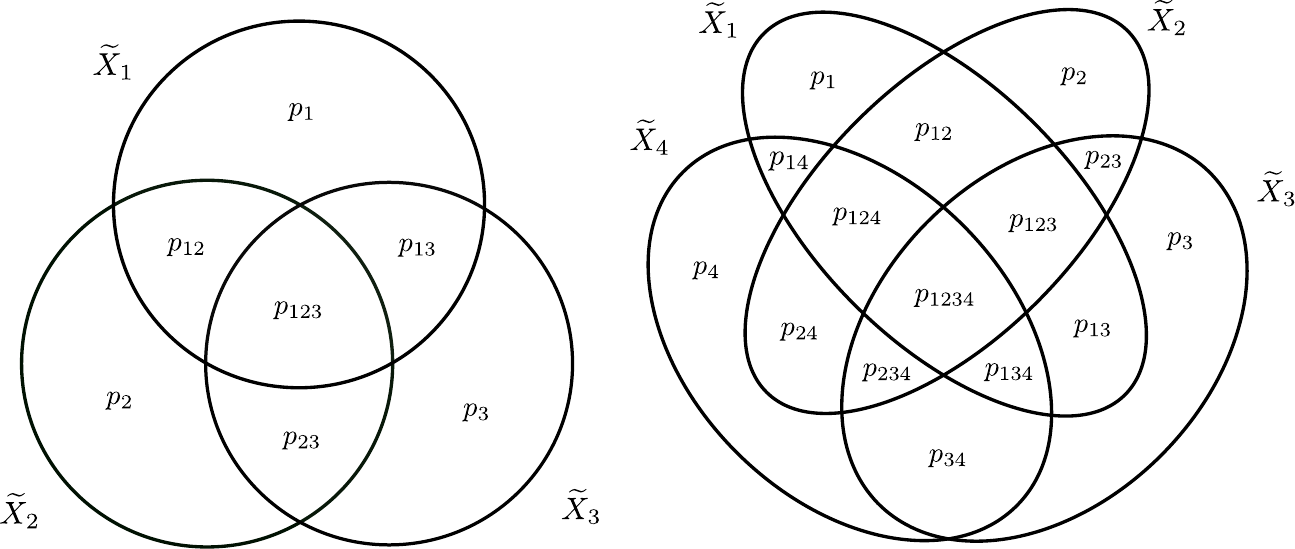}
  \caption{A depiction of $\set{X} = \set{X}(n)$ for $n = 3$ and $n = 4$, which will later be used to represent all the (higher-order) information functions.}
  \label{fig:reduce_circles}
\end{figure}

For $i \in \start{n}$, we denote by $\set{X}_i \coloneqq \big\lbrace \atom{I} \in \set{X} \ \ | \ \ i \in I \big\rbrace$  a set which we can imagine to be depicted by a disk corresponding to the element $X_i$.
We visualize this in Figure~\ref{fig:reduce_circles}.
This is actually the simplest construction that leads to the $\set{X}_i$ being in general position, meaning that for each choice of a nonempty set of disks indexed by $i \in I$, there is a single point inside all of them and not in any of the others:
\begin{equation}\label{eq:general_position}
  \bigcap_{i \in I} \set{X}_i \ \  \setminus \ \  \bigcup_{j \in \start{n} \setminus I}\set{X}_{j} = \{\atom{I}\}. 
\end{equation}
Consequently, we call $p_I$ also the \emph{atom} corresponding to $I$, as it is an atomic part of a diagram of intersecting disks.

With $\set{X}_I \coloneqq \bigcup_{i \in I} \set{X}_i$ we denote the union of the disks corresponding to the joint variable $X_I$.
Clearly, we have $\set{X} = \set{X}_{\start{n}}$.
In the following, we will also be flexible with our notation. 
For example, if we have a commutative, idempotent monoid $\mon{M}$ and fixed elements $X, Y, Z, W \in \mon{M}$, then we can also define 
\begin{equation*}
  \set{XYZW} = \set{X} \cup \set{Y} \cup \set{Z} \cup \set{W}.
\end{equation*}
The atom $\atom{XZ}$ would then be characterized by
\begin{equation*}
  \lbrace \atom{XZ} \rbrace = \Big(\set{X} \cap \set{Z}\Big) \setminus \Big(\set{Y} \cup \set{W}\Big).
\end{equation*}

The measure constructed in the proof of the following theorem is essentially constructed using a Möbius inversion formula:

\begin{theorem}[Generalized Hu Theorem,~\cite{Lang2022}]\label{thm:hu_kuo_ting_generalized}
  Let $\mon{M}$ be a commutative, idempotent monoid, $\gro{G}$ an abelian group, and $. : \mon{M} \times \gro{G} \to \gro{G}$ an additive monoid action.
 
  Assume $\CR_1: \mon{M} \to \gro{G}$ is a function that satisfies the following chain rule: for all $X, Y \in \mon{M}$, one has
    \begin{equation}\label{eq:cocycle_equationn}
      \CR_1(XY) =  \CR_1(X) +  X.\CR_1(Y)  .
   \end{equation}
Construct $\CR_q: \mon{M}^q \to \gro{G}$ for $q \geq 2$ inductively by
\begin{equation}\label{eq:inductive_definition}
  \CR_q(Y_1; \dots ; Y_q) \coloneqq \CR_{q-1}(Y_1; \dots ; Y_{q-1}) - Y_q.\CR_{q-1}(Y_1; \dots ; Y_{q-1})
\end{equation}
for all $Y_1, \dots, Y_q \in \mon{M}$.

Fix elements $X_1, \dots, X_n \in \mon{M}$, $n \geq 0$.
Set $\set{X} = \set{X}(n)$ as in Equation~\eqref{eq:Sigma_definition}.
Then there exists a $\gro{G}$-valued measure $\set{\CR}: 2^{\set{X}} \to \gro{G}$ such that for all $q \geq 1$ and $J, L_1, \dots, L_q \subseteq \start{n}$, the following identity holds:
\begin{equation}\label{eq:hu_kuo_ting_equation}
  X_J.\CR_q\big(X_{L_1}; \dots ; X_{L_q}\big) = \set{\CR}\Bigg(\bigcap_{k = 1}^q \set{X}_{L_k} \setminus \set{X}_J\Bigg).
\end{equation}
\end{theorem}

\begin{figure}
  \centering
  \includegraphics[width=\textwidth]{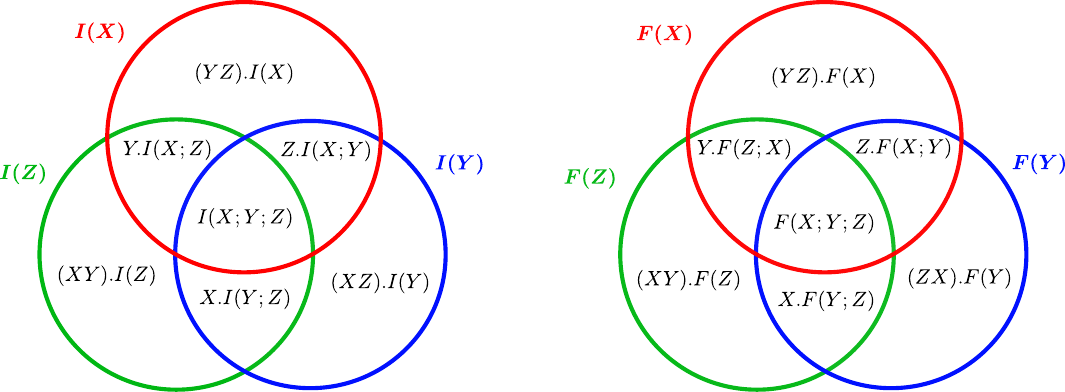}
  \caption{A depiction of the $\Ent$-diagram and the $\CR$-diagram from the (generalized) Hu theorem.
    On the left, it shows the interplay of (conditional) Shannon entropy, mutual information, and interaction information for three random variables $X, Y, Z$.
  On the right, $X, Y$ and $Z$ are elements of a commutative, idempotent monoid, generalizing the collection of equivalence classes of random variables together with their joint operation, and the higher-order terms are all derived from a function $\CR$ satisfying the chain rule $\CR(XY) = \CR(X) + X.\CR(Y)$.}
  \label{fig:venn_diagram_comparison}
\end{figure}

Note that in the preceding theorem, the $\gro{G}$-valued measure $\set{\CR}$ depends on the fixed elements $X_1, \dots, X_n$. 
For ease of notation, we write every $\CR_q$ simply as $\CR$.
Thus, in an expression of the form $\CR(X; Y; Z)$,
$\CR$ is necessarily $\CR_3$.
Since Shannon entropy $I$ and its higher order generalizations satisfy all the assumptions from Theorem~\ref{thm:hu_kuo_ting_generalized}, as explained in Section~\ref{sec:ent_mi_ii}, we recover the $I$-diagrams from~\citet{Yeung1991} as a special case.
We refer back to Figure~\ref{fig:venn_diagram_comparison} for a depiction of the $I$-diagram or $\CR$-diagram that results from the (generalized) Hu theorem. 
For figures that explain how to use Hu's theorem to visualize additive identities involving $\CR$, we refer to~\cite{Lang2022}.
As an example, we note that Figure~\ref{fig:venn_diagram_comparison} shows that
\begin{equation*}
  I(X) = (YZ).I(X) + Y.I(X;Z) + Z.I(X;Y) + I(X; Y; Z),
\end{equation*}
as we can see by observing how the disk corresponding to $I(X)$ decomposes. 
Formally, such identities follow from the fact that $\set{F}$ (or $\set{I}$ in the special case of Shannon's information functions) is a \emph{measure}, which means it is additive over disjoint unions in $\set{X}$.

\subsection{Graph Terminology}\label{sec:graph_terminology}

Before introducing Markov random fields, we need to discuss some necessary graph terminology.

\begin{definition}
  [Graph]
  \label{def:graph}
  A \emph{graph} is a tuple $\Gr{G} = (\Ver{V}, \Ed{E})$, where
  \begin{itemize}
    \item $\Ver{V}$ is a finite set, called the set of \emph{vertices};
    \item $\Ed{E} \subseteq \big\lbrace \{i, j\} \subseteq \Ver{V} \mid i \neq j \big\rbrace$, which we call the set of \emph{edges}.
  \end{itemize}
\end{definition}

Let $\Gr{G} = (\Ver{V}, \Ed{E})$ be a graph.
Instead of $\{i, j\} \in \Ed{E}$, we also write $i \edge j$, or $i \edge j \in \Ed{E}$.
Note that edges are non-oriented, as indicated by them being sets instead of tuples.
Furthermore, we have $i \neq j$ for every edge $i \edge j$, which means that there are no \emph{loops} (i.e., edges from a vertex to itself) in the graph.
However, there can be \emph{cycles}, i.e., walks from a vertex to itself passing through other vertices.

For a subset $\Ver{U} \subseteq \Ver{V}$, we can define $\Grdif{\Gr{G}}{\Ver{U}}$ as the graph obtained from $\Gr{G}$ by removing all the vertices in $\Ver{U}$ and all the edges that have at least one endpoint in $\Ver{U}$.
Formally, one can write
\begin{equation*}
  \Grdif{\Gr{G}}{\Ver{U}} \coloneqq (\Ver{V} \setminus \Ver{U}, \Grdif{\Ed{E}}{\Ver{U}}), \quad \text{where} \quad 
  \Grdif{\Ed{E}}{\Ver{U}} \coloneqq \big\lbrace \{i, j\} \in \Ed{E} \mid i, j \in \Ver{V} \setminus \Ver{U}\big\rbrace.
\end{equation*}

A \emph{walk} in $\Gr{G}$ from $i$ to $j$ is a sequence of edges $i \edge i_1, i_1 \edge i_2, \dots, i_n \edge j$ in $\Ed{E}$.
We write such a walk as $i \edge i_1 \edge \dots \edge i_n \edge j$.
The empty sequence is considered to be a walk from a vertex $i$ to itself, and walks are allowed to have the same vertices twice.
A vertex set $\Ver{U} \subseteq \Ver{V}$ is called \emph{connected in $\Gr{G}$} if for all $i \neq j \in \Ver{U}$ there is a walk in $\Gr{G}$ from $i$ to $j$.\footnote{Note that in Figure~\ref{fig:three_variables_mrfs}, we will also have a notion of connectedness. That notion, however, is about connectedness in $\Grdif{\Gr{G}}{\big( \Ver{V} \setminus \Ver{U}\big)}$, see Definition~\ref{def:type_II_atoms}.}
A set $\Ver{U} \subseteq \Ver{V}$ is called a \emph{component} if it is a maximal connected set, that is:
\begin{itemize}
  \item $\Ver{U}$ is connected in $\Gr{G}$;
  \item if $\Ver{U} \subseteq \Ver{W}$ and $\Ver{W}$ is connected in $\Gr{G}$, then $\Ver{U} = \Ver{W}$.
\end{itemize}
It is easy to see that the set of components of $\Gr{G}$ forms a partition of $\Ver{V}$:
formally, the components are the equivalence classes under the equivalence relation $\sim$ on $\Ver{V}$ that is defined by $i \sim j$ if there is a walk from $i$ to $j$ in $\Gr{G}$.
For a vertex set $\Ver{U} \subseteq \Ver{V}$, we define $\ncomp{\Ver{U}}$ as the number of components in the graph $\Grdif{\Gr{G}}{\Ver{U}}$.
In general, we denote by $\Ver{V}_1(\Ver{U}), \dots, \Ver{V}_{\ncomp{U}}(\Ver{U})$ these components.
We call $\Ver{U}$ a \emph{cutset} if $\ncomp{\Ver{U}} > 1$.
Note that cutsets do not necessarily need to ``cut'' anything: 
if $\Gr{G}$ already contains several components, then $\Ver{U} = \emptyset$ is a cutset.
Finally, for disjoint subsets $\Ver{A}, \Ver{B}, \Ver{C} \subseteq \Ver{V}$, we say that $\Ver{C}$ \emph{separates} $\Ver{A}$ from $\Ver{B}$ if \emph{every} walk from \emph{any} vertex in $\Ver{A}$ to \emph{any} vertex in $\Ver{B}$ passes through \emph{some} vertex in $\Ver{C}$.
Note: if $\Ver{A} \neq \emptyset \neq \Ver{B}$ in this definition, then $\Ver{C}$ is automatically a cutset.

We now define the notions of connected and disconnected atoms with respect to a graph $\Gr{G}$,
which builds a bridge between the graph and information diagrams.
They are identical to the notions of type I and type II atoms defined in~\cite{Yeung2002a}.

\begin{definition}
  [Connected and Disconnected Atoms]
  \label{def:type_II_atoms}
  Let $n \in \N$ and $\Gr{G}$ a graph with vertex set $\mathcal{V} = [n]$, and let $\set{X} = \set{X}(n)$ be defined as in the previous subsection.
  For $\Ver{W} \subseteq \start{n}$, the corresponding atom $\atom{\Ver{W}} \in \set{X}$ is called
  \emph{disconnected} with respect to $\Gr{G}$ if $\Ver{V} \setminus \Ver{W} = \start{n} \setminus \Ver{W}$ is a cutset, i.e., if $\Ver{W}$ is disconnected as a vertex set in $\Grdif{\Gr{G}}{\big(\start{n} \setminus \Ver{W}\big)}$.
  Otherwise, $\atom{\Ver{W}}$ is called \emph{connected} (with respect to $\Gr{G}$).
\end{definition}

\subsection{Yeung's Characterization of Markov Random Fields via \texorpdfstring{$I$}{I}-diagrams}\label{sec:yeungs_characterization}

Markov random fields encode independence relations of random variables that correspond to separations in a corresponding graph.
As such, we first need to define:

\begin{definition}[Probabilistic Conditional Independence]
   \label{def:probabilistic_independence}
 Let $X, Y, Z$ be random variables on $\samp$ and $P \in \Delta(\samp)$ a probability mass function.
   Then $X$ and $Y$ are said to be \emph{$P$-independent} given $Z$, written
   \begin{equation*}
     \IndepF{P}{X}{Y}{Z},
   \end{equation*}
   if for all $(x, y, z) \in \vs{X} \times \vs{Y} \times \vs{Z}$, the joint distribution decomposes as follows:\footnote{If $P(z) = 0$, then $P(x \mid z)$ is not defined. But then $P(x, y, z) = 0$ and $P(y, z) = 0$, so we can simply define the right-hand-side of the equation as zero.}
   \begin{equation}\label{eq:independence_relation}
      P(x, y, z) = P(x \mid z) \cdot P(y, z).
   \end{equation}
   Equivalently, for all $z \in \vs{Z}$ with $P(z) \neq 0$ one has
   \begin{equation*}
     P(x, y \mid z) = P(x \mid z) \cdot P(y \mid z).
   \end{equation*}
   Equivalently, for all $y, z \in \vs{Y} \times \vs{Z}$ with $P(y, z) \neq 0$ one has
   \begin{equation*}
      P(x \mid y, z) = P(x \mid z).
   \end{equation*}
\end{definition}

If $X' \sim X$, $Y' \sim Y$, $Z' \sim Z$, and $P \in \Sim{\samp}$, then
\begin{equation*}
  \IndepF{P}{X}{Y}{Z} \quad \Longleftrightarrow \quad \IndepF{P}{X'}{Y'}{Z'},
\end{equation*}
see, for example,~\cite{Dawid2001}, Section 6.2.
Thus, also in this context, it is suitable to identify a random variable with its equivalence class. 
This definition gives rise to a separoid, which is a useful abstract structure that applies to all settings that we study in this paper. 
We use the slightly adapted but equivalent version of the separoid axioms from~\cite{Forre2021a}, Appendix A.4:

\begin{definition}
  [Separoid, Separoid Axioms]
  \label{def:separoid}
  Let $\mon{M}$ be a commutative, idempotent monoid. 
  Let $\indep$ be a relation on $\mon{M}^2 \times \mon{M}$, written for $X, Y, Z \in \mon{M}$ by
  \begin{equation*}
    \Indep{X}{Y}{Z}.
  \end{equation*}
  The tuple $(\mon{M}, \indep)$ is called a \emph{separoid} if $\indep$ satisfies the following \emph{separoid axioms} for all $X, Y, Z, W \in \mon{M}$:
  \begin{align*}
    & \mathrm{(S1) \ symmetry}: \Indep{X}{Y}{Z} \ \ \Longrightarrow \ \ \Indep{Y}{X}{Z}; \\
    & \mathrm{(S2) \ redundancy}: W \precsim Z \ \ \Longrightarrow \ \ \Indep{W}{Y}{Z};  \\
    & \mathrm{(S3) \ decomposition}: \Indep{WX}{Y}{Z} \ \ \Longrightarrow \ \ \Indep{X}{Y}{Z}; \\
    & \mathrm{(S4) \ weak \ union}: \Indep{WX}{Y}{Z} \ \ \Longrightarrow \ \ \Indep{W}{Y}{XZ}; \\
    & \mathrm{(S5) \ contraction}: \big(\Indep{W}{Y}{XZ} \mathrm{\ and\ } \Indep{X}{Y}{Z}\big) \ \ \Longrightarrow \ \ \Indep{WX}{Y}{Z}.
  \end{align*}
\end{definition}

For a relationship to the more familiar notion of a (semi-)graphoid, see~\cite[Section 3.2]{Dawid2001}.
In this paper, if a relation $\Indep{X}{Y}{Z}$ holds, then we also say ``$X$ and $Y$ are independent given $Z$'', or ``$X$ and $Y$ are independent conditioned on $Z$'', and often append ``with respect to $\indep$'' to clarify the independence relation.

Note that properties (S2)--(S5) have obvious ``right-handed versions'' as well due to symmetry (S1).
When we refer to one of the properties (S2)--(S5), then we mean one of the two versions depending on the context.
We obtain:

\begin{proposition}
  \label{pro:forms_a_separoid}
  Let $\mon{M}$ be a monoid of random variables as in Proposition~\ref{pro:monoid_of_rvs} and $P \in \Sim{\samp}$ a probability mass function.
  Then $(\mon{M}, \indep_P)$ is a separoid.
\end{proposition}

\begin{proof}
  This is well-known and appeared originally in this formulation in~\cite{Dawid2001}.
  It also follows from the generalization given in~\cite{Forre2021a}, Theorem 3.11.
\end{proof}

We are now ready to define Markov random fields via the global Markov property, in general separoids and also in the specific probabilistic case:

\begin{definition}
  [Global Markov Property, Markov Random Field]
  \label{def:global_markov_property}
  Let $\Gr{G} = (\Ver{V}, \Ed{E})$ be a graph with $\Ver{V} = \start{n}$ and $(M, \indep)$ a separoid.
  A sequence of elements $X_1, \dots, X_n \in \mon{M}$ is called a \emph{Markov random field} with respect to $\indep$ and $\Gr{G}$ if it satisfies the \emph{global Markov property}, that is:
  for all disjoint $\Ver{A}, \Ver{B}, \Ver{C} \subseteq \Ver{V}$ such that $\Ver{C}$ separates $\Ver{A}$ from $\Ver{B}$, we have
  \begin{equation*}
    \Indep{X_{\Ver{A}}}{X_{\Ver{B}}}{X_{\Ver{C}}}.\footnotemark
  \end{equation*}
  \footnotetext{We could have also defined the separation relation $A \perp B \ | \ C$ on $\Gr{G}$, defined for all vertex sets $A, B, C$ in $\Gr{G}$, and not only those that are disjoint as in our terminology. This relation would satisfy the separoid axioms. The global Markov property is then equivalent to the statement that $A \perp B \ | \ C$ implies $\Indep{X_{\Ver{A}}}{X_{\Ver{B}}}{X_{\Ver{C}}}$. This, in turn, can be phrased as the statement that $A \mapsto X_A$ constitutes a \emph{separoid homomorphism} as defined in~\citet[Definition 1.3]{Dawid2001}.}
  
  Finally, in the probabilistic context, if random variables $X_1, \dots, X_n$ on $\samp$ form a Markov random field with respect to $\Gr{G}$ and $\indep_P$, then we say that $X_1, \dots, X_n$ form a $P$-Markov random field with respect to $\Gr{G}$.
\end{definition}

\begin{figure}
  \centering
  \includegraphics[width=\textwidth]{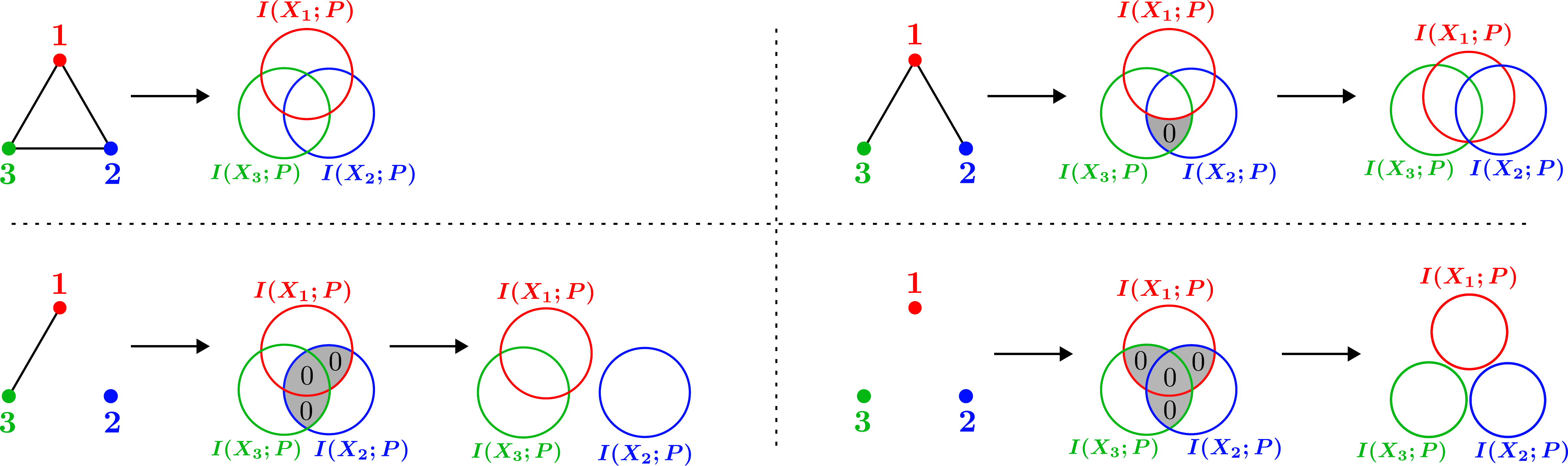}
  \caption{We show the effect of the graph structure of simple Markov random fields on the corresponding $I$-diagrams for a fixed probability mass function $P$.
    The $\Ent$-diagram visualizes the relationships of the entropy, mutual information, and interaction information of the variables.
    The figure is based on Yeung's result Theorem~\ref{thm:charac_proba_mrfs}, which shows that a set of random variables forms a Markov random field corresponding to a graph if and only if all atoms in the $\Ent$-diagram corresponding to disconnected sets of vertices in the graph disappear.
    \newline  Concretely, to a set of vertices $J$, the corresponding atom in the $\Ent$-diagram is the intersection of all disks with indices $j \in J$, without any element in the union of all the other disks.
  For the lower left panel, the three sets of vertices $\{1, 2\}$, $\{2, 3\}$, and $\{1, 2, 3\}$ are disconnected, giving rise to three disappearing atoms in the $\Ent$-diagram. 
    Consequently, $\Ent(X_2)$ can be drawn to not intersect with the other disks.
  However, the other four sets of vertices $\{1\}, \{2\}, \{3\}, \{1, 3\}$ are clearly connected, which means that we cannot infer their corresponding atoms to vanish.
  Similar reasoning applies to the other three panels. 
}
  \label{fig:three_variables_mrfs}
\end{figure}

The definition of $P$-Markov random fields is stated in terms of conditional $P$-independences.
Now, one crucial observation is the well-known fact that such independences can be characterized using conditional mutual information:
\begin{equation}\label{eq:equivalence}
  \IndepF{P}{X}{Y}{Z} \quad \Longleftrightarrow \quad Z.I(X; Y; P) = 0.
\end{equation}
Yeung's insight was that this should make it possible to characterize $P$-Markov random fields by properties of their corresponding $I$-diagrams. 
Let $\mon{M}$ be the monoid of (equivalence classes of) random variables on $\samp$, and $\Ent: \mon{M} \to \Meas\big( \Sim{\samp}, \R \big)$ the Shannon entropy function, as defined in Definition~\ref{def:shannon_entropy}.
Let $X_1, \dots, X_n$ be fixed random variables on the sample space $\samp$.
Let $\set{\Ent}: 2^{\set{X}} \to \Meas\big( \Sim{\samp}, \R \big)$ be the measure resulting from Hu's Theorem~\ref{thm:hu_kuo_ting_generalized}.
For any probability mass function $P \in \Sim{\samp}$, we then obtain the ``slice''
\begin{equation}\label{eq:slice_function}
  \set{\Ent}^P: 2^{\set{X}} \to \R, \quad A \mapsto \big[\set{\Ent}(A)\big](P).
\end{equation}
This is a signed measure with values in the real numbers instead of functions, and it visualizes the interplay of the (higher-order) Shannon information functions \emph{for the specific probability mass function $P$}.

This can be imagined as a ``slice'' of the original $\Ent$-diagram: 
e.g., in the special case that $\samp = \{0, 1\}$, the mass functions $P \in \Sim{\samp}$ are essentially numbers in $[0, 1]$.
For each $P \in [0, 1]$, we then get an information diagram corresponding to $\set{I}^P$.
If we stack these ``on top of each other'', the full stack is the $\Ent$-diagram of $\set{I}$, with the ``slices'' given by the diagrams for $\set{I}^P$.
\citet{Yeung2002a} then showed:

\begin{theorem}
  [Characterization of $P$-Markov Random Fields]
  \label{thm:charac_proba_mrfs}
  Let $\Ent: \mon{M} \to \Meas\big( \Sim{\samp}, \R \big)$ be the Shannon entropy function.
  Let $X_1, \dots, X_n$ be random variables on $\samp$ and $P \in \Sim{\samp}$ a probability mass function, giving rise to $\set{\Ent}^P: 2^{\set{X}} \to \R$ by Equation~\eqref{eq:slice_function}.
  Let $\Gr{G}$ be a graph with vertex set $\start{n}$.
  Then the following two statements are equivalent:
  \begin{itemize}
    \item $X_1, \dots, X_n$ form a $P$-Markov random field with respect to $\Gr{G}$; 
    \item $\set{\Ent}^{P}(\atom{\Ver{W}}) = 0$ for all atoms $\atom{\Ver{W}}$ that are disconnected with respect to $\Gr{G}$.
  \end{itemize}
\end{theorem}

We visualize this result in Figure~\ref{fig:three_variables_mrfs}, which shows the intuitive relation between the graph of a Markov random field and vanishing regions of atoms that are ``disconnected'' according to the graph.
One can then specialize the theorem to Markov chains:

\begin{definition}
  [Markov Chain]
  \label{def:markov_chain}
  Let $n \geq 0$ and $(M, \indep)$ a separoid.
  A sequence of elements $X_1, \dots, X_n$ in $\mon{M}$ is called a \emph{Markov chain} (with respect to $\indep$)
  if for all $2 \leq i \leq n$, the following independence holds:
  \begin{equation*}
    \Indep{X_i}{X_{\start{i-2}}}{X_{i-1}}.
  \end{equation*}
  We also write $X_1 \to X_2 \to \dots \to X_n$ to indicate that the sequence $X_1, \dots, X_n$ forms a Markov chain.

  Finally, in the probabilistic context, a collection of random variables $X_1, \dots, X_n$ on $\samp$ are said to form a $P$-Markov chain if they form a Markov chain with respect to $\indep_P$.
\end{definition}

For $n = 0, 1$ or $2$, all sequences form a Markov chain.
The first interesting case is $n = 3$.
The only non-vacuous condition for a Markov chain is then $\Indep{X_3}{X_1}{X_2}$.
The following corollary of Theorem~\ref{thm:charac_proba_mrfs} was already shown in~\citet{Kawabata1992}:

\begin{corollary}
  \label{cor:Markov_Chain_Charac_proba}
  With all notation as above, the following two statements are equivalent:
  \begin{itemize}
    \item $X_1, \dots, X_n$ form a $P$-Markov chain.
    \item $\set{I}^P(\atom{\Ver{W}}) = 0$ for all $\Ver{W} \subseteq \start{n}$ that do \emph{not} only contain consecutive numbers.
  \end{itemize}
\end{corollary}

\begin{figure}
  \centering
  \begin{subfigure}[b]{0.43\textwidth}
    \centering
    \includegraphics[width=\textwidth]{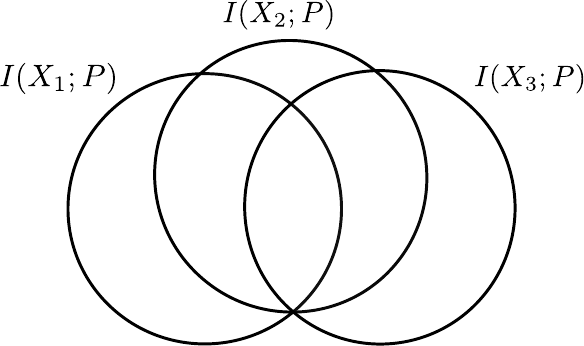}
  \end{subfigure}
  \hfill
  \begin{subfigure}[b]{0.47\textwidth}
    \centering
    \includegraphics[width=\textwidth]{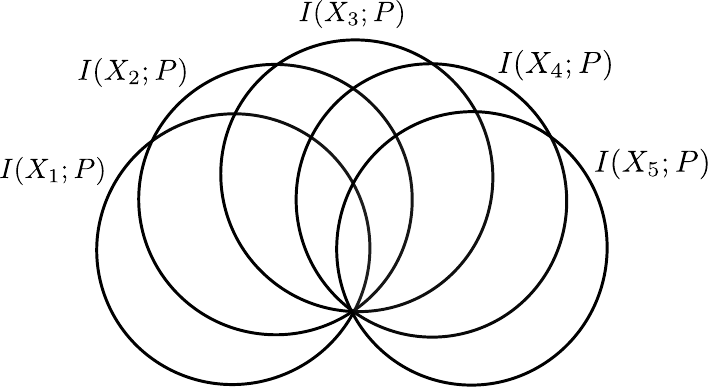}
  \end{subfigure}
  \caption{If random variables $X_1, \dots, X_n$ form a $P$-Markov chain, then many atoms in the $I$-diagram with respect to $P$ disappear by Corollary~\ref{cor:Markov_Chain_Charac_proba}. 
    The only atoms that remain are those corresponding to ``intervals'' in $\start{n}$. 
    This leads to a fan-like structure of the $I$-diagram with respect to $P$, as visualized here for $n = 3$ and $n = 5$.}
  \label{fig:four_and_five_circles_P}
\end{figure}

We visualize this result in Figure~\ref{fig:four_and_five_circles_P}.
As an example, this structure of the $I$-diagram can be used to prove the well-known data processing inequality:

\begin{corollary}[Data Processing Inequality]
  \label{cor:data_processing_in}
  Let $X_1 \to X_2 \to X_3$ be a $P$-Markov chain of random variables on $\samp$.
  Then $I(X_1; X_3; P) \leq I(X_1;X_2;P)$.
\end{corollary}

\begin{proof}
  Using Hu's theorem, Theorem~\ref{thm:hu_kuo_ting_generalized}, Equation~\eqref{eq:slice_function}, and the Venn diagram of the Markov chain from Figure~\ref{fig:four_and_five_circles_P}, we obtain:
  \begin{align*}
    I(X_1; X_2; P) &= \set{I}^P(\set{X}_1 \cap \set{X}_2)  \\
    &=  \set{I}^P(\set{X}_1 \cap \set{X}_3) + \set{I}^P(\set{X}_1 \cap \set{X}_2 \setminus \set{X}_3) \\
    &=  I(X_1; X_3; P) + X_3.I(X_1; X_2; P) \\
    & \geq I(X_1; X_3; P).
  \end{align*}
  In the last step, we used the well-known property of conditional mutual information to be non-negative.
\end{proof}

\subsection{An Outline of our Generalizations of Yeung's Results}\label{sec:outline}

We now explain how we aim to generalize Yeung's results on the $I$-diagram characterization of Markov random fields.
The motivation is that we would like to study implications of Markov random field structures on more general information diagrams.
For example, assume that $X_1, \dots, X_n$ are random variables on $\samp$ that form a Markov random field with respect to a graph $\Gr{G}$ and two different probability mass functions $P$ and $Q$.
Can we then say anything about how the Kullback-Leibler divergence $D(X_1 \cdots X_n; P \| Q)$ of $P$ and $Q$ over the whole joint variable $X_1 \cdots X_n$ distributes into components of different variables in this structure?
Such questions are commonplace for example in machine learning~\citep{Bishop2007}, where it is sometimes assumed that both the model- and data distribution follow a specific graphical form, and the loss function takes the form of a Kullback-Leibler divergence (or, equivalently up to a constant, cross-entropy) between the two.
We will demonstrate this type of application by showing how our theory allows to find a conceptually simple derivation of a decomposition of the loss function of diffusion models in Section~\ref{sec:diffusion_models}.

In this whole subsection, let $M$ be a commutative, idempotent monoid acting on an abelian group $G$, and $F: M \to G$ a function satisfying the chain rule, Equation~\eqref{eq:cocycle_equationn}.
To generalize Yeung's characterization, we use the setting of $F$-diagrams as in Hu's theorem, Theorem~\ref{thm:hu_kuo_ting_generalized}.
Motivated by Equation~\eqref{eq:equivalence}, which relates $P$-independence to mutual information, we generalize $P$-independence using $F$ itself.

\begin{definition}
  [$\CR$-independence]
  \label{def:F_independence}
  We define the relation $\indep_\CR$ on $\mon{M}^2 \times \mon{M}$ by
  \begin{equation*}
    \IndepF{\CR}{X}{Y}{Z} \ \ :\Longleftrightarrow \ \ Z.\CR(X; Y) = 0.
  \end{equation*}
  If $\IndepF{\CR}{X}{Y}{Z}$, then $X$ is called $\CR$-independent from $Y$ given $Z$.
\end{definition}

The definition is analogous to the conditional independence relation defined for general submodular information functions in~\cite{Steudel2010}.
In Proposition~\ref{pro:F-separoid} we will show that $(M, \indep_F)$ forms a separoid.
We can then specialize the notion of a Markov random field and Markov chains from general separoids to the separoid $(M, \indep_F)$:

\begin{terminology}
  [$\CR$-Markov Random Field, $\CR$-Markov Chain]
  \label{ter:mrf_mc_f}
  Elements $X_1, \dots, X_n \in M$ are said to form an $\CR$-Markov random field with respect to $\Gr{G}$ if they form a Markov random field with respect to $\indep_{\CR}$ and $\Gr{G}$, see Definition~\ref{def:global_markov_property}.

  Similarly, $X_1, \dots, X_n$ are said to form an $\CR$-Markov chain if they form a Markov chain with respect to $\indep_{\CR}$, see Definition~\ref{def:markov_chain}.
\end{terminology}

We will then prove the following theorem in Section~\ref{sec:markov_random_fields}:

\begin{theorem}
  [$\CR$-Markov Random Field Characterization]
  \label{thm:mrf_characterization}
  Let $\mon{M}$ be a commutative, idempotent monoid acting additively on an abelian group $\gro{G}$, and $\CR: \mon{M} \to \gro{G}$ a function satisfying the chain rule Equation~\eqref{eq:cocycle_equationn}.
  Additionally, fix elements $X_1, \dots, X_n$ giving rise to $\set{F}: 2^{\set{X}} \to \gro{G}$, and a graph $\Gr{G}$ with vertex set $\start{n}$.
  Then the following statements are equivalent:
  \begin{itemize}
    \item $X_1, \dots, X_n$ form an $\CR$-Markov random field with respect to $\Gr{G}$;
    \item $\set{\CR}(\atom{\Ver{W}}) = 0$ for all disconnected atoms $\atom{\Ver{W}} \in \set{X}$. 
  \end{itemize}
\end{theorem}

\subsection{An Outline of the Coming Sections}\label{sec:outline_coming}

In Section~\ref{sec:mrfs_in_separoids}, we take a step back from independences induced by a function $\CR$ by studying independences in separoids, Definition~\ref{def:separoid}, which generalizes both $P$-independence and $F$-independence.
We also define conditional mutual independences $\bigindep_{i = 1}^{n}{X_i} \ | \ Y$ and prove a simple characterization. 
We study Markov random fields in separoids and characterize them equivalently by the \emph{cutset property}; It states that if a vertex set $\Ver{U}$ is a ``cutset'' that, when removed, cuts the rest of the graph into components $\Ver{V}_1, \dots, \Ver{V}_q$, then this leads to a corresponding mutual independence $\bigindep_{i = 1}^{q} {X_{\Ver{V}_i}} \ | \ X_{\Ver{U}}$.
Finally, we characterize Markov chains in separoids.

In Section~\ref{sec:yeung2002}, we first study important consequences of Hu's theorem, Theorem~\ref{thm:hu_kuo_ting_generalized}, including the statement that $F$-independence gives rise to a separoid, Proposition~\ref{pro:F-separoid}.
One crucial property that is used in its proof is what we call \emph{subset determination}.
It states that when the value of a region in an $\CR$-diagram is known, this fully determines the value of all subregions.
In particular, if a region in an $\CR$-diagram vanishes, then all subsets vanish as well.
A precise formulation is given in Theorem~\ref{thm:subset_determination}.
Crucially, this replaces the repeated usage of inequalities in the proofs in~\citep{Yeung2002a}, which we could not make use of in our general context.
As an aside, in Appendix~\ref{sec:classification_CR}, we show that the basic idea of subset determination implies a characterization of all functions $\CR$ satisfying the chain rule, as long as $\mon{M}$ contains a ``top element'' $\top$ such that $X \cdot \top = \top$ for all $X \in \mon{M}$ --- which is the case for finitely generated $\mon{M}$.
The functions $\CR$ then correspond to elements in $\gro{G}$ that are \emph{annihilated} by $\top$, by the simple mapping $\CR \mapsto \CR(\top)$.
We also provide a cohomological interpretation of this result.

To avoid misunderstandings, we mention a subtlety with subset determination: it only applies to $\CR$-diagrams for functions $\CR: \mon{M} \to \gro{G}$ from a monoid to an abelian group \emph{on which the monoid additively acts}.
Formally, for information functions such as entropy or Kullback-Leibler divergence, it becomes invalid when \emph{fixing the underlying probability mass functions}.
For example, the joint entropy $I(XY; P)$ does not determine the mutual information $I(X;Y; P)$ even though the entire function $I(XY)$ determines the function $I(X;Y)$ via the monoid action.
Nevertheless, we will manage to apply our results to classical information functions in Section~\ref{sec:specializing_to_probabilistic} and can deduce Yeung's results for fixed probability mass functions, including Theorem~\ref{thm:charac_proba_mrfs}, as we demonstrate in Appendix~\ref{sec:slices}.

\begin{figure}
  \centering
  \includegraphics[width=\textwidth]{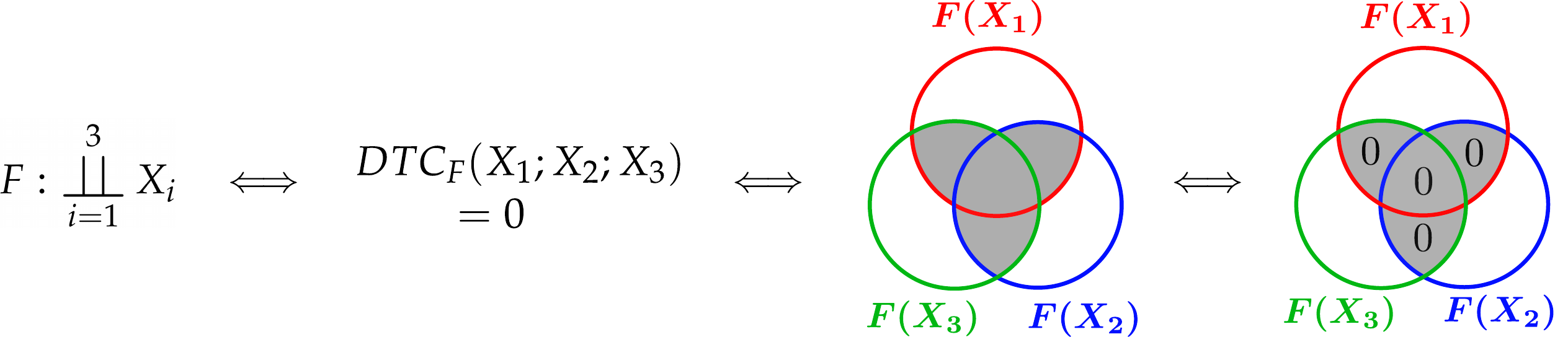}
  \caption{One key ingredient in the $\CR$-diagram characterization of $\CR$-Markov random fields is the characterization of $\CR$-mutual independences, Theorem~\ref{thm:charac_using_dual}, here visualized for three elements $X_1, X_2, X_3 \in \mon{M}$.
    The characterization shows that the mutual independence is equivalent to the vanishing of $\CR$-dual total correlation $\DTC{\CR}(X_1; X_2; X_3)$, which, by Hu's theorem, Theorem~\ref{thm:hu_kuo_ting_generalized}, corresponds to the vanishing of a region of four atoms in the $\CR$-diagram, visualized as gray.
    Subset determination, Theorem~\ref{thm:subset_determination}, then allows to conclude that every \emph{individual} atom in this region vanishes, as shown in the rightmost part of the figure.
  The implication from right to left again follows from Hu's theorem and the fact that $\CR$-diagrams visualize a \emph{measure}, meaning that larger regions are the sum of their atoms.}
  \label{fig:subset_determination_application}
\end{figure}

In the later parts of Section~\ref{sec:yeung2002}, we generalize the results from~\cite{Kawabata1992,Yeung2002a}.
We start by a simple characterization of pairwise $\CR$-independences via the $\CR$-diagram.
Motivated by this special case, and generalizing the classical case~\citep{Han1978}, we define $\CR$-dual total correlation by 
\begin{equation*}
  \DTC{\CR}(X_1; \dots; X_n) \coloneqq \CR(X_{[n]}) - \sum_{i = 1}^{n} X_{[n] \setminus i}.\CR(X_i),
\end{equation*}
where $X_{[n]} = X_1\cdots X_n$.
In Theorem~\ref{thm:charac_using_dual}, we then characterize conditional \emph{mutual} $\CR$-independences $\bigindepF{\CR}_{i = 1}^{n} X_i \ | \ Y$ by the vanishing of (conditional) $\CR$-dual total correlation $Y.\DTC{\CR}(X_1; \dots; X_n)$.
The theorem also uses Hu's theorem to characterize this by the vanishing of the corresponding atomic regions.
This crucially uses subset determination, Theorem~\ref{thm:subset_determination}, as visualized in Figure~\ref{fig:subset_determination_application}.
Notably, we also study $\CR$-total correlation $\TC{\CR}$, generalizing classical total correlation~\citep{Watanabe1960}, and find that it also provides a characterization of mutual $\CR$-independences, albeit only if the abelian group $\gro{G}$ is \emph{torsion-free}; The reason is that the total correlation double-counts atoms. 
This is studied in Appendix~\ref{sec:total_correlation_charac}.

We then prove Theorem~\ref{thm:mrf_characterization} in Section~\ref{sec:markov_random_fields}, which shows that $F$-Markov random fields (Terminology~\ref{ter:mrf_mc_f}) are fully characterized by the vanishing of atoms in the $\CR$-diagram that correspond to disconnected vertex sets in the underlying graph.
One ingredient of Theorem~\ref{thm:mrf_characterization} will be Theorem~\ref{thm:FCMIs_characterization}, the characterization of so-called full conditional mutual $\CR$-independences in terms of the $\CR$-diagram.
Finally, in Appendix~\ref{sec:generalizing_yeung_2019}, we explain which results from~\cite{Yeung2019} generalize to our setting, in particular by reproving Theorem~\ref{thm:charac_proba_mrfs} with our results.

In Section~\ref{sec:specializing_to_probabilistic}, we then come back to the case where the information functions are applied to probability mass functions. 
In particular, $\mon{M}$ is then a monoid of equivalence classes of random variables acting additively on the abelian group $\gro{G} = \Meas(\Delta(\samp)^{r+1}, \R)$ of measurable functions from tuples of $r+1$ probability mass functions over a finite sample space $\samp$ to the real numbers.
$\CR: \mon{M} \to \gro{G}$ is from a specific class of functions such as entropy (for $r = 0$), cross-entropy, or Kullback-Leibler divergence (for $r = 1$).
To give meaning to the resulting functions, we briefly explain how the higher-order cross-entropy and Kullback-Leibler divergence relate to the cluster cross-entropy from~\cite{Cocco2012}.
We then adapt $\gro{G}$ to only contain functions on \emph{subsets} of probability mass functions that are conditionally stable, which allows to restrict $\CR$ without losing the monoid action from $\mon{M}$ to $\gro{G}$. 
As we show, the property to form a $P$-Markov random field is conditionally stable, i.e., stable under conditioning $P$ on any evaluation $X = x$ of any random variable in the Markov random field, and so we can restrict information functions $\CR$ to this and similar properties.

In Theorem~\ref{thm:get_mrf_for_free_with_restriction}, we obtain an important result: When restricting $\CR$ to tuples of probability mass functions that satisfy a conditionally stable property that \emph{implies} the Markov random field property, then the underlying random variables formally form an $\CR$-Markov random field as in Section~\ref{sec:yeung2002}.\footnote{Importantly, there is no reverse of this statement. E.g., if two probability mass functions $P$ and $Q$ give rise to a Kullback-Leibler diagram with vanishing regions corresponding to the graph of a Markov random field, it does \emph{not} imply that $P$ and $Q$ factorize according to the graph. For example, whenever $P = Q$, the whole Kullback-Leibler diagram trivially vanishes, but we can't conclude anything nontrivial about $P$ or $Q$ from it.}
In particular, using Theorem~\ref{thm:mrf_characterization}, this leads to the vanishing of regions in the $\CR$-diagram that correspond to disconnected vertex sets.
We apply this to the case that $\CR$ is the Kullback-Leibler divergence restricted to tuples of joint probability mass functions on a Markov chain with equal transition probabilities from one time-step to the next.
In Theorem~\ref{thm:second_law} and Figure~\ref{fig:second_law}, we obtain a degeneracy of the Kullback-Leibler diagram that can be interpreted as a diagram-representation of a weak version of the second law of thermodynamics: The Kullback-Leibler divergence progressively ``shrinks over time''.
Additionally, we apply Theorem~\ref{thm:get_mrf_for_free_with_restriction} to obtain a simple explicit derivation of the evidence lower bound in diffusion models by decomposing the Kullback-Leibler divergence over a Markov chain.

Finally, in Section~\ref{sec:discussion}, we summarize our results and discuss possible extensions of the theory and open questions.

\section{Markov Random Fields in Separoids}\label{sec:mrfs_in_separoids}

In this section, we study Markov random fields in general separoids as defined in Definition~\ref{def:separoid} to prepare for our generalizations of information characterizations of Markov random fields.
In Section~\ref{sec:mutual_independence}, we define conditional \emph{mutual} independences in separoids and present three equivalent descriptions. 
They are used in the \emph{cutset} description of Markov random fields, which we prove in Section~\ref{sec:mrfs_separoids} to be equivalent to the usual global Markov property.
Finally, Section~\ref{sec:markov_chains_separoids_app} defines Markov chains in separoids and characterizes them as Markov random fields corresponding to a path-shaped graph.
All these results will apply to $\CR$-independence, as we will in the next section, in Proposition~\ref{pro:F-separoid}, show that it satisfies the separoid axioms.

\subsection{Conditional Mutual Independences in Separoids}\label{sec:mutual_independence}

Fix a separoid $(\mon{M}, \indep)$.
We remind again of our notation: For $i, k \in \N$, we set $\range{i}{k} \coloneqq \{i, i+1, \dots, k\}$ if $i \leq k$ and $\range{i}{k} = \emptyset$, otherwise.
As a special case, set $\start{k} \coloneqq \range{1}{k}$.
For a set $I$, set $I \setminus i \coloneqq I \setminus \{i\}$.

\begin{definition}[Conditional Mutual Independence]
  \label{def:mutual_independence}
  Let $X_1, \dots X_n, Y \in \mon{M}$.
  We say $X_1, \dots, X_n$ are \emph{mutually independent given $Y$}, written
  \begin{equation*}
    \bigindep_{i = 1}^{n} X_i \ | \ Y,
  \end{equation*}
  if for all $i = 1, \dots, n$, the following pairwise independence holds:
  \begin{equation*}
    \Indep{X_i}{X_{\start{n} \setminus i}}{Y}.
  \end{equation*}
  Recall that in this expression, we have $X_{[n] \setminus i} = \prod_{j \in [n] \setminus \{i\}} X_j$.
\end{definition}

Note: 
in some contexts we may also write
\begin{equation*}
  \Big( X_1 \indep X_2 \indep  \dots \indep X_n\Big) \ | \ Y \quad \text{or} \quad
  X_1 \indep \bigindep_{i = 2}^n X_i \ | \ Y
\end{equation*}
for the mutual independence of $X_1, \dots, X_n$ given $Y$.

\begin{proposition} 
  \label{pro:mutual_independence_characterization}
  For $X_1, \dots, X_n, Y \in \mon{M}$, the following statements are equivalent:
  \begin{enumerate}
    \item $\bigindep_{i = 1}^{n} X_i \ | \ Y$;
    \item For all $i = 1, \dots, n$, the following pairwise independence holds:
      \begin{equation*}
	\Indep{X_i}{X_{\start{i-1}}}{Y};
      \end{equation*}
    \item $\bigindep_{i = 1}^{n-1} X_i \ | \ Y$ and $\Indep{X_{n}}{X_{\start{n-1}}}{Y}$;
    \item for all disjoint $I_1, I_2 \subseteq \start{n}$, one has $\Indep{X_{I_1}}{X_{I_2}}{Y}$.
  \end{enumerate}
\end{proposition}

\begin{proof}
  1 immediately implies 2 by using decomposition (S3).
  Assume 2 holds. We want to prove $1$.
  Let $i \in \start{n}$.
  Assume by induction that we already know
  \begin{equation}\label{eq:what_we_know}
    \Indep{X_i}{X_{\start{l} \setminus i}}{Y}
  \end{equation}
  for some $l \geq i$, where the case $l = i$ holds by assumption. 
  If we can show the same with $l+1$ replacing $l$, then induction shows the statement for $l = n$, which then results in mutual independence since $i$ was arbitrary.
  For the induction step, note that
  \begin{equation}
    \Indep{X_{l+1}}{X_{\start{l}}}{Y},
  \end{equation}
  which, using symmetry (S1) and weak union (S4), implies
  \begin{equation}\label{eq:another_thing_know}
    \Indep{X_i}{X_{l+1}}{Y X_{\start{l} \setminus i}}.
  \end{equation}
  Contraction (S5) applied to Equations~\eqref{eq:what_we_know} and~\eqref{eq:another_thing_know} gives
  \begin{equation*}
    \Indep{X_i}{X_{\start{l+1} \setminus i}}{Y}, 
  \end{equation*}
  which shows the induction step.

  Knowing that 1 and 2 are equivalent then immediately implies that both are equivalent to 3.

  It is also clear that 4 implies 1.
  Finally, we show that 1 implies 4:
  let $I_1, I_2 \subseteq \start{n}$ be disjoint.
  The case $I_1 = \emptyset$ is trivial by redundancy $(S2)$.
  Thus, assume $\emptyset \neq I_1 = \big\lbrace i_1, \dots, i_k\big\rbrace$ with pairwise different elements $i_l$.
  Then from decomposition (S3) applied to the independence $\Indep{X_{i_1}}{X_{\start{n} \setminus i_1}}{Y}$ we obtain
  \begin{equation*}
    \Indep{X_{i_1}}{X_{I_2}}{Y}.
  \end{equation*}
  By induction, we can assume 
  \begin{equation}\label{eq:dd}
    \Indep{X_{I_1 \setminus i_k}}{X_{I_2}}{Y}
  \end{equation}
  and want to show that we can add $X_{i_k}$ to the left-hand-side. 
  By decomposition (S3) applied to the independence $\Indep{X_{i_k}}{X_{\start{n} \setminus i_k}}{Y}$ we obtain 
  \begin{equation*}
    \Indep{X_{i_k}}{X_{I_1 \cup I_2 \setminus i_k}}{Y},
  \end{equation*}
  and thus by weak union (S4)
  \begin{equation}\label{eq:ddd}
    \Indep{X_{i_k}}{X_{I_2}}{YX_{I_1 \setminus i_k}}.
  \end{equation}
  Contraction (S5) applied to Equations~\eqref{eq:dd} and~\eqref{eq:ddd} shows
  \begin{equation*}
    \Indep{X_{I_1}}{X_{I_2}}{Y},
  \end{equation*}
  which is the independence we wanted to show.
\end{proof}

\subsection{Markov Random Fields in Separoids}\label{sec:mrfs_separoids}

Fix again a general separoid $(\mon{M}, \indep)$.
In this part, we show that Markov random fields as defined in Definition~\ref{def:global_markov_property} by the global Markov property are equivalently described by the cutset property.

\begin{definition}
  [Full Conditional Mutual Independence (FCMI)]
  Let $X_1, \dots, X_n \in \mon{M}$.
  A conditional mutual independence
  \begin{equation*}
    \bigindep_{i = 1}^q X_{L_i} \ | \ X_J,
  \end{equation*}
  where $L_1, \dots, L_q, J \subseteq \start{n}$ are pairwise disjoint sets whose union is $\start{n}$, is called a \emph{full conditional mutual independence} (with respect to n and $X_1, \dots, X_n$) --- or FCMI for short.
  \label{def:fcmi}
\end{definition}

The term ``full'' indicates that each of the variables $X_i$, $i = 1, \dots, n$, appears exactly once. 
Recall the notion of a cutset and the corresponding components from Section~\ref{sec:graph_terminology}.

\begin{definition}
  [Cutset Property]
  \label{def:mrf}
  Let $\Gr{G} = (\Ver{V}, \Ed{E})$ a graph with $\Ver{V} = \start{n}$.
  A sequence of elements $X_1, \dots, X_n \in \mon{M}$ is said to satisfy the \emph{cutset property} with respect to $\indep$ and $\Gr{G}$ if for all cutsets $\Ver{U} \subseteq \Ver{V}$, the FCMI
  \begin{equation*}
    \bigindep_{i = 1}^{\ncomp{\Ver{U}}} X_{\Ver{V}_i(\Ver{U})} \ | \ X_{\Ver{U}}
  \end{equation*}
  holds.
\end{definition}

The following proposition shows that the global Markov property and cutset property are equivalent.
For the special case of probabilistic independence, this was also stated in the introduction of~\cite{Yeung2002a}.

\begin{proposition} 
  \label{pro:equivalence_global_markov}
  Let $\Gr{G} = (\Ver{V}, \Ed{E})$ a graph with $\Ver{V} = \start{n}$.
  A sequence of elements $X_1, \dots, X_n \in \mon{M}$ form a Markov random field if and only if they satisfy the cutset property (with respect to $\indep$ and $\Gr{G}$).
\end{proposition}

\begin{proof}
  Assume $X_1, \dots, X_n$ satisfy the cutset property with respect to $\indep$ and $\Gr{G}$.
  Let $\Ver{A}, \Ver{B}, \Ver{C} \subseteq \Ver{V}$ disjoint such that $\Ver{C}$ separates $\Ver{A}$ from $\Ver{B}$.
  Without loss of generality, we can assume that $\Ver{A}$ and $\Ver{B}$ are both nonempty, since otherwise the independence $\Indep{X_{\Ver{A}}}{X_{\Ver{B}}}{X_{\Ver{C}}}$ follows from redundancy (S2) and we are done.
  This, together with the separation assumption, implies that $\Ver{C}$ is a cutset.
  Let $\Ver{V}_1(\Ver{C}), \dots, \Ver{V}_{\ncomp{\Ver{C}}}(\Ver{C})$ be the components of $\Grdif{\Gr{G}}{\Ver{C}}$.
  Since $\Ver{A}, \Ver{B} \subseteq \Ver{V} \setminus \Ver{C} = \bigcup_{i = 1}^{\ncomp{\Ver{C}}} \Ver{V}_i(\Ver{C})$, we have
  \begin{equation*}
    \Ver{A} = \bigcup_{i = 1}^{\ncomp{\Ver{C}}} \Ver{V}_i^{\Ver{A}}(\Ver{C}), \quad \Ver{B} = \bigcup_{i = 1}^{\ncomp{\Ver{C}}} \Ver{V}_i^{\Ver{B}}(\Ver{C}),
  \end{equation*}
  where $\Ver{V}_i^{\Ver{A}}(\Ver{C}) \coloneqq \Ver{A} \cap \Ver{V}_i(\Ver{C})$ and $\Ver{V}_i^{\Ver{B}}(\Ver{C}) \coloneqq \Ver{B} \cap \Ver{V}_i(\Ver{C})$.
  Now, we claim that for all $i$, we have $\Ver{V}_i^{\Ver{A}}(\Ver{C}) = \emptyset$ or $\Ver{V}_i^{\Ver{B}}(\Ver{C}) = \emptyset$:
  indeed, if there were elements $a \in \Ver{V}_i^{\Ver{A}}(\Ver{C})$ and $b \in \Ver{V}_i^{\Ver{B}}(\Ver{C})$, then they would both be in $\Ver{V}_i(\Ver{C})$ and thus connected by a walk that lies completely within $\Ver{V}_i(\Ver{C}) \subseteq \Ver{V} \setminus \Ver{C}$, a contradiction to the assumption that $\Ver{C}$ separates $\Ver{A}$ from $\Ver{B}$.

  Thus, we can write
  \begin{equation}\label{eq:how_A_written}
    \Ver{A} = \bigcup_{i \in I_1} \Ver{V}_i^{\Ver{A}}(\Ver{C}), \quad \Ver{B} = \bigcup_{i \in I_2} \Ver{V}_i^{\Ver{B}}(\Ver{C}) 
  \end{equation}
  with $I_1 \cap I_2 = \emptyset$.
  Since $X_1, \dots, X_n$ satisfy the cutset property and $\Ver{C}$ is a cutset, we obtain the FCMI $\bigindep_{i = 1}^{\ncomp{\Ver{C}}} X_{\Ver{V}_i(\Ver{C})} \ | \ X_{\Ver{C}}$,
  from which, by the equivalence of parts $1$ and $4$ in Proposition~\ref{pro:mutual_independence_characterization}, we obtain
  \begin{equation}\label{eq:weird_pairwise_independence}
    \Indep{X_{\big[\bigcup_{i \in I_1}\Ver{V}_i(\Ver{C})\big]}}{X_{\big[\bigcup_{i \in I_2}\Ver{V}_i(\Ver{C})\big]}}{X_{\Ver{C}}}.
  \end{equation}
  Equation~\eqref{eq:how_A_written} and decomposition (S3) applied to both the left and right side of Equation~\eqref{eq:weird_pairwise_independence} implies $\Indep{X_{\Ver{A}}}{X_{\Ver{B}}}{X_{\Ver{C}}}$
  and thus the global Markov property.
  Therefore, $X_1, \dots, X_n$ form a Markov random field with respect to $\indep$ and $\Gr{G}$.

  For the other direction, assume the global Markov property holds.
  Let $\Ver{U} \subseteq \Ver{V}$ be a cutset and let $\Ver{V}_1(\Ver{U}), \dots, \Ver{V}_{\ncomp{\Ver{U}}}(\Ver{U})$ be the components of $\Grdif{\Gr{G}}{\Ver{U}}$.
  We know that $\Ver{U}$ separates $\Ver{V}_1(\Ver{U})$ from $\Ver{V}_2(\Ver{U})$, which by the global Markov property implies $\Indep{X_{\Ver{V}_1(\Ver{U})}}{X_{\Ver{V}_2(\Ver{U})}}{X_{\Ver{U}}}$.
  This can be interpreted as the conditional mutual independence $\bigindep_{i = 1}^{2} X_{\Ver{V}_i(\Ver{U})} \ | \ X_{\Ver{U}}$.
  Assume by induction that we know
  \begin{equation}\label{eq:an_induction_hyp}
    \bigindep_{i = 1}^{m} X_{\Ver{V}_i(\Ver{U})} \ | \ X_{\Ver{U}}.
  \end{equation}
  Note that $\Ver{U}$ also separates $\bigcup_{i = 1}^{m} \Ver{V}_i(\Ver{U})$ from $\Ver{V}_{m+1}(\Ver{U})$, which by the global Markov property implies
  \begin{equation*}
    \Indep{X_{\Ver{V}_{m+1}(\Ver{U})}}{\prod_{i = 1}^m X_{\Ver{V}_i(\Ver{U})}}{X_{\Ver{U}}}.
  \end{equation*}
  This, together with Equation~\eqref{eq:an_induction_hyp} and the equivalence of parts 1 and 3 in Proposition~\ref{pro:mutual_independence_characterization}, implies $\bigindep_{i = 1}^{m+1} X_{\Ver{V}_i(\Ver{U})} \ | \ X_{\Ver{U}}$.
  By induction, this shows the FCMI
  \begin{equation*}
    \bigindep_{i = 1}^{\ncomp{\Ver{U}}} X_{\Ver{V}_i(\Ver{U})} \ | \ X_{\Ver{U}}.
  \end{equation*}
  Overall, this shows that $X_1, \dots, X_n$ satisfies the cutset property with respect to $\indep$ and $\Gr{G}$.
\end{proof}

\subsection{Markov Chains in Separoids}\label{sec:markov_chains_separoids_app}

Again, we fix a general separoid $(\mon{M}, \indep)$.
Recall the definition of a Markov chain from Definition~\ref{def:markov_chain}.

\begin{proposition}[Characterization of Markov Chains]
  \label{pro:characterization_markov_chain}
  Let $X_1, \dots, X_n \in \mon{M}$.
  Let $\Gr{G} = \big( \Ver{V}, \Ed{E}\big)$ be the graph with $\Ver{V} = \start{n}$ and $\Ed{E} = \big\lbrace \{i, i+1\} \mid i = 1, \dots, n-1 \big\rbrace$.
  The following properties are equivalent:
  \begin{enumerate}
    \item $X_1 \to X_2 \to \dots \to X_n$, i.e., the sequence forms a Markov chain;
    \item $X_{\start{i-1}} \to X_i \to X_{\range{i+1}{n}}$ for all $i = 1, \dots, n-1$;
    \item $X_1, \dots, X_n$ form a Markov random field with respect to $\Gr{G}$ and $\indep$.
  \end{enumerate}
\end{proposition}

\begin{proof}
  Assume 1. 
  To prove 2, let $i \in \{1, \dots, n-1\}$.
  We need to show the independence 
  \begin{equation}\label{eq:resa}
     \Indep{X_{\range{i+1}{n}}}{X_{\start{i-1}}}{X_i}.
  \end{equation}
  We already know that the independence $\Indep{X_{i+1}}{X_{\start{i-1}}}{X_i}$ holds.
  Assume by induction that 
  \begin{equation}\label{eq:ing1}
    \Indep{X_{\range{i+1}{l}}}{X_{\start{i-1}}}{X_i}
  \end{equation}
  for some $l \geq i+1$.
  If $l = n$, then we are done. 
  Otherwise, note that the independence $\Indep{X_{l+1}}{X_{\start{l-1}}}{X_l}$
  gives us, by weak union (S4), the property
  \begin{equation}\label{eq:ing2}
    \Indep{X_{l+1}}{X_{\start{i-1}}}{X_iX_{\start{i+1:l}}}.
  \end{equation}
  Contraction (S5) applied to Equations~\eqref{eq:ing1} and~\eqref{eq:ing2} results in $\Indep{X_{\range{i+1}{l+1}}}{X_{\start{i-1}}}{X_i}$.
  By induction, we obtain the result, Equation~\eqref{eq:resa}.

  Now assume 2.
  To prove 3, by Proposition~\ref{pro:equivalence_global_markov} it is enough to show the cutset property.
  For this, let $\Ver{U} \subseteq \Gr{G}$ a cutset and $\Ver{V}_1(\Ver{U})$, \dots, $\Ver{V}_{\ncomp{\Ver{U}}}(\Ver{U})$ the corresponding components in $\Grdif{\Gr{G}}{\Ver{U}}$.
  A component in $\Gr{G}$ necessarily consists of consecutive elements, and we can thus assume that they are ordered in such a way that $v_1 < \dots < v_{\ncomp{\Ver{U}}}$ for all vertices $v_1 \in \Ver{V}_1(\Ver{U})$, \dots, $v_{\ncomp{U}} \in \Ver{V}_{\ncomp{\Ver{U}}}(\Ver{U})$.

  Now, let $i  \in \{2, \dots, \ncomp{\Ver{U}}\}$ be arbitrary.
  Let $u \in \Ver{U}$ be any element such that $v < u < w$ for all $v \in \bigcup_{k = 1}^{i-1} \Ver{V}_k(\Ver{U})$ and $w \in \bigcup_{k = i}^{\ncomp{\Ver{U}}}\Ver{V}_k(\Ver{U})$; this clearly exists since otherwise we could merge the components with indices $i-1$ and $i$.
  Let $\Ver{U}_{u-}$ and $\Ver{U}_{u+}$ be the sets of elements of $\Ver{U}$ that are smaller and larger than $u$, respectively.
  Then by 2, we obtain
  \begin{equation*}
    \Indep{\big(X_{\Ver{V}_{i}(\Ver{U})} \cdots X_{\Ver{V}_{\ncomp{\Ver{U}}}(\Ver{U})} X_{\Ver{U}_{u+}}\big)}
    {\big(X_{\Ver{V}_{1}(\Ver{U})} \cdots X_{\Ver{V}_{i-1}(\Ver{U})} X_{\Ver{U}_{u-}}\big)}{X_u}.
  \end{equation*}
  With weak union (S4) applied to both sides, we obtain
  \begin{equation*}
     \Indep{\big(X_{\Ver{V}_{i}(\Ver{U})} \cdots X_{\Ver{V}_{\ncomp{\Ver{U}}}(\Ver{U})}\big)}
     {\big(X_{\Ver{V}_{1}(\Ver{U})} \cdots X_{\Ver{V}_{i-1}(\Ver{U})}\big)}{X_{\Ver{U}}}.
  \end{equation*}
  Decomposition (S3) applied to the left side gives
  \begin{equation*}
      \Indep{X_{\Ver{V}_{i}(\Ver{U})} }
     {\big(X_{\Ver{V}_{1}(\Ver{U})} \cdots X_{\Ver{V}_{i-1}(\Ver{U})}\big)}{X_{\Ver{U}}}. 
  \end{equation*}
  Since $i$ was arbitrary, by the equivalence of 1 and 2 in Proposition~\ref{pro:mutual_independence_characterization}, we obtain the FCMI $\bigindep_{i = 1}^{\ncomp{\Ver{U}}} X_{\Ver{V}_i(\Ver{U})} \ | \ X_{\Ver{U}}$,
  showing the cutset property and thus 3.

  Since $\{i-1\}$ separates $\{i\}$ from $\{1, \dots, i-2\}$ in $\Gr{G}$, the global Markov property shows that 3 implies 1.
\end{proof}

\section{Characterizing \texorpdfstring{$\CR$}{CR}-Independences and \texorpdfstring{$\CR$}{CR}-Markov Random Fields}\label{sec:yeung2002}

In this section we generalize the work~\cite{Yeung2002a} on information characterizations of (full) conditional mutual independences and Markov random fields from Shannon entropy to general $\CR$.
The reader may also compare with~\cite{Yeung2008}, Chapter 12, which contains the same content as \cite{Yeung2002a} with more explanations.

In Section~\ref{sec:subset_determination}, we state a main technique used in our proofs, Theorem~\ref{thm:subset_determination}, which we term \emph{subset determination}.
It states that the value of a region in an $\CR$-diagram always fully determines the value of all subregions, including the atomic parts. 
This result will replace the use of inequalities in~\cite{Yeung2002a}.
In Section~\ref{sec:separoids}, we then show that the $\CR$-independence defined in Definition~\ref{def:F_independence} satisfies the separoid axioms, and so all results from Section~\ref{sec:mrfs_in_separoids} apply.

In Section~\ref{sec:dual_total_correlation_charac}, we define conditional mutual $\CR$-independences and show in Theorem~\ref{thm:charac_using_dual} a characterization by the vanishing of a conditional $\CR$--dual total correlation and its corresponding atoms.
In Section~\ref{sec:full_conditional_independences}, with Theorem~\ref{thm:FCMIs_characterization}, we then generalize this to a characterization of full conditional mutual $\CR$-independences.
By Proposition~\ref{pro:equivalence_global_markov}, Markov random fields in a separoid are equivalently characterized by the cutset property, and thus by a set of full conditional mutual independences. 
Thus, the previous results lead to the proof for the characterization of $\CR$-Markov random fields in terms of the $\CR$-diagram, as stated in Theorem~\ref{thm:mrf_characterization}. 
We also specialize this to $\CR$-Markov chains.
Finally, in Appendix~\ref{sec:generalizing_yeung_2019}, we briefly look at~\cite{Yeung2019}, which builds on~\cite{Yeung2002a}, and explain which of the results transfer to our generalized setting.
We put this into the appendix since no later results build on this.

Let in this whole section $\mon{M}$ be a commutative, idempotent monoid acting additively on an abelian group $\gro{G}$ and $\CR: \mon{M} \to \gro{G}$ a function satisfying the chain rule Equation~\eqref{eq:cocycle_equationn}.

\subsection{Subset Determination}\label{sec:subset_determination}

The following theorem, which did not appear in~\cite{Lang2022}, highlights a property that we call \emph{subset determination}.
This crucial property lies at the heart of the proofs of the main results in this work.

\begin{theorem}[Subset Determination] 
  \label{thm:subset_determination}
  Let $\mon{M}$ be a commutative, idempotent monoid acting additively on an abelian group $\gro{G}$ and $\CR: \mon{M} \to \gro{G}$ be a function satisfying the chain rule Equation~\eqref{eq:cocycle_equationn}. 
  Fix elements $X_1, \dots, X_n \in \mon{M}$ and let $\set{X} \coloneqq \set{X}(n)$ be defined as in Equation~\eqref{eq:Sigma_definition}, resulting in the $\gro{G}$-valued measure $\set{\CR}: 2^{\set{X}} \to \gro{G}$ from Theorem~\ref{thm:hu_kuo_ting_generalized}.

  Then for any $A \subseteq \set{X}$ and any atom $\atom{I} \in A$, one has
  \begin{equation*}
    \set{\CR}(\atom{I}) = \sum_{K \subseteq I} (-1)^{\num{K} - \num{I}} \cdot X_{\start{n} \setminus K}.\set{\CR}(A).
  \end{equation*}
  In particular, if $\set{\CR}(A) = 0$, then $\set{\CR}(p_I) = 0$ for all atoms $p_I \in A$, and consequently $\set{\CR}(B) = 0$ for all $B \subseteq A$.
\end{theorem}

\begin{remark} 
  \label{rem:subset_determination_and_slices}
  It is important to note that this theorem is \emph{not} true if we work in the setting of information functions that apply to probability mass functions and restrict our attention to a fixed probability mass function $P$.
  For example, the total entropy $\Ent(XY; P)$ of a joint variable $XY$ does not determine the mutual information $\Ent(X;Y; P)$ between the two, even though $\Ent(XY)$ \emph{does} determine $\Ent(X; Y)$.
  Nevertheless, we are able to apply our results also to fixed probability mass functions, as we explain in Appendix~\ref{sec:slices} on \emph{slices} of $\Ent$-diagrams.
\end{remark}

We first need more notation and several lemmas:

\begin{notation}
  \label{not:interaction_notation}
  For all $\emptyset \neq I = \big\lbrace i_1 <  \dots < i_q\big\rbrace \subseteq \start{n}$, we write
  \begin{equation*}
    \CR\big( \bigscolon_{i \in I} X_i\big) \coloneqq \CR\big( X_{i_1}; \dots ; X_{i_q} \big).
  \end{equation*}
\end{notation}

\begin{lemma}
  \label{lem:how_atoms_look_like}
  For all $\emptyset \neq I \subseteq \start{n}$, one has 
  \begin{equation*}
    \set{\CR}(\atom{I}) = X_{\start{n} \setminus I}.\CR\big( \bigscolon_{i \in I} X_i\big).
  \end{equation*}
\end{lemma}

\begin{proof}
  This follows directly from Equation~\eqref{eq:general_position} and Hu's Theorem.
\end{proof}

\begin{lemma}
  \label{lem:atoms_annihilated}
  Assume $\atom{L}$ is an atom and $K \subseteq \start{n}$.
  Then 
  \begin{equation*}
    X_{\start{n} \setminus K}.\set{\CR}(\atom{L}) = 
    \begin{cases}
      \set{\CR}(\atom{L}), \ L \subseteq K, \\
      0, \ \mathrm{else}.
    \end{cases}
  \end{equation*}
\end{lemma}

\begin{proof}
  By Lemma~\ref{lem:how_atoms_look_like}, we have
  \begin{equation*}
    X_{\start{n} \setminus K}.\set{\CR}(\atom{L}) = X_{\start{n} \setminus K}. \Big( X_{\start{n} \setminus L}.\CR\big(\bigscolon_{l \in L} X_l \big)\Big)  = X_{\start{n} \setminus (L \cap K)}.\CR\big( \bigscolon_{l \in L} X_l \big). \\
  \end{equation*}
  From this, we immediately see the result for the case $L \subseteq K$.
  If $L \nsubseteq K$, then there is $l \in L \setminus K$. 
  Then $l \in \start{n} \setminus (L \cap K)$ and consequently, with $Y \coloneqq X_{\big(\start{n}\setminus{l}\big) \setminus (L \cap K)}$:
  \begin{equation*}
    X_{\start{n} \setminus K}.\set{\CR}(\atom{L}) = Y. \Big( X_l.\CR\big( \bigscolon_{l' \in L} X_{l'})\Big) 
     = Y. \set{\CR} \Bigg( \bigcap_{l' \in L} \set{X}_{l'} \setminus \set{X}_l\Bigg) 
     = Y.\set{\CR}(\emptyset) 
     = 0,
  \end{equation*}
  where we have used Hu's Theorem~\ref{thm:hu_kuo_ting_generalized} in the second step.
  That was to show.
\end{proof}

\begin{lemma}
  \label{lem:pure_combinatorics}
  Let $L \subseteq I$ be two sets. 
  Then
  \begin{equation*}
    \sum_{K: \ L \subseteq K \subseteq I} (-1)^{\num{K}} = (-1)^{\num{L}} \cdot \indic{I = L}.
  \end{equation*}
\end{lemma}

\begin{proof}
  We have:
  \begin{align*}
    \sum_{K: \ L \subseteq K \subseteq I} (-1)^{\num{K}} & = \sum_{k = \num{L}}^{\num{I}} (-1)^{k} \cdot \num{\Big\lbrace K \  \big| \  L \subseteq K \subseteq I, \ \num{K} = k\Big\rbrace} \\
    & = \sum_{k = \num{L}}^{\num{I}} (-1)^k \cdot \binom{\num{I} - \num{L}}{k - \num{L}} \\
    & = (-1)^{\num{L}} \cdot \sum_{k = 0}^{\num{I} - \num{L}} \binom{\num{I} - \num{L}}{k} \cdot 1^{\num{I} - \num{L} - k} \cdot (-1)^{k} \\
    & \overset{(\star)}{=} (-1)^{\num{L}} \cdot (1 - 1)^{\num{I} - \num{L}} \\
    & = (-1)^{\num{L}} \cdot  \indic{I = L}.
  \end{align*}
  In $(\star)$, we used the well-known binomial theorem and in the final step that $0^0 = 1$.\footnote{If one is not comfortable with the definition $0^0 = 1$, one can also directly verify the overall result in the case $L = I$.}
\end{proof}

\begin{proof}[Proof of Theorem~\ref{thm:subset_determination}]
  We have
  \begingroup
  \allowdisplaybreaks
  \begin{align*}
    \sum_{K \subseteq I}(-1)^{\num{K} - \num{I}} \cdot X_{\start{n} \setminus K}. \set{\CR}(A) & = 
    \sum_{K \subseteq I}(-1)^{\num{K} - \num{I}} \cdot X_{\start{n} \setminus K}. \Bigg( \sum_{\atom{L} \in A} \set{\CR}(\atom{L})\Bigg)\\
    & = \sum_{\atom{L} \in A} (-1)^{ - \num{I}} \sum_{K \subseteq I}  (-1)^{\num{K}} \cdot X_{\start{n} \setminus K}.\set{\CR}(\atom{L}) \\
    & = \sum_{\atom{L} \in A} (-1)^{- \num{I}} \cdot \Bigg(\sum_{K: \ L \subseteq K \subseteq I} (-1)^{\num{K}} \Bigg) \cdot \set{\CR}(\atom{L}) & \big( \mathrm{Lemma}~\ref{lem:atoms_annihilated}\big) \\
    & = \sum_{\atom{L} \in A} (-1)^{- \num{I}} \cdot (-1)^{\num{L}} \cdot \indic{L = I} \cdot \set{\CR}(\atom{L}) & \big( \mathrm{Lemma}~\ref{lem:pure_combinatorics} \big)\\
    & = (-1)^{- \num{I}} \cdot (-1)^{\num{I}} \cdot \set{\CR}(\atom{I}) & \big( \atom{I} \in A \big)\\
    & = \set{\CR}(\atom{I}). \qedhere
  \end{align*}
  \endgroup
\end{proof}

\begin{remark}
  \label{rem:connection_to_F_classification}
  For the special case that $A = \set{X}$, the theorem says that $\CR\big(X_{\start{n}}\big)$ determines $\set{F}(\atom{I})$ for all atoms $\atom{I} \in \set{X}$.
  If $\mon{M}$ is generated by the $X_1, \dots, X_n$, this means that $\CR\big(X_{\start{n}}\big)$ entirely determines $\CR$.
  In Appendix~\ref{sec:classification_CR}, we generalize this observation and show that if $\top \in \mon{M}$ is any so-called \emph{top element}, then $\CR(\top)$ determines $\CR$ entirely.
  This leads to a one-to-one correspondence between elements of $\gro{G}$ that are \emph{annihilated} by $\top$ and functions $\CR: \mon{M} \to \gro{G}$ satisfying the chain rule.
  We also interpret this result as the vanishing of a cohomology group of degree one. 
\end{remark}

\subsection{\texorpdfstring{$\CR$}--Independence Satisfies the Separoid Axioms}\label{sec:separoids}

Let $\mon{M}$ again be a commutative, idempotent monoid acting on an abelian group $\gro{G}$, and $\CR: \mon{M} \to \gro{G}$ a function satisfying the chain rule Equation~\eqref{eq:cocycle_equationn}.
We now show that the $\CR$-independence introduced in Definition~\ref{def:F_independence} satisfies the separoid axioms.

Our proof of the following proposition, which shows that $(\mon{M}, \indep_{\CR})$ is a separoid, makes extensive use of Hu's Theorem~\ref{thm:hu_kuo_ting_generalized} and the subset determination property, Theorem~\ref{thm:subset_determination}.

\begin{proposition}
  \label{pro:F-separoid}
  Let $\mon{M}$ be a commutative, idempotent monoid acting on the abelian group $\gro{G}$ and $\CR: \mon{M} \to \gro{G}$ a function satisfying the chain rule Equation~\eqref{eq:cocycle_equationn}.
  Let $\indep_\CR$ be the $\CR$-independence relation on $\mon{M}$ from Definition~\ref{def:F_independence}. 
  Then $(\mon{M}, \indep_\CR)$ is a separoid.
\end{proposition}

\begin{proof}
  We need to prove the five separoid axioms from Definition~\ref{def:separoid}.
  In proving each axiom, we will use the elements of $\mon{M}$ appearing in the formulas as the elements $X_1, \dots, X_n$ in Hu's Theorem~\ref{thm:hu_kuo_ting_generalized}.
  \begin{enumerate}
    \item Symmetry: for any $X, Y, Z \in \mon{M}$, we have
      \begin{equation*}
	Z.\CR(X; Y) = \set{\CR} \big( \set{X} \cap \set{Y} \setminus \set{Z} \big) = \set{\CR} \big(  \set{Y} \cap \set{X} \setminus \set{Z} \big) = Z.\CR(Y; X),
      \end{equation*}
      showing symmetry $\IndepF{\CR}{X}{Y}{Z} \ \ \Longrightarrow \ \  \IndepF{\CR}{Y}{X}{Z}$.
     \item Redundancy:
       assume $W \precsim Z$, i.e., $WZ = Z$.
       It follows
     \begin{equation*}
       \set{\CR}\big(\set{W} \setminus \set{Z}\big) = Z.\CR(W) \overset{\text{Eq.~\eqref{eq:cocycle_equationn}}}{=} \CR(ZW) - \CR(Z) = \CR(Z) - \CR(Z) = 0.
     \end{equation*}
     Since $\set{W} \cap \set{Y} \setminus \set{Z} \subseteq \set{W} \setminus \set{Z}$, subset determination Theorem~\ref{thm:subset_determination} shows
     \begin{equation*}
       Z.\CR(W;Y) = \set{\CR}\big( \set{W} \cap \set{Y} \setminus \set{Z}\big) = 0
     \end{equation*}
     and thus $\IndepF{\CR}{W}{Y}{Z}$.
     \item Decomposition and weak union: 
       assume $\IndepF{\CR}{WX}{Y}{Z}$.
       Then we have
       \begin{align*}
	\set{\CR} \big( \set{WX} \cap \set{Y} \setminus \set{Z}\big) =  Z.\CR(WX; Y) = 0 . \\
       \end{align*}
       Note that 
       \begin{equation*}
	 \set{X} \cap \set{Y} \setminus \set{Z} \subseteq \set{WX} \cap \set{Y} \setminus \set{Z} \supseteq \set{W} \cap \set{Y} \setminus \set{XZ}. 
       \end{equation*}
       Thus, subset determination Theorem~\ref{thm:subset_determination} shows
       \begin{align*}
	&  Z.\CR(X;Y) = \set{\CR}\big( \set{X} \cap \set{Y} \setminus \set{Z} \big) = 0, \\
	&  (XZ).\CR(W; Y) = \set{\CR}\big( \set{W} \cap \set{Y} \setminus \set{XZ}\big) = 0,
       \end{align*}
       which shows $\IndepF{\CR}{X}{Y}{Z}$ and $\IndepF{\CR}{W}{Y}{XZ}$ and thus both decomposition and weak union.
     \item Contraction: assume $\IndepF{\CR}{W}{Y}{XZ}$ and $\IndepF{\CR}{X}{Y}{Z}$.
       We obtain from Hu's Theorem:
       \begin{align*}
	 Z.\CR(WX; Y) & = \set{\CR} \big( \set{WX} \cap \set{Y} \setminus \set{Z} \big) \\
	 & = \set{\CR}\big( \set{W} \cap \set{Y} \setminus \set{XZ}\big) + \set{\CR}\big( \set{X} \cap \set{Y} \setminus \set{Z} \big) \\
	 & = (XZ).\CR(W; Y) + Z.\CR(X;Y) \\
	 & = 0 + 0 = 0.
       \end{align*}
       Thus, we have proven $\IndepF{\CR}{WX}{Y}{Z}$.
       In the calculation, the second step uses that $\set{\CR}$ is additive over disjoint unions.
   \end{enumerate}
\end{proof}

In the following proposition, we show that $\CR$-independences are preserved under conditioning:

\begin{proposition} 
  \label{pro:weird_conditioning_trick}
  Let $(\mon{M}, \indep_{\CR})$ be the separoid from above.
  Let $X, Y, Z, W \in \mon{M}$.
  Then the following implication holds:
  \begin{equation*}
    \IndepF{\CR}{X}{Y}{Z} \quad \Longrightarrow \quad \IndepF{\CR}{X}{Y}{WZ}.
  \end{equation*}
\end{proposition}

\begin{proof}
  The claim follows from subset determination Theorem~\ref{thm:subset_determination} and Hu's Theorem~\ref{thm:hu_kuo_ting_generalized} by using that $\set{X} \cap \set{Y} \setminus \set{WZ} \subseteq \set{X} \cap \set{Y} \setminus \set{Z}$.
\end{proof}

\begin{remark}
  Note that the preceding proposition is incorrect when working in the probabilistic setting and fixing the probability mass function. 
  For example, the case of a collider
  \begin{equation*}
    \begin{tikzcd}
      X \ar[dr] & & Y \ar[dl] \\
    &  W
    \end{tikzcd}
  \end{equation*}
  in a Bayesian network shows that it is possible that the probabilistic independence $X \indep_P Y$ is true, while $X \indep_P Y \mid W$ is not.
  Nevertheless, we will be able to apply our theory to the probabilistic case in Sections~\ref{sec:specializing_to_probabilistic} and~\ref{sec:slices}.
  \label{rem:independence_propagates}
\end{remark}

\subsection{Conditional Mutual \texorpdfstring{$\CR$}--Independences and \texorpdfstring{$\CR$}--Dual Total Correlation}

The following definition specializes conditional mutual independences, Definition~\ref{def:mutual_independence}, to the case that the independence relation is given by $\indep_\CR$:

\begin{definition}
  [Conditional Mutual $\CR$-Independence]
  \label{def:conditional_mutual_independence_F}
  Let $X_1, \dots, X_n, Y \in \mon{M}$.
  If $X_1, \dots, X_n \in \mon{M}$ are mutually independent given $Y \in \mon{M}$ with respect to $\indep_{\CR}$, then we write
  \begin{equation*}
    \bigindepF{\CR}_{i = 1}^{n} X_i \ | \ Y.
  \end{equation*}
  We call this a \emph{conditional mutual $\CR$-independence}, and we say $X_1, \dots, X_n$ are mutually $\CR$-independent given $Y$.
\end{definition}

As in Proposition~\ref{pro:weird_conditioning_trick}, we get the following peculiar implication which we will later use:

\begin{proposition}
  \label{pro:weird_mutual_conditioning_trick}
  Let $X_1, \dots, X_n, Y, W \in \mon{M}$.
  Then we have the following implication:
  \begin{equation*}
    \bigindepF{\CR}_{i = 1}^{n} X_i \ | \ Y \quad \Longrightarrow \quad \bigindepF{\CR}_{i = 1}^{n} X_i \ | \ WY.
  \end{equation*}
\end{proposition}

\begin{proof}
  This follows immediately from Proposition~\ref{pro:weird_conditioning_trick} using that conditional mutual $\CR$-independences are defined as a combination of conditional pairwise $\CR$-independences.
\end{proof}

As in Remark~\ref{rem:independence_propagates}, we mention that the preceding proposition does not hold in the probabilistic setting when \emph{fixing} the underlying probability mass function.

\subsubsection*{Warmup: Characterizing Pairwise Conditional $\CR$-Independence}

\begin{proposition}
  \label{pro:pairwise_independence_characterization}
  Let $X_1, X_2, Y \in \mon{M}$.
  Then the following are equivalent:
  \begin{enumerate}
     \item $\IndepF{\CR}{X_1}{X_2}{Y}$;
     \item $Y.\CR(X_1X_2) = Y.\CR(X_1) + Y.\CR(X_2)$;
     \item $Y.\CR(X_1X_2) = (YX_2).\CR(X_1) + (YX_1).\CR(X_2)$;
     \item $Y.\CR(X_1) = (YX_2).\CR(X_1)$;
     \item $Y.\CR(X_2) = (YX_1).\CR(X_2)$.
  \end{enumerate}
\end{proposition}

\begin{proof}
  Using Hu's Theorem~\ref{thm:hu_kuo_ting_generalized} with elements $X_1, X_2, Y$, we generally have
  \begin{align*}
    Y.\CR(X_1X_2) & = \set{\CR} \Big( \set{X}_1 \cup \set{X}_2 \setminus \set{Y}\Big) \\
    & = \set{\CR} \Big( \set{X}_1 \setminus \set{Y}\Big) + \set{\CR} \Big( \set{X}_2 \setminus \set{Y} \Big) - \set{\CR} \Big( \set{X}_1 \cap \set{X}_2 \setminus \set{Y} \Big) \\
    & = Y.\CR(X_1) + Y.\CR(X_2) - Y.\CR(X_1; X_2).
  \end{align*}
  This shows the equivalence of 1 and 2.
  Similarly, one can show the decomposition
  \begin{equation*}
    Y.\CR(X_1X_2) = (YX_2).\CR(X_1) + (YX_1).\CR(X_2) + Y.\CR(X_1; X_2),
  \end{equation*}
  which shows the equivalence of 1 and 3.
  Finally, we similarly obtain decompositions
  \begin{align*}
    Y.\CR(X_1) & = (YX_2).\CR(X_1) + Y.\CR(X_1; X_2); \\
    Y.\CR(X_2) & = (YX_1).\CR(X_2) + Y.\CR(X_1; X_2).
  \end{align*}
  The first decomposition shows the equivalence of 1 and 4, and the second the one of 1 and 5.
\end{proof}

Of interest to us are especially properties 2 and 3, which can equivalently be expressed as the following vanishing conditions:
\begin{align*}
  & Y. \Big[ \CR(X_1) + \CR(X_2) - \CR(X_1X_2) \Big] = 0; \\
  & Y. \Big[\CR(X_1X_2) - \big( X_2.\CR(X_1) + X_1.\CR(X_2)\big) \Big] = 0.
\end{align*}
The concepts of $\CR$--total correlation and $\CR$--dual total correlations provide natural generalizations of the quantities at the left-hand-sides of these conditions.
These generalize total correlation~\citep{Watanabe1960} and dual total correlation~\citep{Han1978} by replacing Shannon entropy $\Ent$ in the defining expressions by $\CR$:

\begin{definition}
  [$\CR$--(Dual )Total Correlation]
  \label{def:(dual)_total_correlation}
  Let $X_1, \dots, X_n \in \mon{M}$.
  Then their \emph{$\CR$--total correlation} is given by
  \begin{equation*}
    \TC{\CR}(X_1; \dots ; X_n) \coloneqq \sum_{i = 1}^{n} \CR(X_i) - \CR\big( X_{[n]}\big).
  \end{equation*}
  Similarly, the \emph{$\CR$--dual total correlation} is given by
  \begin{equation*}
    \DTC{\CR}(X_1; \dots ; X_n) \coloneqq \CR\big( X_{\start{n}}\big) - \sum_{i = 1}^{n} X_{\start{n} \setminus i}.\CR(X_i),
  \end{equation*}
  where $\start{n} \setminus i \coloneqq \start{n} \setminus \{i\}$.
  If $I = \big\lbrace i_1 < \dots < i_q \big\rbrace \subseteq \start{n}$, then we also write
  \begin{align*}
    \TC{\CR}\big( \bigscolon_{i \in I} X_i \big) & \coloneqq \TC{\CR}(X_{i_1}; \dots ; X_{i_q}), \\
    \DTC{\CR}\big( \bigscolon_{i \in I} X_i \big) & \coloneqq \DTC{\CR}(X_{i_1}; \dots ; X_{i_q}).
  \end{align*}
\end{definition}

Similarly to Proposition~\ref{pro:pairwise_independence_characterization}, we want to use $\CR$--total correlation and $\CR$--dual total correlation to characterize conditional \emph{mutual} $\CR$-independences. 
We will focus on the case of $\CR$--dual total correlation, which is slightly easier and works in full generality.
For the interested reader, we consider the case of $\CR$--total correlation in Appendix~\ref{sec:total_correlation_charac}.
This only provides a valid characterization in the case that the group $\gro{G}$ is \emph{torsion-free}, as Example~\ref{exa:torsion_example} will demonstrate.

\subsubsection*{Characterization using \texorpdfstring{$\CR$}--Dual Total Correlation}\label{sec:dual_total_correlation_charac}

Let $X_1, \dots, X_n \in \mon{M}$. For $I = \big\lbrace i_1 < \dots < i_q\big\rbrace \subseteq \start{n}$, we can then consider the interaction term
$\CR(X_{i_1}; \dots ; X_{i_q})$.
Recall the notation $\CR\big(\bigscolon_{i \in I} X_i \big) \coloneqq \CR\big(X_{i_1}; \dots ; X_{i_q}\big)$.

\begin{theorem}
  \label{thm:charac_using_dual}
  Let $\mon{M}$ be a commutative, idempotent monoid acting additively on an abelian group $\gro{G}$, and $\CR: \mon{M} \to \gro{G}$ a function satisfying the chain rule Equation~\eqref{eq:cocycle_equationn}.
  Let $X_1, \dots, X_n, Y \in \mon{M}$.
  Then the following properties are equivalent:
  \begin{enumerate}
    \item $\bigindepF{\CR}_{i = 1}^{n} X_i \ | \ Y$;
    \item $Y.\DTC{\CR}\big( \bigscolon_{i \in \start{n}} X_i\big) = 0$;
    \item $\big( YX_{\start{n} \setminus I}\big).\CR\big(\bigscolon_{i \in I} X_i \big) = 0$ for all $I \subseteq \start{n}$ with $\num{I} \geq 2$;
    \item $Y.\CR\big( X_i ; X_{\start{n} \setminus i}\big) = 0$ for all $i = 1, \dots, n$.
  \end{enumerate} 
\end{theorem}

\begin{proof}
  Assume 1. 
  We prove 2 by induction over $n$, with the case $n = 2$ corresponding to the equivalence of 1 and 3 in Proposition~\ref{pro:pairwise_independence_characterization}.
  Let $n \geq 3$.
  By Proposition~\ref{pro:mutual_independence_characterization} we have
  \begin{equation*}
    \bigindepF{\CR}_{i = 1}^{n-1} X_i \ | \ Y, \quad \IndepF{\CR}{X_n}{X_{\start{n-1}}}{Y}.
  \end{equation*}
  The first $\CR$-independence implies by Proposition~\ref{pro:weird_mutual_conditioning_trick} the following: $\bigindepF{\CR}_{i = 1}^{n-1} X_i \ | \ YX_n.$
  By induction, first using the pairwise and then the mutual case,\footnote{Similarly to Proposition~\ref{pro:pairwise_independence_characterization}, we write the vanishing of the $\CR$--dual total correlation as a different equality by bringing a term to the other side.} we then obtain:
  \begin{align*}
    Y.\CR\big( X_{\start{n}}\big) & = (YX_{n}).\CR\big(X_{\start{n-1}}\big) + \big( YX_{\start{n-1}}\big).\CR(X_n) \\
    & = \sum_{i = 1}^{n-1} \big( Y X_n X_{\start{n-1} \setminus i}\big).\CR(X_i) + \big( YX_{\start{n-1}}\big).\CR(X_n) \\
    & = Y. \bigg( \sum_{i = 1}^{n} X_{\start{n} \setminus i}.\CR(X_i) \bigg).
  \end{align*}
  That shows 2.

  For the rest of the proof, write $X_{n+1} \coloneqq Y$ and assume that $X_1, \dots, X_{n+1}$ are the elements in Hu's Theorem~\ref{thm:hu_kuo_ting_generalized}.

  Now assume 2. 
  Note that
  \begin{align*}
   0 =  Y.\DTC{\CR}\big(\bigscolon_{i \in \start{n}} X_i \big) & = X_{n+1}.\CR\big(X_{\start{n}}\big) - \sum_{i = 1}^{n} X_{\start{n+1} \setminus i}.\CR(X_i) \\
    & = \set{\CR}\Big( \set{X}_{\start{n}} \setminus \set{X}_{n+1}\Big) - \sum_{i = 1}^{n} \set{\CR}(\atom{i}) \\
    & = \set{\CR}\Big( \big\lbrace \atom{I} \mid I \subseteq \start{n}, \num{I} \geq 2 \big\rbrace \Big),
  \end{align*}
  where we used Lemma~\ref{lem:how_atoms_look_like} in the second step.
  Let $I \subseteq \start{n}$ with $\num{I} \geq 2$. Then, subset determination Theorem~\ref{thm:subset_determination} and the same corollary again imply:
  \begin{align*}
    0 = \set{\CR}(\atom{I}) = X_{\start{n+1} \setminus I}.\CR\big( \bigscolon_{i \in I} X_i\big) = \big( YX_{\start{n} \setminus I}\big).\CR\big(\bigscolon_{i \in I} X_i\big).
  \end{align*}
  That is precisely 3.

  Now, assume 3.
  To prove 4, for symmetry reasons it is enough to show that $Y.\CR\big(X_n, X_{\start{n-1}} \big) = 0$.
  We have
  \begin{align*}
    Y.\CR\big( X_n; X_{\start{n-1}}\big) & = X_{n+1}.\CR\big( X_n; X_{\start{n-1}}\big) \\
    & = \set{\CR}\Big( \set{X}_n \cap  \set{X}_{\start{n-1}} \setminus \set{X}_{n+1}\Big) \\
    & = \sum_{\substack{I \subseteq \start{n}:  \ n \in I, \\ I \cap \start{n-1} \neq \emptyset}} \set{\CR}(\atom{I}) \\
    & = \sum_{\substack{I \subseteq \start{n}: \ n \in I,  \\ I \cap \start{n-1} \neq \emptyset}} X_{\start{n+1} \setminus I}.\CR\big(\bigscolon_{i \in I} X_i \big) \\
    & = \sum_{\substack{I \subseteq \start{n}:  \ n \in I, \\ I \cap \start{n-1} \neq \emptyset}} \big( YX_{\start{n} \setminus I}\big).\CR\big(\bigscolon_{i \in I} X_i \big) \\
    & = 0,
  \end{align*}
  where the fourth step used Lemma~\ref{lem:how_atoms_look_like}.
  The last step follows since for all $I$ over which we sum, we necessarily have $\num{I} \geq 2$.

  Assuming 4, 1 follows immediately by the definitions of conditional mutual and pairwise $\CR$-independences.
\end{proof}

\subsection{Full Conditional Mutual \texorpdfstring{$\CR$}--Independences}\label{sec:full_conditional_independences}

We now build on Section~\ref{sec:dual_total_correlation_charac}.
We will consider mutual $\CR$-independences of variables that are \emph{themselves} products of several variables. 
The main result of this section, Theorem~\ref{thm:FCMIs_characterization}, will generalize the equivalence between properties 1 and 3 in Theorem~\ref{thm:charac_using_dual}.

Fix $n \geq 0$ and $X_1, \dots, X_n \in \mon{M}$.
This gives rise to a set $\set{X} = \set{X}(n)$ of $2^{n} - 1$ atoms according to Equation~\eqref{eq:Sigma_definition} and a $\gro{G}$-valued measure $\set{F}: 2^{\set{X}} \to \gro{G}$ according to Hu's Theorem~\ref{thm:hu_kuo_ting_generalized}.

In the following, if $W_i$ are sets indexed with $i \in I$, then $W_I$ denotes $\bigcup_{i \in I} W_i$.
We now define $\CR$-FCMIs, specializing the notion of FCMIs from Definition~\ref{def:fcmi} to the setting with the independence relation given as $\indep_{\CR}$.
We will, however, first define the notion of a conditional partition, since this will prove valuable when studying the effect of $\CR$-FCMIs on $\CR$-diagrams:

\begin{definition}
  [Conditional Partition]
  \label{def:conditional_partition}
  Let $q \geq 1$, $L_i \subseteq \start{n}$ for all $i \in \start{q}$ and $J \subseteq \start{n}$.
  Set $L \coloneqq L_{\start{q}} = \bigcup_{i = 1}^{q} L_i$.
  Assume that the $L_i$ and $J$ are all pairwise disjoint and cover $[n]$.
  Then the family $\FCMI{K} \coloneqq \big( J, L_i, 1 \leq i \leq q \big)$
  is called a \emph{conditional partition} of $\start{n}$.\footnote{We allow $J$ and the $L_i$ to be empty, different from usual partitions.} 
\end{definition}

\begin{definition}
  [Full Conditional Mutual $\CR$-Independence ($\CR$-FCMI)]
  \label{def:f_fcmi}
  Let $\FCMI{K} = \big( J, L_i, 1 \leq i \leq q \big)$ be a conditional partition with $q \geq 2$.
  Then an FCMI of the form
  \begin{equation}\label{eq:first_fcmi}
    \bigindepF{\CR}_{i = 1}^{q} X_{L_i} \ | \ X_J
  \end{equation}
  is called a full conditional mutual $\CR$-independence (with respect to the previously fixed elements $X_1, \dots, X_n$) --- or $\CR$-FCMI for short.
  If this $\CR$-FCMI holds, then we also say that the conditional partition $\FCMI{K} = \big( J, L_i, 1 \leq i \leq q \big)$ \emph{induces} an $\CR$-FCMI (with respect to the previously fixed elements $X_1, \dots, X_n$).
\end{definition}

\begin{definition}
  [Image of a Conditional Partition]
  \label{def:image_of_FCMIs}
  Let $\FCMI{K} = \big( J, L_i, 1 \leq i \leq q \big)$ be a conditional partition of $\start{n}$ with $q \geq 2$.
  Then its image is defined as
  \begin{equation*}
    \IM(\FCMI{K}) \coloneqq \bigg\lbrace \atom{W} \in \set{X}(n) \ \Big| \ \exists I \subseteq \start{q}, \num{I} \geq 2, \forall i \in I \ \ \exists \emptyset \neq W_i \subseteq L_i: W = W_I = \bigcup_{i \in I} W_i \bigg\rbrace ,
  \end{equation*}
  i.e., as a certain set of atoms in $\set{X}(n)$.
\end{definition}

\begin{lemma}  
  \label{lem:atoms_in_peculiar_set}
  Let $q \geq 1$ and $\FCMI{K} = \big( J, L_i, 1 \leq i \leq q\big)$ a conditional partition of $\start{n}$.
  Let $\emptyset \neq I \subseteq \start{q}$. Then the equality
  \begin{equation*}
    A_I \coloneqq \bigcap_{i \in I} \set{X}_{L_i} \ \setminus \  \set{X}_{J \cup L \setminus L_I} = \bigg\lbrace \atom{W} \in \set{X}(n) \  \Big| \ \forall i \in I \  \exists \emptyset \neq W_i \subseteq L_i: W = W_I = \bigcup_{i \in I}W_i\bigg\rbrace
  \end{equation*}
  holds.
  Furthermore, if $q \geq 2$, then we have $\IM(\FCMI{K}) = \bigcup_{I \subseteq \start{q}: \ \num{I} \geq 2} A_I$.
\end{lemma}

\begin{proof}
  Let $\atom{W}$ be an atom, where by definition $\emptyset \neq W \subseteq \start{n} = J \cup L$.
  Then we have
  \begin{align*}
    \atom{W} \in A_I \quad & \Longleftrightarrow \quad  \forall i \in I \ \  \exists l \in L_i: p_W \in \set{X}_l \ \ \wedge \ \ \forall l \in J \cup L \setminus L_I: p_W \notin \set{X}_l \\
    & \Longleftrightarrow \quad \forall i \in I: \emptyset \neq W \cap L_i \ \ \wedge \ \ W \cap \big( J \cup L \setminus L_I\big) = \emptyset.
  \end{align*}
  Now, assume $\atom{W} \in A_I$ and let $W_i \coloneqq W \cap L_i \neq \emptyset$ for $i \in I$.
  Then from the preceding characterization, we obtain:
  \begin{equation*}
    W = W \cap (J \cup L) = W \cap L_I = W \cap \bigcup_{i \in I} L_i = \bigcup_{i \in I} W_i = W_I
  \end{equation*}
  and thus $\atom{W} = \atom{W_I}$ is in the set at the right-hand-side.

  If, vice versa, $\atom{W} = \atom{W_I}$ is of the stated form of the set at the right-hand-side, then $W \cap L_i = W_i \neq \emptyset$ for all $i \in I$ and $W \cap \big( J \cup L \setminus L_I\big) = \emptyset$, and thus $\atom{W} \in A_I$.

  The last statement follows immediately from the definition of $\IM(\FCMI{K})$ and what we just showed.
\end{proof}

Before we state the following proposition, we remind of the fact that, by Lemma~\ref{lem:how_atoms_look_like}, we have
\begin{equation*}
  \CR\big( X_{\start{n}}\big) = \set{\CR}\big( \set{X} \big) = \sum_{\atom{I} \in \set{X}} \set{\CR}(\atom{I}) = \sum_{\emptyset \neq I \subseteq \start{n}} X_{\start{n} \setminus I} . \CR\big( \bigscolon_{i \in I} X_i\big).
\end{equation*}

The following proposition generalizes this to the case that the left-hand-side is itself a ``term of higher degree'':

\begin{proposition}
  \label{pro:another_sum_formula}
  Let $q \geq 1$, and let $L_1, \dots, L_q \subseteq \start{n}$ be pairwise disjoint sets.
  For any $\emptyset \neq I \subseteq \start{q}$, we have
  \begin{equation*}
    \CR\big(  \bigscolon_{i \in I} X_{L_i}\big) = \sum_{\substack{(W_i)_{i \in I}: \\ \forall i: \ \emptyset \neq W_i \subseteq L_i}} X_{L_I \setminus W_I}. \CR\big( \bigscolon_{w \in W_I} X_{w}\big).\footnotemark
  \end{equation*}
  \footnotetext{For mitigating confusion, note that each $W_i$ is itself a set, and so each $w \in W_I$ is an element contained in exactly one of the $W_i$.}
\end{proposition}

\begin{proof}
  Assume without loss of generality that $I = \start{q}$ and $L_I = \start{n}$.
Then, applying Hu's Theorem~\ref{thm:hu_kuo_ting_generalized} and Lemma~\ref{lem:atoms_in_peculiar_set}, we obtain:
  \begin{equation*}
    \CR\big( \bigscolon_{i \in I} X_{L_i}\big) = \set{\CR} \Bigg( \bigcap_{i \in I}  \set{X}_{L_i} \Bigg)  = \sum_{\substack{\emptyset \neq W \subseteq L_I: \\ \atom{W} \in \bigcap_{i \in I}  \set{X}_{L_i}}} \set{\CR}(\atom{W}) = \sum_{\substack{(W_i)_{i \in I}: \\ \forall i: \ \emptyset \neq W_i \subseteq L_i}} \set{\CR}(\atom{W_I}).
  \end{equation*}
  The result follows from $\set{\CR}(\atom{W_I}) = X_{L_I \setminus W_I}.\CR\big( \bigscolon_{w \in W_I} X_w\big)$,
  see Lemma~\ref{lem:how_atoms_look_like}.
\end{proof}

The following Theorem generalizes Theorem 5 of~\cite{Yeung2002a}.
We illustrate it in Figures~\ref{fig:fcmi_visualization_1} and~\ref{fig:fcmi_visualization_2}.
The reader may also compare with Figures~\ref{fig:three_variables_mrfs} and~\ref{fig:subset_determination_application} from the introduction, which illustrates the effect of simple $\CR$-FCMIs for the case that $\CR = \Ent$ and general $\CR$, respectively.
In the first of those figures, the FCMIs come from the graph-structure of a Markov random field.

\begin{figure}
  \centering
  \includegraphics[width=0.6\textwidth]{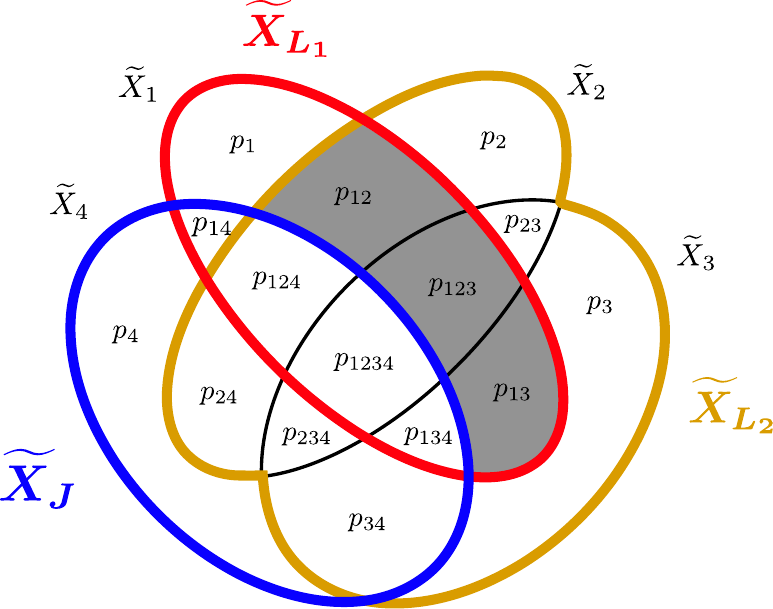}
  \caption{In this illustration, we visualize the full conditional mutual $\CR$-independence ($\CR$-FCMI) $\IndepF{\CR}{X_{L_1}}{X_{L_2}}{X_J}$ for the case $n = 4$, $J = \{4\}$, $L_1 = \{1\}$, and $L_2 = \{2, 3\}$.
  The independence is by Hu's Theorem~\ref{thm:hu_kuo_ting_generalized} given by $\set{\CR}\big(\set{X}_{L_1} \cap \set{X}_{L_2} \setminus \set{X}_J \big) = X_J.\CR(X_{L_1}; X_{L_2}) = 0$.
That is, $\set{\CR}$ vanishes on the gray region. 
  Subset determination Theorem~\ref{thm:subset_determination} results in $\set{\CR}$ vanishing even on all the atoms \emph{within} that region --- namely $\atom{12}, \atom{123}$, and $\atom{13}$.
Defining the conditional partition $\FCMI{K} = \big(J, L_1, L_2 \big)$, these atoms are precisely the elements in $\IM(\FCMI{K})$, thus confirming the characterization of $\CR$-FCMIs in terms of $\set{F}$ given in Theorem~\ref{thm:FCMIs_characterization}. }
  \label{fig:fcmi_visualization_1}
\end{figure}

\begin{theorem}
  \label{thm:FCMIs_characterization}
  Let $\mon{M}$ be a commutative, idempotent monoid acting additively on an abelian group $\gro{G}$ and $\CR: \mon{M} \to \gro{G}$ a function satisfying the chain rule Equation~\eqref{eq:cocycle_equationn}.
  We assume fixed elements $X_1, \dots, X_n \in \mon{M}$, giving rise to a $\gro{G}$-valued measure $\set{F}: 2^{\set{X}} \to \gro{G}$ according to Hu's Theorem~\ref{thm:hu_kuo_ting_generalized}.
  
  Let $\FCMI{K} = \big( J, L_i, 1 \leq i \leq q\big)$ be a conditional partition of $\start{n}$ with $q \geq 2$.
  Set $L \coloneqq L_{\start{q}}$.
  Then the following properties are equivalent:
  \begin{enumerate}
    \item $\FCMI{K}$ induces an $\CR$-FCMI with respect to $X_1, \dots, X_n$, i.e.\ $\bigindepF{\CR}_{i = 1}^{q} X_{L_i} \ | \ X_J$;
    \item for all $I \subseteq \start{q}$ with $\num{I} \geq 2$: $\big( X_{J \cup L \setminus L_I}\big).\CR\big( \bigscolon_{i \in I} X_{L_i}\big) = 0$;
    \item for all $I \subseteq \start{q}$ with $\num{I} \geq 2$ and all $(W_i)_{i \in I}$ with $\emptyset \neq W_i \subseteq L_i$ for all $i \in I$, we have $X_{J \cup L \setminus W_I}.\CR\big( \bigscolon_{w \in W_I} X_w\big) = 0$;
    \item $\set{\CR}(\atom{W}) = 0$ for all $\atom{W} \in \IM(\FCMI{K})$.
  \end{enumerate}
\end{theorem}

\begin{proof}  That 1 and 2 are equivalent follows immediately from the equivalence of properties 1 and 3 in Theorem~\ref{thm:charac_using_dual}.
  In doing so, $q$ replaces $n$, $X_J$ replaces $Y$, and $X_{L_i}$ replaces $X_i$.

  3 immediately implies 2 by Proposition~\ref{pro:another_sum_formula} and observing that $X_{J \cup L \setminus L_I} \cdot X_{L_I \setminus W_I} = X_{J \cup L \setminus W_I}$.
  3 and 4 are clearly seen to be equivalent using the definition of $\IM(\FCMI{K})$ and the fact that
  \begin{equation*}
    X_{J \cup L \setminus W_I}.\CR\big( \bigscolon_{w \in W_I} X_w\big) = X_{\start{n} \setminus W_I}.\CR\big( \bigscolon_{w \in W_I} X_w\big) = \set{\CR}(\atom{W_I}).
  \end{equation*}
  We used $J \cup L = \start{n}$ in the first step and Lemma~\ref{lem:how_atoms_look_like} in the second.

  Finally, we need to see that 2 implies 4.
  Let $\atom{W} \in \IM(K)$.
  Then by Lemma~\ref{lem:atoms_in_peculiar_set}, there is a set $I \subseteq \start{q}$ with $\num{I} \geq 2$ such that $\atom{W} \in   A_I = \bigcap_{i \in I} \set{X}_{L_i} \setminus  \set{X}_{J \cup L \setminus L_I}$.
  Then, using Hu's Theorem~\ref{thm:hu_kuo_ting_generalized} and property 2, we obtain $\set{\CR}\big(A_I\big) = \big( X_{J \cup L \setminus L_I}\big).\CR\big( \bigscolon_{i \in I} X_{L_i}\big) = 0$.
  By subset determination, Theorem~\ref{thm:subset_determination}, this results in particular in $\set{\CR}(\atom{W}) = 0$, and we are done.
\end{proof}

\begin{figure}
  \centering
  \includegraphics[width=\textwidth]{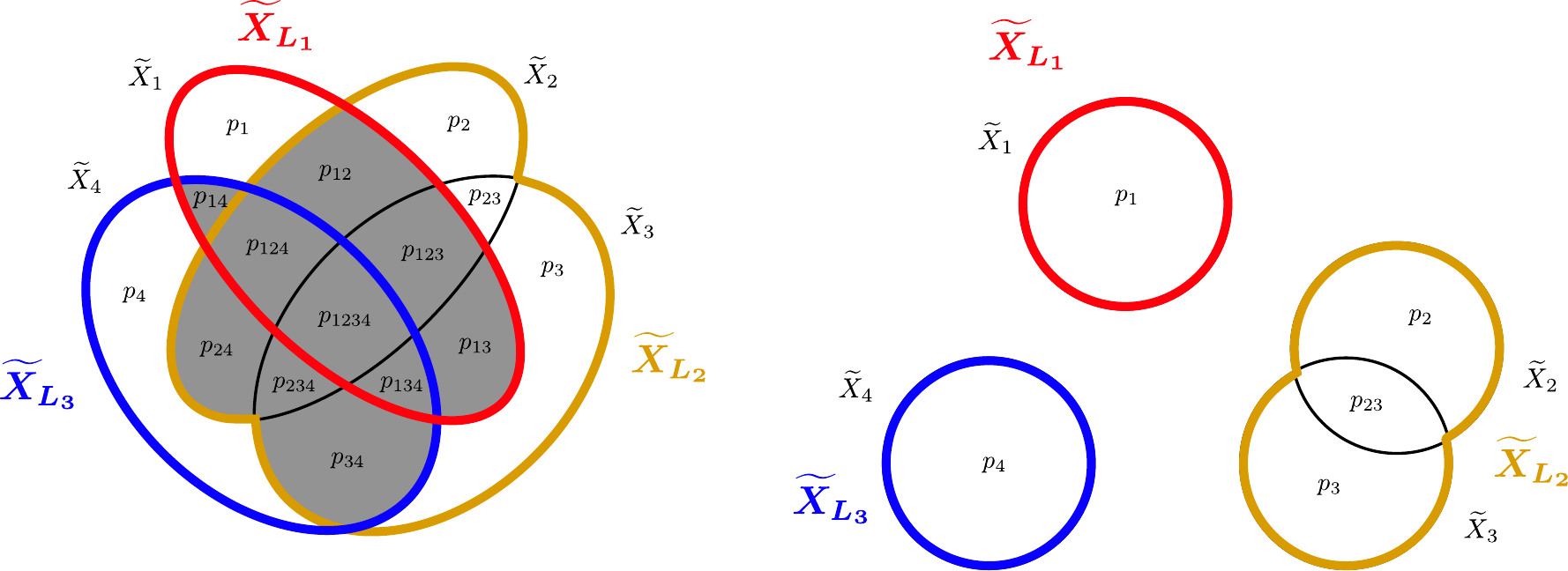}
  \caption{This figure visualizes the full conditional mutual $\CR$-independence ($\CR$-FCMI) of $X_{L_1}$, $X_{L_2}$ and $X_{L_3}$ for the case $n = 4$, $L_1 = \{1\}$, $L_2 = \{2, 3\}$, and $L_3 = \{4\}$.
  Here, the conditional variable is trivial, corresponding to $J = \emptyset$.
By the characterization of $\CR$-FCMIs, Theorem~\ref{thm:FCMIs_characterization}, $\set{\CR}$ vanishes on all atoms in the gray region.
  This leads to a degeneracy of the $\CR$-diagram, which can then be depicted on the small set of atoms at the right-hand-side of the figure.}
  \label{fig:fcmi_visualization_2}
\end{figure}

\begin{remark}\label{rem:remark_on_generality}
  If we replace property 1 in the preceding theorem by simply stating the independence relation --- without naming this an $\CR$-FCMI --- then properties 1, 2, and 3 are also equivalent without assuming $L \cup J = \start{n}$.
  This corresponds to conditional mutual $\CR$-independences that are not \emph{full}.  The reason is that $n$ does not even appear in those statements, and so we can always relabel the elements in $J \cup L$ to be equal to $\start{n'}$ for some $n' < n$.
  However, for the equivalence to statement $4$ we need the property $L \cup J = \start{n}$.
\end{remark}

\subsection{\texorpdfstring{$\CR$}--Markov Random Fields and \texorpdfstring{$\CR$}--Markov Chains}\label{sec:markov_random_fields}

In this section, we characterize $\CR$-Markov random fields with respect to a graph $\Gr{G}$ in terms of the $\CR$-diagram. 
Then we specialize this to a characterization of $\CR$-Markov chains.
These notions were defined in Terminology~\ref{ter:mrf_mc_f}.
As in the previous subsection, fix elements $X_1, \dots, X_n \in \mon{M}$, giving rise to a $\gro{G}$-valued measure $\set{F}: 2^{\set{X}} \to \gro{G}$ by Hu's Theorem~\ref{thm:hu_kuo_ting_generalized}.
Additionally, we now fix a graph $\Gr{G} = (\Ver{V}, \Ed{E})$ with vertex set $\Ver{V} = \start{n}$, see Definition~\ref{def:graph}.

\begin{lemma}
  \label{lem:atoms_image_disconnected}
  Let $\Ver{U}$ a cutset for $\Gr{G}$. 
  Let $\Ver{V}_i(\Ver{U})$, $1 \leq i \leq \ncomp{\Ver{U}}$ be the corresponding components and $\FCMI{K} \coloneqq \big( \Ver{U}, \Ver{V}_i(\Ver{U}), 1 \leq i \leq \ncomp{\Ver{U}} \big)$ the corresponding conditional partition.
  Then all $\atom{\Ver{W}} \in \IM(\FCMI{K})$ are disconnected.
\end{lemma}

\begin{proof}
  Let $\atom{\Ver{W}} \in \IM(\FCMI{K})$.
  There exists $I \subseteq \start{\ncomp{\Ver{U}}}$ with $\num{I} \geq 2$ and $\emptyset \neq \Ver{W}_i \subseteq \Ver{V}_i(\Ver{U})$ for all $i \in I$ such that 
  $\Ver{W} = \Ver{W}_I = \bigcup_{i \in I} \Ver{W}_i$.
  Now, let $i \neq j \in I$ and $w_i \in \Ver{W}_i \subseteq \Ver{V}_i(\Ver{U})$ and $w_j \in \Ver{W}_j \subseteq \Ver{V}_j(\Ver{U})$.
  Since $\Ver{V}_i(\Ver{U})$ and $\Ver{V}_j(\Ver{U})$ are components of $\Grdif{\Gr{G}}{\Ver{U}}$, there is no walk connecting $w_i$ and $w_j$ in $\Grdif{\Gr{G}}{\Ver{U}}$.
  Since $\Ver{U} \subseteq \Ver{V} \setminus \Ver{W}_I$, there can also be no such walk in $\Grdif{\Gr{G}}{(\Ver{V} \setminus \Ver{W}_I)}$.
  This shows that $\Ver{V} \setminus \Ver{W}_I$ is a cutset, and thus $\atom{\Ver{W}} = \atom{\Ver{W}_I}$ is disconnected.
\end{proof}

The following proof of the characterization of $\CR$-Markov random fields is similar to the one given in~\cite{Yeung2002a}, Theorem 8, for the special case of probabilistic Markov random fields.

\begin{proof}[Proof of Theorem~\ref{thm:mrf_characterization}]
  Assume that $X_1, \dots, X_n$ form an $\CR$-Markov random field with respect to $\Gr{G}$.
  Let $\atom{\Ver{W}}$ be a disconnected atom.
  We need to show $\set{\CR}(\atom{\Ver{W}}) = 0$.

  Define $\Ver{U} \coloneqq \Ver{V} \setminus \Ver{W}$.
  By assumption of $\Ver{W}$ being disconnected, $\Ver{U}$ is a cutset.
  Thus, $\ncomp{\Ver{U}} \geq 2$ and we have components $\Ver{V}_1(\Ver{U}), \dots, \Ver{V}_{\ncomp{\Ver{U}}}(\Ver{U})$ in $\Grdif{\Gr{G}}{\Ver{U}}$.
  Since $X_1, \dots, X_n$ form an $\CR$-Markov random field with respect to $\Gr{G}$, we obtain by Proposition~\ref{pro:equivalence_global_markov} the $\CR$-independence
  \begin{equation*}
    \bigindepF{\CR}_{i = 1}^{\ncomp{\Ver{U}}} X_{\Ver{V}_i(\Ver{U})} \ | \ X_{\Ver{U}}.
  \end{equation*}
  Now, notice that this is precisely the $\CR$-FCMI corresponding to the conditional partition $\FCMI{K} = \big(\Ver{U}, \Ver{V}_i(\Ver{U}), 1 \leq i \leq \ncomp{\Ver{U}}\big)$.
  By Theorem~\ref{thm:FCMIs_characterization}, $\set{\CR}(\atom{\Ver{W}'}) = 0$ for all atoms $\atom{\Ver{W}'} \in \IM(\FCMI{K})$.
  Now, notice that 
  $\bigcup_{i = 1}^{\ncomp{\Ver{U}}} \Ver{V}_i(\Ver{U}) = \Ver{V} \setminus \Ver{U} = \Ver{W}$, which, due to $\ncomp{\Ver{U}} \geq 2$ and $\Ver{V}_i(\Ver{U}) \neq \emptyset$, shows that $\atom{\Ver{W}} \in \IM(\FCMI{K})$.
  This shows $\set{\CR}(\atom{\Ver{W}}) = 0$, and we are done.

  For the other direction, assume that $\set{\CR}$ vanishes on all disconnected atoms.
  Now, let $\Ver{U} \subseteq \Ver{V}$ be a cutset, with components $\Ver{V}_1(\Ver{U}), \dots, \Ver{V}_{\ncomp{\Ver{U}}}(\Ver{U})$ of $\Grdif{\Gr{G}}{\Ver{U}}$.
  By Proposition~\ref{pro:equivalence_global_markov}, we need to show the $\CR$-independence
  \begin{equation*}
    \bigindepF{\CR}_{i = 1}^{\ncomp{\Ver{U}}}X_{\Ver{V}_i(\Ver{U})} \ | \ X_{\Ver{U}}.
  \end{equation*}
  This is the $\CR$-FCMI corresponding to the conditional partition $\FCMI{K} = \big(\Ver{U}, \Ver{V}_i(\Ver{U}\big), 1 \leq i \leq \ncomp{U})$.
  Thus, by Theorem~\ref{thm:FCMIs_characterization}, we need to show that $\set{\CR}(\atom{\Ver{W}}) = 0$ for all $\atom{\Ver{W}} \in \IM(\FCMI{K})$.
  Since all $\atom{\Ver{W}} \in \IM(\FCMI{K})$ are disconnected by Lemma~\ref{lem:atoms_image_disconnected}, this follows from assuming that $\set{\CR}$ vanishes on all disconnected atoms. 
\end{proof}

We refer back to Figure~\ref{fig:three_variables_mrfs} for a visualization of how $\CR$-diagrams of $\CR$-Markov random fields look like.
That visualization focused on $\Ent$-Markov random fields, but is entirely correct also in our general case, as the preceding theorem shows.

\begin{corollary}
  \label{cor:F_Markov_Chain_Charac}
  With all notation as above, the following two statements are equivalent:
  \begin{itemize}
    \item $X_1, \dots, X_n$ form an $\CR$-Markov chain.
    \item $\set{\CR}(\atom{\Ver{W}}) = 0$ for all $\Ver{W} \subseteq \start{n}$ that do \emph{not} only contain consecutive numbers.
  \end{itemize}
\end{corollary}

\begin{proof}
  For the proof, we specialize to the graph $\Gr{G} = \big(\Ver{V}, \Ed{E}\big)$ with $\Ver{V} = \start{n}$ and $\Ed{E} = \big\lbrace \{i, i+1\} \mid i = 1, \dots, n-1\big\rbrace$, so $\Gr{G}$ is intuitively a chain.
Then from Proposition~\ref{pro:characterization_markov_chain} we know that the first statement is equivalent to $X_1, \dots, X_n$ forming an $\CR$-Markov random field with respect to $\Gr{G}$.
Since the disconnected atoms $\atom{\Ver{W}}$ with respect to $\Gr{G}$ are precisely those where $\Ver{W}$ does not only consist of consecutive numbers, the result follows from Theorem~\ref{thm:mrf_characterization}.
\end{proof}

\begin{figure}
  \centering
  \begin{subfigure}[b]{0.43\textwidth}
    \centering
    \includegraphics[width=\textwidth]{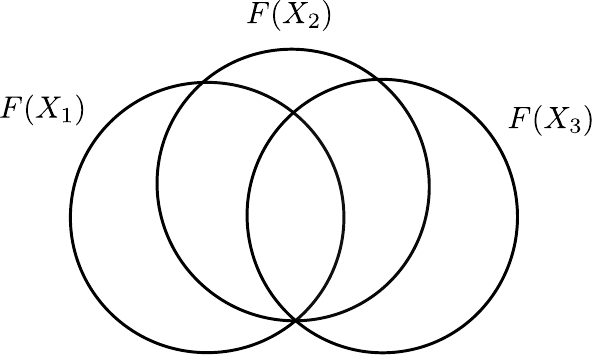}
  \end{subfigure}
  \hfill
  \begin{subfigure}[b]{0.47\textwidth}
    \centering
    \includegraphics[width=\textwidth]{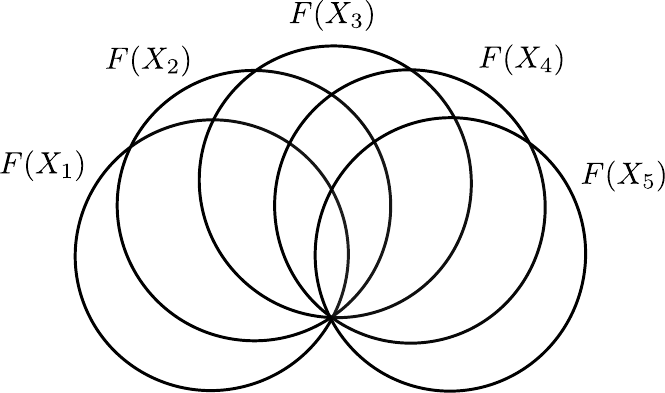}
  \end{subfigure}
  \caption{If elements $X_1, \dots, X_n \in \mon{M}$ form an $\CR$-Markov chain, then many atoms in the $\CR$-diagram disappear by Corollary~\ref{cor:F_Markov_Chain_Charac}. 
    The only atoms that remain are those corresponding to ``intervals'' in $\start{n}$. 
    This leads to a fan-like structure of the $\CR$-diagram, as visualized here for $n = 3$ and $n = 5$.}
  \label{fig:four_and_five_circles}
\end{figure}

The effect of the preceding corollary on $\CR$-diagrams is visualized in Figure~\ref{fig:four_and_five_circles}.

\begin{corollary} 
  \label{cor:information_terms_markov_chains}
  Assume $X_1, \dots, X_n$ form an $\CR$-Markov Chain.
  Then for all $I, J \subseteq \start{n}$ with $ \emptyset \neq I = \big\lbrace i_1 < i_2 < \dots < i_q \big\rbrace$, the following equality holds:
  \begin{equation*}
    X_J.\CR\big(X_{i_1};X_{i_2}; \dots ; X_{i_q}\big) = X_J.\CR\big(X_{i_1}; X_{i_q}\big).
  \end{equation*}
\end{corollary}

\begin{proof}
  We can without loss of generality assume $X_J = \one$.
  Using Hu's Theorem~\ref{thm:hu_kuo_ting_generalized}, we obtain
  \begin{align*}
    \CR(X_{i_1}; \dots ; X_{i_q}) & = \set{\CR}\Bigg( \bigcap_{k = 1}^{q} \set{X}_{i_k}\Bigg) = \sum_{L: \ i_1, \dots, i_q \in L} \set{\CR}(\atom{L}) 
    \overset{(\star)}{=} \sum_{L: \ i_1, i_q \in L} \set{\CR}(\atom{L}) \\
    & = \set{\CR}\Big( \set{X}_{i_1} \cap \set{X}_{i_q}\Big) 
     = \CR(X_{i_1}; X_{i_q}).
  \end{align*}
  In $(\star)$, we used Corollary~\ref{cor:F_Markov_Chain_Charac}: both sums are equal since for all summands where $L$ does not contain the whole ``interval'' $\range{i_1}{i_q}$, we have $\set{\CR}(\atom{L}) = 0$.
\end{proof}

\section{Probabilistic Independences and Markov Random Fields}\label{sec:specializing_to_probabilistic}

In this section, we specialize the results from Section~\ref{sec:yeung2002} to the probabilistic setting. 
Since we want to study probabilistic FCMIs and Markov random fields, we want to restrict our information measures --- e.g., Shannon entropy, Kullback-Leibler divergence, and cross-entropy --- to probability mass functions satisfying these properties. 
However, to be able to work in the setting of Section~\ref{sec:background} of a general function $\CR: \mon{M} \to \gro{G}$, we need a monoid action, which in the probabilistic case involves conditional probability mass functions.
We thus need to make sure that properties of distributions such as ``forms a Markov random field'' are preserved under conditioning.
We prove this for several properties in Section~\ref{sec:stability_conditioning}.

In Section~\ref{sec:functions_satisfying_chain_rule} we briefly study a general class of information functions that satisfy the chain rule, and motivate the case of higher-order cross-entropy and Kullback-Leibler divergence by a connection to the cluster cross-entropy from~\cite{Cocco2012}.
Building on Section~\ref{sec:stability_conditioning}, in Section~\ref{sec:restricting_measures}, we restrict general information functions to stable properties like ``forms a Markov random field'' and show that this preserves the monoid action and the chain rule. 
In Theorem~\ref{thm:get_mrf_for_free_with_restriction}, we then recover the notion of a Markov random field with respect to restricted information functions, providing examples for the $\CR$-Markov random fields studied in Section~\ref{sec:yeung2002}.
In Section~\ref{sec:new_kl_section} we then specialize to the case of Kullback-Leibler Markov chains with fixed transition probabilities between time-steps, reminiscent of the physical laws. 
This leads to Theorem~\ref{thm:second_law} and Corollary~\ref{cor:second_law_thermo}, where we interpret a weak version of the second law of thermodynamics in terms of a degeneracy of Kullback-Leibler diagrams.
Finally, in Section~\ref{sec:diffusion_models}, we use the Kullback-Leibler decomposition over Markov chains that follows from Theorem~\ref{thm:get_mrf_for_free_with_restriction} to derive the explicit decomposition of the evidence lower bound for diffusion models, demonstrating the applicability of our work to machine learning.

The results in~\cite{Yeung2002a} of $\Ent$-diagram characterizations of FCMIs and Markov random fields focus on \emph{fixed} probability mass functions, thus only looking at what we will call a \emph{slice} of the $\Ent$-diagram. 
This does not directly correspond to our supposed generalizations; after all, they always involve a monoid action, and thus let the probability mass functions unspecified in the probabilistic case.
In Appendix~\ref{sec:slices}, we demonstrate that this is not a problem by showing that our Theorems~\ref{thm:FCMIs_characterization} and~\ref{thm:mrf_characterization} actually imply the main results in~\cite{Yeung2002a}.
This validates the claim that our results are indeed generalizations.

In this whole section, let $\samp$ be a countable, discrete sample space, with the probability mass function \emph{not} yet fixed, and assume that all random variables $X: \samp \to \vs{X}$ have a finite discrete value space $\vs{X}$.
Also, recall the notions of the conditional independence $X \indep_P Y \ | \ Z$ of random variables on $\samp$ with respect to probability mass function $P$, and recall that equivalence classes of random variables form a monoid according to Proposition~\ref{pro:monoid_of_rvs}.
Conditional independence is well-defined at the level of equivalence classes, giving rise to a separoid (Proposition~\ref{pro:forms_a_separoid}), and we do not distinguish in notation between $X$ and its equivalence class.

\subsection{Stability under Conditioning}\label{sec:stability_conditioning}

In this subsection, we show that for random variables $X_1, \dots, X_n$, the property to be an FCMI, the property to be a Markov random field, and the property of two probability mass functions to have ``equal transition probabilities'' between time-steps in a Markov chain, are all \emph{stable under conditioning}.
That means that if they hold with respect to a probability mass function $P$, they also hold with respect to $P|_{Y = y}$ for arbitrary $Y = X_I$ with $I \subseteq \start{n}$ and $y \in \vs{Y}$.

We will later also need the following definition, for example for defining the Kullback-Leibler divergence:

\begin{definition}
  [Absolutely Continuous]
  \label{def:absolutely_continuous}
  Let $P, Q \in \Sim{\samp}$ be two probability mass functions.
  Then $P$ is said to be absolutely continuous with respect to $Q$, written $P \ll Q$, if the following implication holds for all $\omega \in \samp$:
  \begin{equation*}
       Q(\omega) = 0 \quad \Longrightarrow \quad P(\omega) = 0.
  \end{equation*}
\end{definition}

We now specialize the notions of conditional mutual independences and FCMIs from general separoids to the probabilistic case to make wordings and notation more efficient.
We did the similar specializations already for Markov random fields and Markov chains in Definitions~\ref{def:global_markov_property} and~\ref{def:markov_chain}.

\begin{terminology}\label{ter:probabilistic_terminology}
  Let $X_1, \dots, X_n, Y$ be random variables on $\samp$, $P \in \Delta(\samp)$ a probability mass function, and $\indep_P$ the corresponding separation relation.
  If $X_1, \dots, X_n$ are mutually independent given $Y$ with respect to $\indep_P$ (Definitions~\ref{def:mutual_independence}), then we write
  \begin{equation*}
    \bigindepF{P}_{i = 1}^{n} X_i \ | \  Y
  \end{equation*}
  and call this a conditional mutual $P$-independence.
  Let $\FCMI{K} = \big( J, L_i, 1 \leq i \leq q\big)$ be conditional partition of $\start{n}$ (Definition~\ref{def:conditional_partition}). 
  Then the $FCMI$ (Definition~\ref{def:fcmi})
  \begin{equation*}
    \bigindepF{P}_{i = 1}^{q} X_{L_i} \ | \ X_J
  \end{equation*}
  is called a full conditional mutual $P$-independence (with respect to $X_1, \dots, X_n$) --- or $P$-FCMI for short.
  If it holds, then we say $\FCMI{K}$ induces a $P$-FCMI.
\end{terminology}

We first need two lemmas:

\begin{lemma}
  \label{lem:double_conditioning}
  Let $X, Y, Z$ be three random variables on $\samp$ and $P \in \Delta(\samp)$ a probability mass function.
  If $P(y,z) \neq 0$, then
  \begin{equation*}
    \big(\cond{P}{Y=y}\big)(x \mid z) = P(x \mid y, z).
  \end{equation*}
\end{lemma}

\begin{proof}
  We have:
  \begin{align*}
    \big(\cond{P}{Y=y}\big)(x \mid z) = \Big( \cond{ \big( \cond{P}{Y=y}\big)}{Z = z} \Big)(x) 
     = \big( \cond{P}{YZ=(y,z)}\big)(x) 
     = P(x \mid y, z).
  \end{align*}
  The second step can be showed by explicit computation.
\end{proof}

\begin{lemma}
  \label{lem:independence_propagates}
  Let $U, V, W, Y$ be random variables on $\samp$ and $P \in \Delta(\samp)$ a probability mass function such that the following independence holds:
  \begin{equation*}
    \IndepF{P}{U}{V}{WY}.
  \end{equation*}
  Let $y \in \vs{Y}$ be arbitrary with $P(y) \neq 0$.
  Then the following independence follows:
  \begin{equation*}
    \IndepF{\cond{P}{Y=y}}{U}{V}{W}.
  \end{equation*}
\end{lemma}

\begin{proof}
  Let $u \in \vs{U}, v \in \vs{V}$ and $w \in \vs{W}$ be arbitrary.
  We want to show the following:
  \begin{equation}
    \big( \cond{P}{Y=y}\big)(u, v, w) = \big( \cond{P}{Y=y}\big)(u \mid w) \cdot \big( \cond{P}{Y=y}\big)(v, w).
    \label{eq:neww_equation_to_show}
  \end{equation}
  If $P(w \mid y) = 0$, then both sides of Equation~\eqref{eq:neww_equation_to_show} vanish and the result is clear.
  Thus, we can assume $P(w \mid y) \neq 0$, and therefore, together with $P(y) \neq 0$, we obtain $P(w, y) \neq 0$.
  By Lemma~\ref{lem:double_conditioning}, this results in $\big( \cond{P}{Y=y}\big)(u \mid w) = P(u \mid w, y)$.
  We obtain that the desired Equation~\eqref{eq:neww_equation_to_show} is equivalent to the following:
  \begin{equation*}
    \frac{P(u, v, w, y)}{P(y)} = P(u \mid w, y) \cdot \frac{P(v, w, y)}{P(y)},
  \end{equation*}
  which is equivalent to
  \begin{equation*}
    P(u, v, w, y) = P(u \mid w, y) \cdot P(v, w, y).
  \end{equation*}
  This follows from the assumption $\IndepF{P}{U}{V}{WY}$.
\end{proof}

\begin{proposition}
  \label{pro:fcmi_stable_conditioning}
  Let $X_1, \dots, X_n$ be random variables on $\samp$, $P \in \Delta(\samp)$ a probability mass function, and $\FCMI{K} = \big( J, L_i, 1 \leq i \leq q\big)$ a conditional partition of $\start{n}$ that induces a $P$-FCMI with respect to $X_1, \dots, X_n$.
  Then for all $I \subseteq \start{n}$, for $Y \coloneqq X_I$, and for all $y \in \vs{Y}$ with $P(y) \neq 0$, $\FCMI{K}$ also induces a $\big(P|_{Y = y}\big)$--FCMI with respect to $X_1, \dots, X_n$.
\end{proposition}

\begin{proof}
  By assumption, we have the $P$-FCMI $\bigindepF{P}_{i = 1}^{q} X_{L_i} \ | \ X_J$.
  This means that for all $i \in \{1, \dots, q\}$, we have
  \begin{equation}\label{eq:pairwise_induced_indep}
    \IndepF{P}{X_{L_i}}{X_{\bigcup_{l \neq i} L_l}}{X_J}.
  \end{equation}
  Now, decompose $I$ into three sets $I_i = I \cap L_i$, $I_{\setminus i} = I \cap \bigcup_{l \neq i} L_l$ and $I_J = I \cap J$.
  We obtain the corresponding random variables $Y_i = X_{I_i}$, $Y_{\setminus i} = X_{I_{\setminus i}}$ and $Y_J = X_{I_J}$.
  Since the monoid of random variables is idempotent, the above pairwise $P$-independence is equivalent to the following:
  \begin{equation*}
    \IndepF{P}{Y_{i} X_{L_i}}{Y_{\setminus i} X_{\bigcup_{l \neq i} L_l}}{X_J Y_J}.
  \end{equation*}
  Applying weak union (S4) to both the left and right factor, we obtain
  \begin{equation*}
    \IndepF{P}{X_{L_i}}{X_{\bigcup_{l \neq i} L_l}}{X_J \big(Y_iY_{\setminus i} Y_J\big)}.
  \end{equation*}
  Since $\FCMI{K}$ is \emph{full}, meaning the contained sets cover all of $\start{n}$, we obtain $I = I_i \cup I_{\setminus i} \cup I_J$ and consequently $Y_iY_{\setminus i} Y_J = Y$.
  Lemma~\ref{lem:independence_propagates} then shows the $\big(P|_{Y = y}\big)$--independence ${\IndepF{P|_{Y = y}}{X_{L_i}}{X_{\bigcup_{l \neq i}L_l}}{X_J}}$.
  Since $i$ was arbitrary, we obtain the $\big( P|_{Y = y} \big)$--FCMI $\bigindepF{P|_{Y = y}}_{i = 1}^{q} X_{L_i} \ | \ X_J$.
  That was to show.
\end{proof}

\begin{corollary}
  \label{cor:mrf_stable_conditioning}
  Let $X_1, \dots, X_n$ be random variables on $\samp$, $\Gr{G} = (\Ver{V}, \Ed{E})$ a graph with $\Ver{V} = \start{n}$, and $P \in \Delta(\samp)$ a probability mass function.
  Assume $X_1, \dots, X_n$ form a $P$-Markov random field with respect to $\Gr{G}$.

  Then for all $I \subseteq \start{n}$, for $Y \coloneqq X_I$, and for all $y \in \vs{Y}$ with $P(y) \neq 0$, $X_1, \dots, X_n$ also form a $\big(P|_{Y = y}\big)$--Markov random field with respect to $\Gr{G}$.
\end{corollary}

\begin{proof}
  We show this by proving the cutset property, which by Proposition~\ref{pro:equivalence_global_markov} is equivalent to the global Markov property.
  Thus, let $\Ver{U} \subseteq \start{n}$ be a cutset and $\Ver{V}_1(\Ver{U})$, \dots, $\Ver{V}_{\ncomp{\Ver{U}}}(\Ver{U})$ the corresponding components.
  By assumption we know that the conditional partition $\FCMI{K} = \big( \Ver{U}, \Ver{V}_{i}(\Ver{U}), 1 \leq i \leq \ncomp{\Ver{U}}\big)$ induces a $P$-FCMI.
  Then Proposition~\ref{pro:fcmi_stable_conditioning} shows that we also obtain a $\big(P|_{Y = y}\big)$--FCMI, meaning we have ${\bigindepF{P|_{Y = y}}_{i = 1}^{\ncomp{\Ver{U}}} X_{\Ver{V}_i(\Ver{U})} \ | \ X_{\Ver{U}}}$.
  This was to show.
\end{proof}

Next, we show that the property to have ``equal transition probabilities'' between time-steps in a Markov chain is stable under conditioning. 
We will use this in Section~\ref{sec:new_kl_section} to illustrate a weak version of the second law of thermodynamics.
First, we state a simple lemma about general separoids that we will use:

\begin{lemma}[\cite{Dawid2001}, Lemma 1.1]
  \label{lem:equivalence_independence}
  Let $(\mon{M}, \indep)$ be a separoid.
  For $X, Y, Z \in \mon{M}$, we have the following equivalence:
  \begin{equation*}
    \Indep{X}{Y}{Z} \quad \Longleftrightarrow \quad \Indep{XZ}{YZ}{Z}.
  \end{equation*}
\end{lemma}

\begin{proposition}
  \label{pro:conditionals_still_same_propagates}
  Let $X_1, \dots, X_n$ be random variables on $\samp$ and $(P, Q) \in \Sim{\samp}^2$ be two probability mass functions such that $P$ is absolutely continuous with respect to $Q$, see Definition~\ref{def:absolutely_continuous}.
  Assume that $X_1, \dots, X_n$ form a $P$-Markov chain \emph{and} $Q$-Markov chain.
  Additionally, assume that $P$ and $Q$ have the same transition probabilities, i.e.: $P(x_i \mid x_{i-1}) = Q(x_i \mid x_{i-1})$ for all $i$ and all $x_{i-1}, x_i$ with $P(x_{i-1}) \neq 0$.

  Then for all $I \subseteq \start{n}$, for $Y \coloneqq X_I$, and for all $y \in \vs{Y}$ with $P(y) \neq 0$, $\cond{P}{Y = y}$ is also absolutely continuous with respect to $\cond{Q}{Y = y}$ and $X_1, \dots, X_n$ also forms a $\big(\cond{P}{Y = y}\big)$--Markov chain and a $\big( \cond{Q}{Y = y} \big)$--Markov chain.
  Additionally, $\cond{P}{Y = y}$ and $\cond{Q}{Y = y}$ also have the same transition probabilities, i.e.: $(\cond{P}{Y=y})(x_i \mid x_{i-1}) = (\cond{Q}{Y=y})(x_i \mid x_{i-1})$ for all $i$ and all $x_{i-1}, x_i$ with $(\cond{P}{Y=y})(x_{i-1}) \neq 0$.
\end{proposition}

\begin{proof}
  That $\cond{P}{Y = y}$ is absolutely continuous with respect to $\cond{Q}{Y = y}$ is clear. 
  That $X_1, \dots, X_n$ also forms a $\big( \cond{P}{Y = y} \big)$--Markov chain and a $\big( \cond{Q}{Y = y} \big)$--Markov chain follows from the fact that Markov chains are Markov random fields by Proposition~\ref{pro:characterization_markov_chain}, and from Corollary~\ref{cor:mrf_stable_conditioning}.

  For the second statement, let $i$ be given.
  Write $I = I_- \cup I_+$ with $I_- \subseteq \start{i-1}$, $I_+ \subseteq \range{i}{n}$.
  Correspondingly, we write $Y = Y_- Y_+$. 
  Remember that being a Markov chain implies being a Markov random field according to Proposition~\ref{pro:characterization_markov_chain}.
  Together with Lemma~\ref{lem:equivalence_independence} and the decomposition separoid axiom (S3), one can show the following two independences:
  \begin{equation*}
    \IndepF{P}{X_i}{Y_-}{X_{i-1}Y_+}, \quad  \quad\IndepF{Q}{X_i}{Y_-}{X_{i-1}Y_{+}}.\footnotemark
  \end{equation*}
  \footnotetext{More precisely, the global Markov property ensures $X_i \indep_P X_{[i-2]} \ | \ X_{i-1}Y_+$ (Note that this is trivial if $X_i$ is part of $Y_+$, and otherwise we have the disjointness stated in the global Markov property). Then by Lemma~\ref{lem:equivalence_independence}, we can add the variables $X_{i-1}Y_+$ to the left-hand-sides. The decomposition separoid axiom (S3) then allows to remove the variables that we do not need for the final result.}
  It follows
  \begingroup
  \allowdisplaybreaks
  \begin{align*}
    \big(\cond{P}{Y = y }\big)(x_i \mid x_{i-1}) & \overset{(\star)}{=} P(x_i \mid x_{i-1}, y)  & \big( \text{Lemma~\ref{lem:double_conditioning}}\big) \\
    & = P(x_i \mid x_{i-1}, y_-, y_+) \\
    & = P(x_i \mid x_{i-1}, y_+) \\
    & = \frac{P(x_i, y_+ \mid x_{i-1})}{P(y_+ \mid x_{i-1})}  \\
    & = \frac{P(x_{i} \mid x_{i-1}) \cdot P(y_+ \mid x_{i-1}, x_i)}{\sum_{x_{i}'} P(x_{i}' \mid x_{i-1}) \cdot P(y_+ \mid x_{i-1}, x_{i}')}, 
  \end{align*}
  \endgroup
  In $(\star)$, to apply the lemma, we used $P(y) \neq 0$ and the additional assumption that $\big( P|_{Y = y}\big)(x_{i-1}) \neq 0$, which results in $P(y, x_{i-1}) \neq 0$.
  Due to $P$ being absolutely continuous with respect to $Q$, we also obtain $Q(y, x_{i-1}) \neq 0$.
  Thus, we similarly get:
  \begin{equation*}
    \big(\cond{Q}{Y=y}\big)(x_{i} \mid x_{i-1}) = \frac{Q(x_{i} \mid x_{i-1}) \cdot Q(y_+ \mid x_{i-1}, x_i)}{\sum_{x_{i}'} Q(x_{i}' \mid x_{i-1}) \cdot Q(y_+ \mid x_{i-1}, x_{i}')}.
  \end{equation*} 
  By assumption, we already know $P(x_i \mid x_{i-1}) = Q(x_i \mid x_{i-1})$.
  The preceding computations therefore show that it is enough to prove that $P(y_+ \mid x_{i-1}, x_{i}) = Q(y_+ \mid x_{i-1}, x_i)$.
  We can assume that $I_+ \subseteq \range{i+1}{n}$ since otherwise $x_i$ determines the first entry of $y_+$, which can then be removed from the expression if it is equal to $x_i$.
  Then, let $J \subseteq \range{i+1}{n}$ be such that $J \dot\cup I_+ = \range{i+1}{n}$.
  We obtain:
  \begingroup
  \allowdisplaybreaks
  \begin{align*}
    P(y_+ \mid x_{i-1}, x_i) & = \sum_{x_J} P\big(x_{\range{i+1}{n}} \mid x_{i-1}, x_i\big) \\
    & = \sum_{x_{J}}\prod_{j = i}^{n-1} P(x_{j+1} \mid x_j) & \big( \text{Markov chain property} \big) \\
    & = \sum_{x_{J}} \prod_{j = i}^{n-1} Q(x_{j+1} \mid x_j) \\
    & = \sum_{x_J} Q\big( x_{\range{i+1}{n}} \mid x_{i-1}, x_i \big) & \big(\text{Markov chain property}\big) \\
    & = Q(y_{+} \mid x_{i-1}, x_i).
  \end{align*}
  \endgroup
  This was to show.
\end{proof}

\subsection{Information Functions Satisfying the Chain Rule}\label{sec:functions_satisfying_chain_rule}

Before we show in the next subsection that information functions can always be restricted to stable probability sets, we first define a class of information functions to work with.
This includes Shannon entropy, Kullback-Leibler divergence, and cross-entropy.
We also connect higher-order cross-entropy and Kullback-Leibler divergence to the notion of cluster entropy defined in~\cite{Cocco2012}.

To capture many information functions, we work in the following general setting very similar to~\cite{Baudot2015a}, Section 5.
Let $\mon{M}$ be the monoid of equivalence classes of random variables on the fixed sample space $\samp$, and let $r \geq 0$ be a natural number. 
We denote tuples of probability mass functions in $\Sim{\samp}^{r+1}$ by $\big( 
P \| Q_1, \dots, Q_r\big) \coloneqq \big( P, Q_1, \dots, Q_r \big)$ to remind of the notation inside the Kullback-Leibler divergence.
Then, we set
\begin{equation}\label{eq:absolutely_continuous_restriction}
  \AC{\Sim{\samp}^{r+1}} \coloneqq \Big\lbrace  \big( P \| Q_1, \dots, Q_r \big) \in \Sim{\samp}^{r+1} \ \big| \  P \ll Q_1, \dots, P \ll Q_r  \Big\rbrace.
\end{equation}
Here, $P \ll Q$, means that $P$ is absolutely continuous with respect to $Q$, see Definition~\ref{def:absolutely_continuous}.
In applications, $r = 0$ would be used for Shannon entropy and $\alpha$-entropy, and $r = 1$ for Kullback-Leibler divergence, $\alpha$--Kullback-Leibler divergence and cross-entropy, indicating the number of additional probability mass functions.
We do not have applications for $r \geq 2$ in mind.

Let 
\begin{equation}\label{eq:my_ab_group}
  \gro{G} \coloneqq \Meas\Big( \AC{\Sim{\samp}^{r+1}}, \R \Big) \coloneqq \Big\lbrace f: \AC{\Sim{\samp}^{r+1}} \to \R \ \big| \  f \mathrm{\ is \ measurable } \Big\rbrace,
\end{equation}
which clearly is an abelian group.
We also assume an additive monoid action of $\mon{M}$ on $\gro{G}$ given by the formula
 \begin{equation}\label{eq:general_conditioning_formula}
   \big[ X.f \big]\big( P \| Q_1, \dots, Q_r\big) = \sum_{x \in \vs{X}} g\big( P(x), Q_1(x), \dots, Q_r(x)\big) \cdot f\big( P|_{X = x} \| Q_1|_{X = x}, \dots, Q_r|_{X = x} \big)
\end{equation}
for some measurable function $g: [0, 1]^{r+1} \to \R$ with $g(0, \star, \dots, \star) = 0$.\footnote{The condition $g(0, \star, \dots, \star) = 0$ is satisfied in all examples we consider and ensures that we have a zero coefficient whenever the conditionals are not defined. We then define the product as simply zero. 
Additionally, note that due to absolute continuity, we have $P(x) = 0$ whenever $Q_i(x) = 0$ for some $i$. Therefore, it is enough to have the vanishing condition for $g$ in the first component only.}
We assume $g$ to be the same function irrespective of $X, f$, and $P, Q_1, \dots, Q_r$. 
Note that $P \ll Q_i$ implies that also $\cond{P}{X = x} \ll \cond{Q_i}{X = x}$, which means that the preceding formula is actually defined.

Furthermore, we assume to have a function $\CR: \mon{M} \to \gro{G}$
that satisfies the chain rule, meaning that for all $X, Y \in \mon{M}$, we have
\begin{equation*}
  \CR(XY) = \CR(X) + X.\CR(Y).
\end{equation*}
The reader can verify that all of Shannon entropy, Kullback-Leibler divergence, $\alpha$-entropy, $\alpha$--Kullback-Leibler divergence, and cross-entropy, fit precisely under this umbrella.
Precise definitions that fit our exact framework can be found in~\cite{Lang2022}, Sections 2 and 5.
For example, for Shannon entropy we have $r = 0$ and $\gro{G} = \Meas\big( \Sim{\samp}, \R \big)$.
The action is given by
\begin{equation*}
  \big[ X.f \big] (P) \coloneqq \sum_{x \in \vs{X}} P(x) \cdot f\big( P|_{X = x}\big),
\end{equation*}
and the Shannon entropy $\Ent: \mon{M} \to \gro{G}$ is given by
\begin{equation}\label{eq:entropy_definition}
  \big[\Ent(X)\big](P) \coloneqq - \sum_{x \in \vs{X}} P(x) \cdot \log P(x),
\end{equation}
which indeed satisfies the chain rule $\Ent(XY) = \Ent(X) + X.\Ent(Y)$; see Section~\ref{sec:ent_mi_ii}.
Note that $0 \cdot \log 0 \coloneqq 0$ in this formula.
Similarly, for Kullback-Leibler divergence we have $r = 1$ and $\gro{G} = \Meas\Big( \AC{\Sim{\samp}^2}, \R \Big)$.
The action is given by 
\begin{equation}\label{eq:monoid_action_KL}
  \big[ X.f \big] (P \| Q) \coloneqq \sum_{x \in \vs{X}} P(x) \cdot f\big( P|_{X = x} \| Q|_{X = x} \big).
\end{equation}
The Kullback-Leibler divergence $\KL: \mon{M} \to \gro{G}$ is then defined by
\begin{equation}\label{eq:kl_divergence}
  \big[\KL(X)\big](P \| Q) \coloneqq \sum_{x \in \vs{X}} P(x) \cdot \log \frac{P(x)}{Q(x)}.
\end{equation}
This also satisfies the chain rule: $\KL(XY) = \KL(X) + X.\KL(Y)$.
Note that $0 \cdot \log \frac{0}{Q(x)} \coloneqq 0$ in this formula.
Similarly, the cross-entropy $\CE: \mon{M} \to \gro{G}$ is defined by
\begin{equation}\label{eq:cross_entropy}
  \big[ \CE(X) \big](P \| Q) \coloneqq \sum_{x \in \vs{X}} P(x) \cdot \log \frac{1}{Q(x)}.
\end{equation}
Clearly, one then has $\CE(X) = I(X) + \KL(X)$.

Once these definitions are set in place, one can study higher-order terms $\CR: \mon{M}^q \to \gro{G}$ as in Equation~\eqref{eq:inductive_definition}, and visualize them in $\CR$-diagrams.
The higher-order terms for Shannon entropy $I$, i.e., the mutual information and interaction information, are frequently studied, whereas the higher-order terms for $\CE$ and $\KL$ have gotten substantially less attention.
So before we continue to study the $\CR$-diagrams of such functions for Markov random fields in the coming section, we briefly connect these higher-order terms to the cluster entropy from~\cite{Cocco2012}.

Namely,~\cite{Cocco2012} define the cluster (cross-)entropy\footnote{The authors often write ``entropy'' when they actually mean cross-entropy.} $\Delta \CE$ for subsets of $n$ random variables $X_1, \dots, X_n$ by
\begin{equation}\label{eq:cluster_entropy}
  \Delta \CE \big( \bigscolon_{j \in J} X_j \big) \coloneqq \sum_{K \subseteq J} (-1)^{|K| - |J|} \cdot \CE(X_K).
\end{equation}
In contrast, we define higher-order cross-entropies as in Equation~\eqref{eq:inductive_definition} by the inductive formula
\begin{equation*}
  \CE\big( X_{j_1} ; \dots ; X_{j_q} \big) = \CE\big( X_{j_1}; \dots ; X_{j_{q-1}} \big) - X_{j_q}.\CE\big( X_{j_1}; \dots; X_{j_{q-1}} \big).
\end{equation*}
In~\cite{Lang2022}, Corollary 3.5, part 6, it is deduced from the generalized Hu theorem that this obeys the following inclusion-exclusion type formula:
\begin{equation}\label{eq:formula_from_old_paper}
  \CE\big( \bigscolon_{j \in J} X_j \big) = \sum_{K \subseteq J} (-1)^{|K|+1} \cdot \CE(X_K).
\end{equation}
Comparing Equations~\eqref{eq:cluster_entropy} and~\eqref{eq:formula_from_old_paper}, we see that our higher-order cross-entropies and the cluster cross-entropies differ by a simple sign that depends on the order:
\begin{equation*}
  \CE\big( \bigscolon_{j \in J} X_j \big) = (-1)^{|J| + 1} \cdot \Delta \CE \big( \bigscolon_{j \in J} X_j \big).
\end{equation*}
Higher-order Kullback-Leibler divergence is then fully determined by higher-order cross-entropy and information, as $\KL\big( \bigscolon_{j \in J} X_j \big) = \CE\big( \bigscolon_{j \in J} X_j \big) - I \big( \bigscolon_{j \in J} X_j \big)$.
We hope that these connections provide motivation for the reader to study information diagrams for functions such as cross-entropy and Kullback-Leibler divergence, and that the study of these diagrams, in turn, could help to better understand concepts like the cluster cross-entropies.

\subsection{Restricting Information Measures to Stable Probability Sets}\label{sec:restricting_measures}

In this subsection, we show that information functions that ``satisfy the chain rule'' can, under mild conditions, always be restricted to subsets of probability mass functions as long as those are stable.
We define this notion in Definition~\ref{def:stable}, which encompasses the examples we saw in the previous subsection. 
In Theorem~\ref{thm:get_mrf_for_free_with_restriction}, we will then see that restricting to probability mass functions that give rise to Markov random fields leads to the formalism of an $\CR$-Markov random field, as studied in Section~\ref{sec:markov_random_fields}.
Let the notation throughout be as in the preceding subsection.

\begin{definition}
  [Stable Property]
  \label{def:stable}
  Let $n \geq 0$, $X_1, \dots, X_n \in \mon{M}$, $r \geq 0$, and $\Prty{R} = \Prty{R}(X_1, \dots, X_n)$ be a property of probability tuples with $r+1$ entries, which means that for all $\big( P \| Q_1, \dots, Q_r\big) \in \Sim{\samp}^{r+1}$, $\Prty{R}\big( P \| Q_1, \dots, Q_r\big)$ is either true or false.
  $\Prty{R}$ is called \emph{stable} if the following holds:
  \begin{itemize}
    \item $\Prty{R}$ is \emph{well-defined}, i.e., if $Y_1, \dots, Y_n$ are random variables on $\samp$ with $Y_i \sim X_i$, then
      \begin{equation*}
	\Prty{R}(X_1, \dots, X_n) = \Prty{R}(Y_1, \dots, Y_n);
      \end{equation*}
    \item $\Prty{R}$ is \emph{measurable}, i.e., the set
      \begin{equation*}
	\Big\lbrace \big( P \| Q_1, \dots, Q_r\big) \in \Sim{\samp}^{r+1} \ \big| \ \Prty{R}\big( P \| Q_1, \dots, Q_r\big) \Big\rbrace 
      \end{equation*}
      is a measurable subset of $\Sim{\samp}^{r+1}$;
    \item $\Prty{R}$ is \emph{stable under conditioning with respect to $X_1, \dots, X_n$}, i.e.: let $\mon{M}_{\Prty{R}}$ be the submonoid of $\mon{M}$ generated by $X_1, \dots, X_n$.\footnote{I.e., this submonoid consists of all $X_I$ for $I \subseteq [n]$.}
      Then we assume that for all $\big(P \| Q_1, \dots, Q_r\big) \in \Sim{\samp}^{r+1}$, all $Y \in \mon{M}_{\Prty{R}}$, and all $y \in \vs{Y}$ with $P(y) \neq 0$, following implication holds:
\begin{equation*}
  \Prty{R}\big( P \| Q_1, \dots, Q_r \big) \quad \Longrightarrow \quad \Prty{R}\big(P|_{Y = y} \| Q_1|_{Y = y}, \dots, Q_r|_{Y = y}\big).
\end{equation*}
  \end{itemize}
\end{definition}

Now, let $X_1, \dots, X_n \in \mon{M}$, and let $\Prty{R} = \Prty{R}(X_1, \dots, X_n)$ be a stable property of probability tuples with $r+1$ elements.
Define
\begin{equation*}
  \WP{\big(\AC{\Sim{\samp}^{r+1}}\big)}{\Prty{R}} \coloneqq \Big\lbrace \big( P \| Q_1, \dots, Q_r \big) \in \AC{\Sim{\samp}^{r+1}} \ \big| \ \Prty{R}\big( P \| Q_1, \dots, Q_r \big) \Big\rbrace.
\end{equation*}
Then, define
\begin{equation}\label{eq:new_group_definition}
  \WP{\gro{G}}{\Prty{R}} \coloneqq \Meas\Big( \WP{\big(\AC{\Sim{\samp}^{r+1}}\big)}{\Prty{R}} , \R \Big)
\end{equation}
and the action $\WP{.}{\Prty{R}}: \WP{\mon{M}}{\Prty{R}} \times \WP{\gro{G}}{\Prty{R}} \to \WP{\gro{G}}{\Prty{R}}$ by
\begin{equation}\label{eq:new_action}
  \big[ X \WP{.}{\Prty{R}} f \big]\big( P \| Q_1, \dots, Q_r \big) \coloneqq \big[ X.\tilde{f}\big]\big( P \| Q_1, \dots, Q_r \big),
\end{equation}
where $\tilde{f}$ is any measurable extension of $f$ to all of $\AC{\Sim{\samp}^{r+1}}$.
Note that this means that $X.\tilde{f}$ is a measurable extension of $X\WP{.}{\Prty{R}} f$, and that $X \WP{.}{\Prty{R}} f$ is thus again measurable.

\begin{lemma}
  \label{lem:well-defined_additive_action}
  The function $\WP{.}{\Prty{R}}: \WP{\mon{M}}{\Prty{R}} \times \WP{\gro{G}}{\Prty{R}} \to \WP{\gro{G}}{\Prty{R}}$, as defined above, is a well-defined additive monoid action.
\end{lemma}

\begin{proof}
  For well-definedness, one first needs to check that for all measurable $f: \WP{\big(\AC{\Sim{\samp}^{r+1}}\big)}{\Prty{R}} \to \R$, a measurable extension $\tilde{f}$ to all of $\AC{\Sim{\samp}^{r+1}}$ exists:
  note that from the fact that $\Prty{R}$ is stable and thus measurable, we immediately obtain that $\WP{\big(\AC{\Sim{\samp}^{r+1}}\big)}{\Prty{R}}$ is a measurable subset of $\AC{\Sim{\samp}^{r+1}}$. 
  Then $\tilde{f}$ can, for example, be constructed by mapping all elements in $\AC{\Sim{\samp}^{r+1}} \setminus \WP{\big(\AC{\Sim{\samp}^{r+1}}\big)}{\Prty{R}}$ to zero, which is clearly measurable.

  Furthermore, one needs to check that the definition is independent of the extension $\tilde{f}$ and the representative of the \emph{equivalence class} $X$.
  The first requirement follows from Equation~\eqref{eq:general_conditioning_formula}: $g$ was assumed to be independent of $\tilde{f}$, and $\tilde{f}$ is only applied to elements on which already $f$ is uniquely defined, due to the assumption of $\Prty{R}$ being stable under conditioning.
  That the formula is independent of the representative of $X$ then follows since the original action was well-defined.

  That $\WP{.}{\Prty{R}}$ is additive and $\one \in \WP{\mon{M}}{\Prty{R}}$ acts neutrally is obvious.
  We now show that the action is associative:
  \begin{align*}
    \Big[ (XY)\WP{.}{\Prty{R}} f \Big]\big( P \| Q_1, \dots, Q_r\big) & = 
    \Big[ (XY). \tilde{f} \Big]\big( P \| Q_1, \dots, Q_r\big) \\
    & = \Big[ X.\big( Y. \tilde{f}) \Big]\big( P \| Q_1, \dots, Q_r\big) \\
    & \overset{(\star)}{=} \Big[ X\WP{.}{\Prty{R}}\big( Y \WP{.}{\Prty{R}} f) \Big]\big( P \| Q_1, \dots, Q_r\big) .
  \end{align*}
  In step $(\star)$, we used that $Y.\tilde{f}$ is by definition a measurable extension of $Y\WP{.}{\Prty{R}} f$, which by definition of the ``outer'' action $\WP{.}{\Prty{R}}$ gives the result.
  This finishes the proof.
\end{proof}

Finally, let $\CR: \mon{M} \to \gro{G}$ be a function satisfying the chain rule as in Section~\ref{sec:functions_satisfying_chain_rule} and define
\begin{equation}\label{eq:new_F_def}
\WP{\CR}{\Prty{R}}: \WP{\mon{M}}{\Prty{R}} \to \WP{\gro{G}}{\Prty{R}}, \quad \big[\WP{\CR}{\Prty{R}}(X)\big](P \| Q_1, \dots, Q_r) \coloneqq \big[F(X)\big](P \| Q_1, \dots, Q_r).
\end{equation}
In other words, $\WP{\CR}{\Prty{R}}(X)$ is just the restriction of $\CR(X)$ to the tuples of probability mass functions satisfying property $\Prty{R}$.

\begin{proposition}
  \label{pro:property_restriction_chain_rule}
  Let  $r \geq 0$ and $\gro{G}$ be the abelian group defined in Equation~\eqref{eq:my_ab_group}.
  Assume $\mon{M}$, the monoid of equivalence classes of random variables on $\samp$, acts additively on $\gro{G}$ by Equation~\eqref{eq:general_conditioning_formula}, and that $\CR: \mon{M} \to \gro{G}$ is a function satisfying the chain rule Equation~\eqref{eq:cocycle_equationn}.

  Now, let $X_1, \dots, X_n$ be random variables on $\samp$ and $\Prty{R} = \Prty{R}(X_1, \dots, X_n)$ a stable property of tuples of $r+1$ probability mass functions.
  Let $\WP{\mon{M}}{\Prty{R}}$ be the submonoid generated by $X_1, \dots, X_n$ and $\WP{\gro{G}}{\Prty{R}}$ be defined as in Equation~\eqref{eq:new_group_definition}.
  Let $\WP{\mon{M}}{\Prty{R}}$ act on $\WP{\gro{G}}{\Prty{R}}$ as in Equation~\eqref{eq:new_action}.
  Then the restricted information function $\WP{\CR}{\Prty{R}}: \WP{\mon{M}}{\Prty{R}} \to \WP{\gro{G}}{\Prty{R}}$, defined in Equation~\eqref{eq:new_F_def}, satisfies the chain rule:
  \begin{equation*}
    \WP{\CR}{\Prty{R}}(XY) = \WP{\CR}{\Prty{R}}(X) + X\WP{.}{\Prty{R}}\WP{\CR}{\Prty{R}}(Y).
  \end{equation*}
\end{proposition}

\begin{proof}
  We have
  \begin{align*}
    \big[\WP{\CR}{\Prty{R}}(XY)\big] \big(P \| Q_1, \dots, Q_r \big) & = \big[\CR(XY) \big]\big( P \| Q_1, \dots, Q_r \big) \\
    & = \big[ \CR(X)\big]\big( P \| Q_1, \dots, Q_r \big) + \big[ X.\CR(Y) \big]\big(P \| Q_1, \dots, Q_r  \big) \\
    & \overset{(\star)}{=} \big[ \WP{\CR}{\Prty{R}}(X) \big]\big(P \| Q_1, \dots, Q_r \big) + \big[X\WP{.}{\Prty{R}} \WP{\CR}{\Prty{R}}(Y) \big] \big( P \| Q_1, \dots, Q_r\big) \\
    & = \big[ \WP{\CR}{\Prty{R}}(X) +  X\WP{.}{\Prty{R}} \WP{\CR}{\Prty{R}}(Y) \big] \big( P \| Q_1, \dots, Q_r\big).
  \end{align*}
  In step $(\star)$, in the right-hand term, we used that $\CR(Y)$ extends $\WP{\CR}{\Prty{R}}(Y)$, which by definition of the action $\WP{.}{\Prty{R}}$ results in $X.\CR(Y)$ extending $X\WP{.}{\Prty{R}} \WP{\CR}{\Prty{R}}(Y)$.
\end{proof}

The preceding proposition ensures that we can apply Hu's Theorem~\ref{thm:hu_kuo_ting_generalized} to $\WP{\CR}{\Prty{R}}$, which we will make use of in the next section on Kullback-Leibler Markov chains.

In the next theorem, we restrict this setting somewhat further to obtain a result on $\WP{\CR}{\Prty{R}}$-Markov random fields that will apply to Shannon entropy, Kullback-Leibler divergence and cross-entropy.
The result will, however, as stated not apply to $\alpha$-entropy and $\alpha$--Kullback-Leibler divergence.

For the theorem, we make the stronger assumption that the action of $\mon{M}$ on $\gro{G}$ is given by
\begin{equation}\label{eq:stronger_action}
  \big[X.f\big]\big(P \| Q_1, \dots, Q_r\big) = \sum_{x \in \vs{X}} P(x) \cdot f\big( P|_{X = x} \| Q_1|_{X = x}, \dots, Q_r|_{X = x} \big),
\end{equation}
meaning that, in Equation~\eqref{eq:general_conditioning_formula}, we have $g(p, \star, \dots, \star) = p$.
Additionally, we make the assumption that $\CR$ is given by
\begin{equation}\label{eq:F_stronger_assumption}
  \big[\CR(X)\big]\big(P \| Q_1, \dots, Q_r\big) = \sum_{x \in \vs{X}}P(x) \cdot h\big(P(x), Q_1(x), \dots, Q_r(x)\big)
\end{equation}
for some function $h: [0, 1]^{r+1} \to \R$.
Clearly, all these stronger assumptions hold for Shannon entropy, Kullback-Leibler divergence, and cross-entropy, see, e.g., Equations~\eqref{eq:entropy_definition} and~\eqref{eq:kl_divergence}.\footnote{To make those examples fit the framework precisely, one needs to define $\log(0)$ to be an arbitrary real number instead of $- \infty$. This does not change the result as $\log(0)$ is always multiplied with the coefficient $0$.}

Now, fix random variables $X_1, \dots, X_n$ on $\samp$ and a graph $\Gr{G} = (\Ver{V}, \Ed{E})$ with $\Ver{V} = \start{n}$. 
We now define the property
\begin{equation}\label{eq:mrf_property_yeah}
  \Prty{MRF} \coloneqq \Prty{MRF}(\Gr{G}; X_1, \dots, X_n).
\end{equation}
When applied to a probability tuple $\big( P \| Q_1, \dots, Q_r \big)$, this evaluates to ``true'' if $X_1, \dots, X_n$ form a $P$-Markov random field and, for all $i$, a $Q_i$-Markov random field, with respect to $\Gr{G}$.
We know from Corollary~\ref{cor:mrf_stable_conditioning} that $\Prty{MRF}$ is stable under conditioning.
It is also clearly well-defined, since conditional independences do not depend on the equivalence classes of the involved random variables.
Finally, it is measurable since it is a conjunction of conditional independence relations: the subset of probability mass functions that satisfy an independence relation can by Equation~\eqref{eq:independence_relation} easily seen to be measurable.

Thus, $\Prty{MRF}$ is a stable property and we obtain the function $\WP{\CR}{\Prty{MRF}}$ that satisfies the chain rule by Proposition~\ref{pro:property_restriction_chain_rule}.
Consequently, we obtain a separation rule $\indep_{\WP{\CR}{\Prty{MRF}}}$ by Proposition~\ref{pro:F-separoid}.
In the following, we slightly generalize from this by considering properties that \emph{imply} the Markov random field property, which is crucial in the next section when restricting to properties that are motivated by physics.

\begin{theorem}
  \label{thm:get_mrf_for_free_with_restriction}
  Let $r \geq 0$ and $\gro{G}$ be the abelian group defined in Equation~\eqref{eq:my_ab_group}.
  Assume $\mon{M}$, the monoid of equivalence classes of random variables on $\samp$, acts additively on $\gro{G}$ by Equation~\eqref{eq:stronger_action}, and that $\CR: \mon{M} \to \gro{G}$ is a function of the form Equation~\eqref{eq:F_stronger_assumption} satisfying the chain rule Equation~\eqref{eq:cocycle_equationn}.

  Now, let $X_1, \dots, X_n$ be random variables on $\samp$, $\Gr{G}$ a graph with vertex set $[n]$, and $\Prty{R} = \Prty{R}\big(\Gr{G}; X_1, \dots, X_n \big)$ a stable property that implies the property $\Prty{MRF}$ from Equation~\eqref{eq:mrf_property_yeah}:
  \begin{equation*}
    \forall \big( P \| Q_1, \dots, Q_r \big) \in \AC{\Sim{\samp}^{r+1}}: \ \ \Prty{R}\big( P \| Q_1, \dots, Q_r\big) \quad \Longrightarrow \quad \Prty{MRF}\big( P \| Q_1, \dots, Q_r \big).
  \end{equation*}
  Then $X_1, \dots, X_n$ form an $\WP{\CR}{\Prty{R}}$-Markov random field with respect to $\Gr{G}$, where $\WP{\CR}{\Prty{R}}$ is defined as in Equation~\eqref{eq:new_F_def} as a restriction of $\CR$.
\end{theorem}

\begin{proof}
  Let $\Ver{A}, \Ver{B}, \Ver{C} \subseteq \Ver{V} = \start{n}$ be disjoint vertex sets such that $\Ver{C}$ separates $\Ver{A}$ from $\Ver{B}$.
  We need to show the $\WP{\CR}{\Prty{R}}$-independence $\IndepF{\WP{\CR}{\Prty{R}}}{X_{\Ver{A}}}{X_{\Ver{B}}}{X_{\Ver{C}}}$,
  which, by Proposition~\ref{pro:pairwise_independence_characterization}, is equivalent to the equality
  \begin{equation}
    \label{eq:independence_to_show}
    \big( X_{\Ver{C}}X_{\Ver{B}} \big) \WP{.}{\Prty{R}}  \WP{\CR}{\Prty{R}}(X_{\Ver{A}}) = 
    X_{\Ver{C}} \WP{.}{\Prty{R}}  \WP{\CR}{\Prty{R}}(X_{\Ver{A}}).
  \end{equation}
  For all $\big( P \| Q_1, \dots, Q_r \big) \in \WP{\big(\AC{\Sim{\samp}^{r+1}}\big)}{\Prty{R}}$, we obtain:
  \begingroup
    \allowdisplaybreaks
  \begin{align*}
    \Big[ & \big( X_{\Ver{C}} X_{\Ver{B}} \big)  \WP{.}{\Prty{R}} \WP{\CR}{\Prty{R}} (X_{\Ver{A}}) \Big]  \big( P \| Q_1, \dots, Q_r \big) \\
    & =  \sum_{x_{\Ver{C}},x_{\Ver{B}}} P(x_{\Ver{C}},x_{\Ver{B}}) \cdot \Big[ \WP{\CR}{\Prty{R}}(X_{\Ver{A}}) \Big]\Big(P|_{X_{\Ver{C}}X_{\Ver{B}} = (x_{\Ver{C}},x_{\Ver{B}})} \ \| \ Q_1|_{X_{\Ver{C}}X_{\Ver{B}} = (x_{\Ver{C}},x_{\Ver{B}})}, \dots, Q_r|_{X_{\Ver{C}}X_{\Ver{B}} = (x_{\Ver{C}},x_{\Ver{B}})} \Big) \\
    & = \sum_{x_{\Ver{C}}, x_{\Ver{B}}} P(x_{\Ver{C}}, x_{\Ver{B}})
    \sum_{x_{\Ver{A}}} P\big( x_{\Ver{A}} \mid x_{\Ver{C}}, x_{\Ver{B}} \big) \cdot h\Big(P(x_{\Ver{A}} \mid x_{\Ver{C}}, x_{\Ver{B}}), \ Q_1(x_{\Ver{A}} \mid x_{\Ver{C}}, x_{\Ver{B}}), \dots, Q_r(x_{\Ver{A}} \mid x_{\Ver{C}}, x_{\Ver{B}})\Big) \\
& \overset{(\star)}{=}  \sum_{x_{\Ver{C}}, x_{\Ver{B}}} P(x_{\Ver{C}}, x_{\Ver{B}})
\sum_{x_{\Ver{A}}} P\big( x_{\Ver{A}} \mid x_{\Ver{C}}, x_{\Ver{B}} \big) \cdot h\Big(P(x_{\Ver{A}} \mid x_{\Ver{C}}), \ Q_1(x_{\Ver{A}} \mid x_{\Ver{C}}), \dots, Q_r(x_{\Ver{A}} \mid x_{\Ver{C}})\Big) \\
  & = \sum_{x_{\Ver{C}}, x_{\Ver{A}}} 
  \bigg( \sum_{x_{\Ver{B}}} P\big( x_{\Ver{A}}, x_{\Ver{C}}, x_{\Ver{B}} \big) \bigg) \cdot h\Big(P(x_{\Ver{A}} \mid x_{\Ver{C}}), \ Q_1(x_{\Ver{A}} \mid x_{\Ver{C}}), \dots, Q_r(x_{\Ver{A}} \mid x_{\Ver{C}})\Big) \\
  & = \sum_{x_{\Ver{C}}} P(x_{\Ver{C}}) \sum_{x_{\Ver{A}}} P(x_{\Ver{A}} \mid x_{\Ver{C}}) \cdot h\Big(P(x_{\Ver{A}} \mid x_{\Ver{C}}), \ Q_1(x_{\Ver{A}} \mid x_{\Ver{C}}), \dots, Q_r(x_{\Ver{A}} \mid x_{\Ver{C}})\Big) \\
  & = \sum_{x_{\Ver{C}}} P(x_{\Ver{C}}) \cdot \Big[ \WP{\CR}{\Prty{R}}(X_A)  \Big]
  \Big( P|_{X_{\Ver{C}} = x_{\Ver{C}}} \ \| \ Q_1|_{X_{\Ver{C}} = x_{\Ver{C}}}, \dots, Q_r|_{X_{\Ver{C}} = x_{\Ver{C}}} \Big) \\
  & = \Big[ X_{\Ver{C}} \WP{.}{\Prty{R}} \WP{\CR}{\Prty{R}}(X_A) \Big] \big( P \| Q_1, \dots, Q_r\big).
\end{align*}
\endgroup
  In the calculation, in step $(\star)$ we used the assumption that 
  \begin{equation*}
    \IndepF{P}{X_{\Ver{A}}}{X_{\Ver{B}}}{X_{\Ver{C}}}, \quad 
    \IndepF{Q_1}{X_{\Ver{A}}}{X_{\Ver{B}}}{X_{\Ver{C}}}, \quad  \dots \quad
    \IndepF{Q_r}{X_{\Ver{A}}}{X_{\Ver{B}}}{X_{\Ver{C}}},
  \end{equation*}
  which is due to $\big( P \| Q_1, \dots, Q_r \big) \in \WP{\big(\AC{\Sim{\samp}^{r+1}}\big)}{\Prty{R}} \subseteq \WP{\big(\AC{\Sim{\samp}^{r+1}}\big)}{\Prty{MRF}}$.
  The other steps follow from the definitions.
  Overall, this shows Equation~\eqref{eq:independence_to_show}, and we are done.
\end{proof}

\subsection{The Second Law of Thermodynamics}\label{sec:new_kl_section}

In this section, we study a weak version of the second law of thermodynamics and show that it can be represented visually by the degeneracy of a certain Kullback-Leibler diagram.

We fix random variables $X_1, \dots, X_n$ on $\samp$.
We now work with the rather special property $\Prty{P} = \Prty{P}(X_1, \dots, X_n)$ of probability tuples $\big( P \| Q \big) \in \Sim{\samp}^2$ defined as follows:
$\Prty{P}\big( P \| Q\big)$ holds if and only if
\begin{itemize}
  \item $P$ is absolutely continuous with respect to $Q$;
  \item $X_1, \dots, X_n$ forms a $P$-Markov chain and $Q$-Markov chain; and
  \item for all $i \geq 2$ and all $(x_{i-1}, x_i)$ with $P(x_{i-1}) \neq 0$, we have $P(x_i \mid x_{i-1}) = Q(x_i \mid x_{i-1})$.
\end{itemize}

The only difference of $P$ and $Q$ then stems from the initial distributions $P(X_1)$, $Q(X_1)$.
Intuitively, the letter ``P'' in the property $\Prty{P} = \Prty{P}(X_1, \dots, X_n)$ is supposed to remind of ``Physics''. 
Namely, let $(P, Q) \in \widetilde{\Sim{\samp}^2}$ with $\Prty{P}\big( P \| Q \big)$, then:

\begin{itemize}
  \item $i = 1, \dots, n$ indexes time points with equal separation;
  \item The value spaces $\vs{X_i}$ model the sets of possible micro states at time point $i$ (in reality, all these spaces are equal to each other, but we need not make this assumption);
  \item $P(X_1)$ and $Q(X_1)$ are the initial distributions of $P$ and $Q$, which can be considered as macro states of the universe at some fixed point in time;
    later, in the context of the second law of thermodynamics, we will choose $P(X_1)$ to be arbitrary and $Q(X_1)$ to be the uniform distribution;
  \item $P(x_i \mid x_{i-1}) = Q(x_i \mid x_{i-1})$ models the likelihood that $x_{i-1}$ evolves to $x_i$ according to the ``physical laws'' described by $P$ and $Q$;
    In reality, the physical laws are fixed and time independent, which means there is one transition matrix $T$ such that $P(x_i \mid x_{i-1}) = Q(x_i \mid x_{i-1}) = T(x_i \mid x_{i-1})$ for all $i$ and all $x_{i-1}, x_i$, but we do not yet make this stronger assumption;
  \item $P(X_i)$ and $Q(X_i)$ are the macro states of the universe at time points $i$ when evolving $P(X_1)$ and $Q(X_1)$ according to the ``physical laws'' specified by $P$ and $Q$.
    These laws are equal to each other by assumption.
\end{itemize}

In the following, we want to investigate a weak version of the second law of thermodynamics, which states that under some conditions, entropy cannot decrease over time. 
We show this by first investigating the Kullback-Leibler diagram for the property $\Prty{P}$ in Theorem~\ref{thm:second_law}, which will lead to the desired result in Corollary~\ref{cor:second_law_thermo}.

Formally, let $\KL: \mon{M} \to \Meas\Big( \AC{\Sim{\samp}^2}, \R \Big)$ be the Kullback-Leibler divergence, see Equation~\eqref{eq:kl_divergence}.
The property $\Prty{P}$ is stable under conditioning by Proposition~\ref{pro:conditionals_still_same_propagates}, and is also easily seen to be measurable and well-defined.
Therefore, by Proposition~\ref{pro:property_restriction_chain_rule}, the restricted Kullback-Leibler divergence
\begin{equation*}
  \WP{\KL}{\Prty{P}}: \WP{\mon{M}}{\Prty{P}} \to \Meas \Big( \WP{\big( \AC{\Sim{\samp}^2} \big)}{\Prty{P}}, \R \Big) \eqqcolon \WP{\gro{G}}{\Prty{P}}
\end{equation*}
satisfies the chain rule, where $\WP{\mon{M}}{\Prty{P}}$ is the monoid generated by $X_1, \dots, X_n$.
Accordingly, by Hu's Theorem~\ref{thm:hu_kuo_ting_generalized}, we obtain a corresponding $\WP{\gro{G}}{\Prty{P}}$-valued measure $\set{\WP{\KL}{\Prty{P}}}: 2^{\set{X}} \to \WP{\gro{G}}{\Prty{P}}$. 
The following theorem shows that this measure degenerates, as visualized in Figure~\ref{fig:second_law}:

\begin{theorem}[Structure of the $\WP{\KL}{\Prty{P}}$-diagram]
  \label{thm:second_law}
  Let $\KL: \mon{M} \to \gro{G}$ be the Kullback-Leibler divergence from Equation~\eqref{eq:kl_divergence}. 
  Let $X_1, \dots, X_n$ be random variables on $\samp$ and consider the property $\Prty{P} = \Prty{P}(X_1, \dots, X_n)$ defined above, leading to a restriction $\WP{\KL}{\Prty{P}}: \WP{\mon{M}}{\Prty{P}} \to \WP{\gro{G}}{\Prty{P}}$ and a $\WP{\gro{G}}{\Prty{P}}$-valued measure $\set{\WP{\KL}{\Prty{P}}}: 2^{\set{X}} \to \WP{\gro{G}}{\Prty{P}}$. 
  One has the following:
  \begin{itemize}
    \item For all atoms $\atom{I}$ where $I$ does not only consist of consecutive numbers, one has $\set{\WP{\KL}{\Prty{P}}}(\atom{I}) = 0$;
    \item Let $\atom{\range{i}{k}}$ with $1 \leq i \leq k \leq n$ be an atom that consists of only consecutive numbers. 
      If $i \geq 2$, then $\set{\WP{\KL}{\Prty{P}}}\big(\atom{\range{i}{k}}\big) = 0$.
  \end{itemize}
\end{theorem}

\begin{proof}
  By Proposition~\ref{pro:characterization_markov_chain} and Theorem~\ref{thm:get_mrf_for_free_with_restriction}, $X_1, \dots, X_n$ forms a $\WP{\KL}{\Prty{P}}$-Markov random field with respect to the graph $\Gr{G}$ with edges $i \edge (i+1)$ for $i = 1, \dots, n-1$.
  By Proposition~\ref{pro:characterization_markov_chain} again, we know that this is equivalent to forming a $\WP{\KL}{\Prty{P}}$--Markov chain.
  Then by Corollary~\ref{cor:F_Markov_Chain_Charac}, one has $\set{\WP{\KL}{\Prty{P}}}(\atom{I}) = 0$ if $I \subseteq \start{n}$ does \emph{not} only consist of consecutive numbers.
  Thus, only the atoms $\atom{I}$ of the form $I = \range{i}{k}$ are of relevance for the $\WP{\KL}{\Prty{P}}$-diagram. 

  Let $i \geq 2$. 
  Using Lemma~\ref{lem:how_atoms_look_like}, Corollary~\ref{cor:information_terms_markov_chains}, and Equation~\eqref{eq:inductive_definition}, we have
  \begin{align*}
    \set{\WP{\KL}{\Prty{P}}}\big(\atom{\range{i}{k}}\big) & = \big(X_{\start{i-1}}X_{\range{k+1}{n}}\big).\WP{\KL}{\Prty{P}}\big(X_i; \dots; X_k \big) \\
    & = \big(X_{\start{i-1}}X_{\range{k+1}{n}}\big).\WP{\KL}{\Prty{P}}\big( X_i; X_k\big) \\
    & = \big(X_{\start{i-1}}X_{\range{k+1}{n}}\big).\WP{\KL}{\Prty{P}}\big( X_i \big) - \big(X_{\start{i-1}}X_{\range{k}{n}}\big).\WP{\KL}{\Prty{P}}\big( X_i \big) \\
    & = \big( X_{\start{i-2}}X_{\range{k+1}{n}}\big).\Big( X_{i-1}.\WP{\KL}{\Prty{P}}\big( X_i  \big)\Big) -  \big( X_{\start{i-2}}X_{\range{k}{n}} \big). \Big( X_{i-1}.\WP{\KL}{\Prty{P}}\big( X_i \big)\Big)
  \end{align*}
  Since the monoid action sends zero elements to zero, it is enough to show $X_{i-1}.\WP{\KL}{\Prty{P}}(X_i) = 0$.
  For arbitrary $\big( P \| Q \big) \in \WP{\big(\AC{\Sim{\samp}^2}\big)}{\Prty{P}}$, we have
    \begin{align*}
      \big[ X_{i-1}.\WP{\KL}{\Prty{P}}(X_i) \big]\big( P \sep Q\big) & = \sum_{x_{i-1}} P(x_{i-1}) \cdot \big[\WP{\KL}{\Prty{P}}(X_i)\big] \big( \cond{P}{X_{i-1} = x_{i-1}} \ \sep \ \cond{Q}{X_{i-1} = x_{i-1}}\big) \\
  & = \sum_{x_{i-1}} P(x_{i-1}) \sum_{x_{i}} P(x_i \mid x_{i-1}) \cdot \log \frac{P(x_i \mid x_{i-1})}{Q(x_i \mid x_{i-1})} \\
    & =\sum_{x_{i-1}} P(x_{i-1}) \sum_{x_{i}} P(x_i \mid x_{i-1}) \cdot \log 1  \\
    & = 0,
  \end{align*}
  where in the third step we used the assumptions of equal transition probabilities, which holds since $\Prty{P}(P \| Q)$ is true.
  That was to show. 
\end{proof}

\begin{figure}
  \centering
  \includegraphics[width=\textwidth]{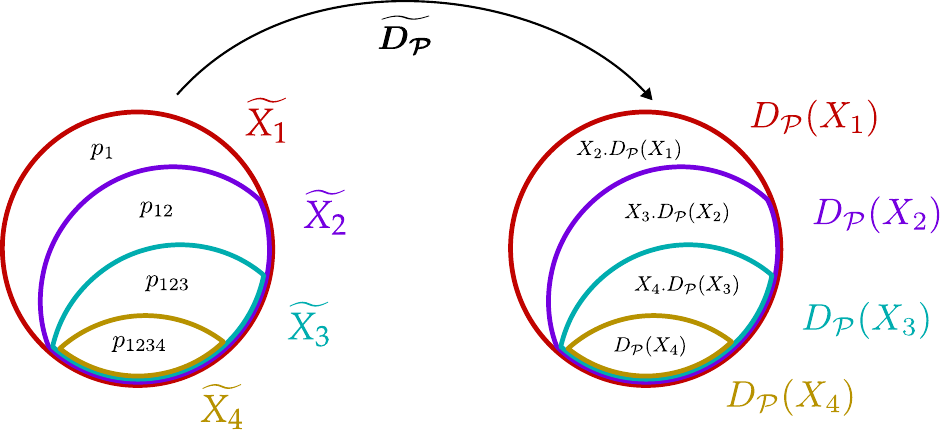}
  \caption{If one restricts the set of tuples $(P, Q)$ of probability mass functions to those that have equal transition probabilities and individually give rise to a Markov chain $X_1, \dots, X_n$, then this leads to a restricted Kullback-Leibler divergence $\WP{\KL}{\Prty{P}}$ and the corresponding measure $\set{\WP{\KL}{\Prty{P}}}$.
    Theorem~\ref{thm:second_law} shows that in the corresponding $\WP{\KL}{\Prty{P}}$-diagram, only atoms of the form $\atom{12\dots k}$ remain.
This is visualized here for the case $n = 4$ by omitting all atoms that are mapped to zero by $\set{\WP{\KL}{\Prty{P}}}$.
\emph{Intuitively, the Kullback-Leibler divergence shrinks over time.}
In particular, this leads by Hu's Theorem~\ref{thm:hu_kuo_ting_generalized} to the chain rule $\WP{\KL}{\Prty{P}}(X_{i-1}) = \WP{\KL}{\Prty{P}}(X_i) + X_i.\WP{\KL}{\Prty{P}}(X_{i-1})$, which we use in the proof of Corollary~\ref{cor:second_law_thermo} to deduce a weak version of the second law of thermodynamics.}
  \label{fig:second_law}
\end{figure}

We deduce the following weak version of the second law of thermodynamics:

\begin{corollary}
  [Second Law of Thermodynamics]
  \label{cor:second_law_thermo}
  Let all notation be as above. 
  Assume that the joint $X_1\cdots X_n: \samp \to \vs{X_1} \times \dots \times \vs{X_n}$ is surjective.
  Let $P$ be a probability mass function on $\samp$ such that $X_1, \dots, X_n$ forms a $P$-Markov chain.
  Assume that $\vs{X_i} = \vs{X_{j}}$ for all $i, j = 1, \dots, n$ and that there is a time-independent doubly stochastic transition matrix $T$ such that $P(x_{i} \mid x_{i-1}) = T(x_{i} \mid x_{i-1})$ for all $i = 2, \dots, n$ and $(x_{i-1}, x_{i})$ with $P(x_{i-1}) \neq 0$.
  Let $I: \mon{M} \to \gro{G}$ be the Shannon entropy, as defined in Equation~\eqref{eq:entropy_definition}.
  Then for all $i = 2, \dots, n$, one has
  \begin{equation*}
    \big[\Ent(X_{i-1})\big](P) \leq \big[ \Ent(X_{i})\big](P),
  \end{equation*}
  i.e., entropy cannot decrease.
\end{corollary}

\begin{proof}
  The proof is a diagrammatic representation of the reasoning in~\cite{Cover2005}, Chapter 4.4.
  Let $Q$ be any probability mass function on $\samp$ such that
  \begin{itemize}
    \item $Q_{X_1}$ is the uniform distribution on $X_1$;
    \item $Q(x_{i} \mid x_{i-1}) = T(x_{i} \mid x_{i-1})$ for all $i = 2, \dots, n$ and $(x_{i-1}, x_{i})$.\footnotemark
  \end{itemize}
  \footnotetext{To see that this exists, first construct a joint distribution on $\vs{X_1} \times \dots \times \vs{X_n}$ with these properties by multiplying the initial distribution with all the transitional distributions. 
  Then, ``pull it back'' to $\samp$, for which there is at least one possible construction since the joint $X_1\cdots X_n$ is surjective.} 
  Then, as is well-known, the fact that $T$ is doubly stochastic implies that $Q_{X_i}$ is uniform for all $i = 1, \dots, n$. 
  Let $\KL$ be the Kullback-Leibler divergence.
  We obtain
  \begin{align*}
    \big[\KL(X_{i-1})\big]( P \| Q) & = - \sum_{x_{i-1}} P(x_{i-1}) \log \frac{Q(x_{i-1})}{P(x_{i-1})} \\
    & = - \sum_{x_{i-1}} P(x_{i-1}) \log \frac{1}{\num{\vs{X_{i-1}}}} - \big[\Ent(X_{i-1})\big] (P)  \\
    & = C - \big[\Ent(X_{i-1})\big]( P),
  \end{align*}
  where $C$ is a constant, i.e., independent of $i, P$ and $Q$, and where the uniformity of $Q_{X_{i-1}}$ was used in the second step.
  Thus, it is enough to show that $\big[\KL(X_{i-1})\big]( P \| Q)$ satisfies the following inequality:
  \begin{equation*}
    \big[\KL(X_{i-1})\big]( P \| Q) \geq \big[\KL(X_{i})\big](P \| Q).
  \end{equation*}
  Now, note that $P \ll Q$ since $Q$ is nowhere zero.
  Also, by assumption, $P$ and $Q$ both have equal transition probabilities given by $T$.
  We obtain $\big( P \| Q\big) \in \WP{\Big(\AC{\Sim{\samp}^2}\Big)}{\Prty{P}}$.
  Thus, Theorem~\ref{thm:second_law}, as visualized by Figure~\ref{fig:second_law}, shows
  \begin{equation}\label{eq:equation_to_quote}
    \big[\KL(X_{i-1})\big](P \| Q) =  \big[\KL(X_{i})\big]( P \| Q) + \big[ X_{i}.\KL(X_{i-1})\big] (P \| Q) \geq \big[ \KL(X_{i})\big] (P \| Q),
  \end{equation}
  where in the last step we used that conditional Kullback-Leibler divergence is non-negative.
  That was to show.
\end{proof}

\subsection{Diffusion Models}\label{sec:diffusion_models}

We briefly elaborate another simple illustration of our theory, namely its help in deriving the explicit decomposition of the evidence lower bound (ELBO) in diffusion models~\citep{Dickstein2015}. 
Diffusion models are used in text-to-image generation models like Dalle~\citep{Ramesh2021} and Imagen~\citep{Chitwan2022} that recently received widespread attention.

In the classical formulation of diffusion models, the data is assumed to come from a distribution $Q(X) \in \Sim{\vs{X}}$ that is progressively transformed by fixed conditional noise distributions $Q(Z_1 \mid X)$ and $Q(Z_t \mid Z_{t-1})$, $t = 2, \dots, T$ over latent value spaces $\vs{Z_1}, \dots, \vs{Z_T}$. 
One can then form the joint distribution $Q(X, \mathbf{Z}) = Q(X) \cdot Q(Z_1 \mid X) \cdot \prod_{t = 2}^{T} Q(Z_t \mid Z_{t-1})$.
$X, Z_1, \dots, Z_T$ then form a Markov chain with respect to $Q$.

The goal for diffusion models is to learn a model distribution $P(X)$ that is close to the data distribution $Q(X)$.
This is done by fixing a latent distribution $P(Z_T)$ and initializing denoising distributions $P(Z_{t-1} \mid Z_{t})$, $t = 2, \dots, T$ and $P(X \mid Z_{1})$ that are parameterized by deep neural networks.
As before, one can then also build the joint distribution $P(X, \mathbf{Z}) = P(Z_T) \cdot \prod_{t = 2}^{T} P(Z_{t - 1} \mid Z_{t}) \cdot P(X \mid Z_{1})$.
Clearly, $Z_T, \dots, Z_1, X$ then form a Markov chain with respect to $P$, and due to symmetry (compare Proposition~\ref{pro:characterization_markov_chain}), also $X, Z_1, \dots, Z_T$ form a Markov chain with respect to $P$.

Now, fix a datapoint $x \in \vs{X}$ sampled from $Q(X)$.
Since minimizing the Kullback-Leibler divergence between $Q(X)$ and $P(X)$ is equivalent to maximizing the log-likelihood of $P(X)$ for datapoints sampled from $Q(X)$, the goal is to make a gradient step that increases $\log P(x)$.
In diffusion models, this is accomplished by constructing a so-called \emph{evidence-lower bound} (ELBO) of this objective that can easily be optimized with respect to the neural network parameters and does not require intractable summations.
This is in contrast to the marginal $P(x)$, which is a sum over all $\mathbf{z} \in \vs{Z_1} \times \dots \times \vs{Z_T}$.
The main result is as follows:

\begin{proposition}[Computing the ELBO for diffusion models]
  \label{pro:evidence_lower_bound}
  Assume $P(Z_T)$ is fixed and that $Q(Z_T \mid x) = P(Z_T)$.\footnote{In practice, this is achieved by adding so much noise to $x$ over the Markov chain $Q$ that eventually, all information in $x$ is destroyed.}
  Define the ELBO by
  \begin{equation*}
    \loss(x) \coloneqq \sum_{\mathbf{z}} Q(\mathbf{z} \mid x)  \log \frac{P(x, \mathbf{z})}{Q(\mathbf{z} \mid x)}.
  \end{equation*}
  Then $\loss(x) \leq \log P(x)$ and
  \begin{equation}\label{eq:elbo_criterion}
    \loss(x) = \sum_{z_1} Q(z_1 \mid x) \log P(x \mid z_1) - \sum_{t = 2}^{T} \sum_{z_t} Q(z_t \mid x) \cdot \KL\Big( Q(Z_{t-1} \mid z_t, x) \ \big\| \ P(Z_{t-1} \mid z_t) \Big).\footnotemark
  \end{equation}
  \footnotetext{For simplicity, we write $\KL\big(Q(Y) \| P(Y)\big)$ for $[\KL(Y)](Q \| P)$.}
\end{proposition}

\begin{proof}
  The two claims involve two different factorizations of $P(x, \mathbf{z})$.
  For the first claim, we note
  \begin{align*}
    \loss(x) &= \sum_{\mathbf{z}} Q(\mathbf{z} \mid x) \log \frac{P(x) \cdot P(\mathbf{z} \mid x)}{Q(\mathbf{z} \mid x)} \\
    &= \log P(x) - D\big( Q(\mathbf{Z} \mid x) \ \| \ P(\mathbf{Z} \mid x) \big) \\
    &\leq \log P(x),
  \end{align*}
  which follows from the non-negativity of Kullback-Leibler divergence.
  
  For the second claim, we note
  \begin{align*}
    \loss(x) &= \sum_{\mathbf{z}} Q(\mathbf{z} \mid x) \log \frac{P(x \mid \mathbf{z}) \cdot P(\mathbf{z})}{Q(\mathbf{z} \mid x)} \\
    &= \sum_{\mathbf{z}} Q(\mathbf{z} \mid x) \log P(x \mid \mathbf{z}) - \KL\big( Q(\mathbf{Z} \mid x) \ \| \ P(\mathbf{Z})  \big)
  \end{align*}
  Now, since $Z_T, \dots, Z_1, X$ form a Markov chain with respect to $P$, we have $P(x \mid \mathbf{z}) = P(x \mid z_1)$, and so the left part of the formula simplifies:
  \begin{equation*}
    \sum_{\mathbf{z}} Q(\mathbf{z} \mid x) \log P(x \mid \mathbf{z}) = \sum_{z_1} Q(z_1 \mid x) \log P(x \mid z_1).
  \end{equation*}
  Thus, it just remains to evaluate the Kullback-Leibler divergence.
  We first perform the computation and then justify our steps:
  \begin{align*}
    \KL\big( Q(\mathbf{Z} \mid x) \ \| \ P(\mathbf{Z}) \big) &\overset{(1)}{=} \big[ \KL(\mathbf{Z}) \big]\big( Q|_{X = x} \ \| P \big) \\
    &\overset{(2)}{=} \Bigg[ \sum_{t = 2}^{T} Z_{t}.\KL(Z_{t-1}) + \KL(Z_T) \Bigg]\big( Q|_{X = x} \ \| P \big) \\
    &\overset{(3)}{=} \sum_{t = 2}^{T} \sum_{z_t} Q(z_t \mid x) \cdot \KL \Big( Q(Z_{t-1} \mid z_t, x) \ \big\| \ P(Z_{t-1} \mid z_t) \Big) \\
    & \ \ \ \ \ + \KL\big( Q(Z_T \mid x) \ \| \ P(Z_T) \big).
  \end{align*}
  Since $Q(Z_T \mid x)$ and $P(Z_T)$ are assumed identical, the last Kullback-Leibler divergence disappears, proving the original claim.

  We now justify the steps.
  Step (1) just makes explicit that Kullback-Leibler divergence is a function whose arguments are a random variable \emph{followed by} a pair of distributions, see Equation~\eqref{eq:kl_divergence}.

  For step (2), remember that $X, Z_1, \dots, Z_T$ form a Markov chain with respect to $P$ and $Q$. 
  The Markov chain property is conditionally stable by Corollary~\ref{cor:mrf_stable_conditioning} and the fact that Markov chains are special Markov random fields (Proposition~\ref{pro:characterization_markov_chain}).
  Thus, $X, Z_1, \dots, Z_T$ also form a Markov chain with respect to $Q|_{X = x}$.
  Consequently, by Theorem~\ref{thm:get_mrf_for_free_with_restriction}, the \emph{restriction} of $\KL$ to the Markov chain property gives rise to a $\KL$-Markov chain.
  Thus, using Corollary~\ref{cor:F_Markov_Chain_Charac}, as visualized in Figure~\ref{fig:four_and_five_circles}, together with Hu's Theorem (Theorem~\ref{thm:hu_kuo_ting_generalized}), gives rise to the decomposition in step (2).

  Step (3) simply uses the definition of the monoid action, Equation~\eqref{eq:monoid_action_KL}.
  This proves the claim.
\end{proof}

Comparing with the exposition in~\cite{Bishop2023}, we see that our derivation avoids a few complications:
We do not need to compute explicit Bayesian posteriors, reason about marginalizations over many variables, or reason about terms that cancel each other. 
Since the Kullback-Leibler divergence is ubiquitous in machine learning, we think there could be many similar applications of our work that aim to find decompositions of loss functions over Markov random fields. 
We think that for Markov random fields that are more complex than the Markov chains encountered in diffusion models, our formalism could make finding useful decompositions and approximations of loss functions simpler than the alternative of more low-level reasoning.

\section{Discussion}\label{sec:discussion}

\subsection{Major Findings: Characterizations of \texorpdfstring{$\CR$}{CR}--FCMIs, \texorpdfstring{$\CR$}{CR}--Markov Random Fields, and Probabilistic Applications}

In this work, we have generalized the main results in~\cite{Yeung2002a}, which characterize probabilistic full conditional mutual independences and Markov random fields in terms of $I$-diagrams~\citep{Hu1962,Yeung1991}.
In doing so, we replaced the Shannon entropy $\Ent$ with any function $\CR: \mon{M} \to \gro{G}$ from a commutative, idempotent monoid $\mon{M}$ to an abelian group $\gro{G}$ that satisfies the chain rule:
\begin{equation*}
  F(XY) = F(X) + X.F(Y).
\end{equation*}
The dot denotes an additive monoid action of $\mon{M}$ on $\gro{G}$ that generalizes the conditioning of information functions on random variables.

Consequently, we also replaced the familiar probabilistic conditional $P$-independence --- which is characterized by a vanishing of conditional mutual information --- by $\CR$-independence:
\begin{equation*}
  \IndepF{\CR}{X}{Y}{Z}  \quad :\Longleftrightarrow \quad Z.\CR(X; Y) = 0.
\end{equation*}
This independence relation gives rise to a separoid~\citep{Dawid2001}, as we have shown in Proposition~\ref{pro:F-separoid}, a powerful framework in which one can study conditional independence relations.
Separoids generally also allow us to study conditional \emph{mutual} independences and full conditional mutual independences (FCMIs).
These can be used to characterize Markov random fields by the \emph{cutset property}, as was first observed in~\cite{Yeung2002a} for the classical case --- see Section~\ref{sec:mrfs_separoids}.
When specializing to the separoid with the independence relation given by $\indep_{\CR}$ we obtained the notions of $\CR$-FCMIs and $\CR$-Markov random fields that generalize the probabilistic counterparts.

By the generalized Hu Theorem~\ref{thm:hu_kuo_ting_generalized} from~\cite{Lang2022}, for fixed elements $X_1, \dots, X_n \in \mon{M}$, one obtains a $\gro{G}$-valued measure $\set{F}$ that describes relations of information terms $X_J.F\big( X_{L_1}; \dots ; X_{L_q} \big)$ of arbitrary degree $q$, where $J, L_1, \dots, L_q \subseteq \start{n}$ are any subsets.
These can then be visualized in $\CR$-diagrams, as we show for $n = 3$ in Figure~\ref{fig:venn_diagram_comparison}.
The core question asked in this work --- generalizing the investigations in~\cite{Yeung2002a} --- was how $\CR$-FCMIs and $\CR$-Markov random fields can be characterized in terms of the $\CR$-diagram. 
In Theorem~\ref{thm:FCMIs_characterization}, we found that $\CR$-FCMIs are characterized by a vanishing of the $\CR$-diagram on the atoms in the \emph{image} of the corresponding conditional partition, see also Figures~\ref{fig:fcmi_visualization_1} and~\ref{fig:fcmi_visualization_2}.

This results in a characterization of $\CR$-Markov random fields for a graph $\Gr{G}$, Theorem~\ref{thm:mrf_characterization}, that is easy to interpret: an atom $\atom{\Ver{W}}$ disappears in the $\CR$-diagram if and only if it is \emph{disconnected}, meaning that the vertex set $\Ver{W}$ is disconnected when removing the vertices outside of $\Ver{W}$ from $\Gr{G}$.
We visualized this result in Figure~\ref{fig:three_variables_mrfs} for Shannon entropy, though the basic picture applies for arbitrary $\CR$.
A visualization for the special case of a path-shaped graph --- recovering $\CR$-Markov chains --- can be found in Figure~\ref{fig:four_and_five_circles}.

When looking closely at the results, which are always given as an equivalence of two or more statements, it becomes clear that one direction is always easy to prove --- namely, if one starts assuming that a suitable set of atoms disappears in the $\CR$-diagram, then one can easily deduce an $\CR$-FCMI from it (possibly in the context of an $\CR$-Markov random field). 
After all, $\set{\CR}$ is a measure, and so the vanishing of a set of atoms automatically leads to the vanishing of all sets \emph{composed} of these atoms, and this includes by Hu's Theorem~\ref{thm:hu_kuo_ting_generalized} precisely the sets encoding the $\CR$-FCMIs.
However, going in the other direction is harder: why should an $\CR$-FCMI induce the vanishing of specific atoms?
This problem is addressed by our main technique ``subset determination'', Theorem~\ref{thm:subset_determination}.
It implies that whenever $\set{F}(A) = 0$ for some set of atoms $A$, then we also have $\set{F}(B) = 0$ for all subsets $B \subseteq A$, and in particular, $\set{\CR}(\atom{I}) = 0$ for all atoms $\atom{I} \in A$.
This is essentially based on inclusion-exclusion type arguments for $\CR$-diagrams, and uses the monoid action from $\mon{M}$ on $\gro{G}$.
We illustrate how to use subset determination to our advantage in Figure~\ref{fig:subset_determination_application}.
Essentially, subset determination replaces the use of inequalities in the original work~\citep{Yeung2002a}.

We then applied our results to the case where the information functions apply to probability mass functions. 
In Lemma~\ref{lem:well-defined_additive_action} and Proposition~\ref{pro:property_restriction_chain_rule} we showed that it is possible to restrict a general notion of information functions to conditionally stable properties while preserving the monoid action and the chain rule, which is similar to the use of adapted probability functors in information cohomology~\citep{Vigneaux2019a}.
Then, in Theorem~\ref{thm:get_mrf_for_free_with_restriction}, we restricted a narrower set of information functions $\CR$ --- including Shannon entropy, Kullback-Leibler divergence, and cross-entropy --- to stable properties $\Prty{R}$ that \emph{imply} the Markov random field property.
The restriction $\WP{\CR}{\Prty{R}}$ then gives rise to an $\WP{\CR}{\Prty{R}}$--Markov random field; consequently, disconnected atoms disappear from the $\CR_{\Prty{R}}$-diagram.

In Theorem~\ref{thm:second_law}, we applied this to the case of tuples of probability mass functions that give rise to a Markov chain and have the same conditional distributions --- reminiscent of physics, where the conditional distributions are governed by the physical laws.
It turns out that the corresponding $\KL$-diagram degenerates even further than the Markov chain property alone would predict; the Kullback-Leibler divergence of all regions outside the initial random variable disappears.
The reason is that in those regions, one ``conditions'' on the initial variable according to Hu's Theorem~\ref{thm:hu_kuo_ting_generalized}, and since the conditional distributions are the same, one obtains zero Kullback-Leibler divergence.

The result is an ever-shrinking\footnote{More precisely: non-increasing.} sequence of Kullback-Leibler divergences across the Markov chain, as we visualize in Figure~\ref{fig:second_law}.
This has the following consequence:
if one of the two distributions is uniform, then the other one cannot move further away from it as measured by the Kullback-Leibler divergence, which means that its entropy cannot decrease over time.
This is a weak version of the second law of thermodynamics, see Corollary~\ref{cor:second_law_thermo}.

Finally, we also used the Kullback-Leibler decomposition over Markov chains that results from Theorem~\ref{thm:get_mrf_for_free_with_restriction} to obtain a conceptually simple derivation of the evidence lower bound in diffusion models, Proposition~\ref{pro:evidence_lower_bound}.
This avoids some of the explicit computations used in prior expositions.

It is important to note that in most of these probabilistic results, we take a view in which information functions are still \emph{functions} of yet unspecified probability mass functions.
We simply restrict the sets of probability mass functions to those that satisfy stable properties, and can thereby preserve our general theory, the monoid action, and the ``subset determination'' property, Theorem~\ref{thm:subset_determination}, that does not generally hold for fixed probability mass functions.
This raises the question of whether the $P$-independence results in~\cite{Yeung2002a} are actually covered by our work.
We answer this affirmatively in Appendix~\ref{sec:slices}.
In the corresponding proofs, we use the fact that a $P$-independence $\IndepF{P}{X}{Y}{Z}$ is equivalent to the vanishing of the conditional mutual $P$-information $\big[Z.I(X;Y)\big](P)$, which is usually proved using Jensen's inequality. 
This reduction is then the only place where inequalities enter the theory, whereas they take a more central role in the proofs in~\cite{Yeung2002a}.

\subsubsection*{Further Findings: \texorpdfstring{$\CR$}--(Dual )Total Correlation, Cohomological Characterizations of Functions $\CR$, and Further Consequences}

Our proof of the $\CR$-diagram characterization of $\CR$-FCMIs uses Theorem~\ref{thm:charac_using_dual}, in which we characterize conditional mutual $\CR$-independences by the vanishing of $\CR$-dual total correlation.
This relates to $\CR$ in the same way that the classical dual total correlation from~\cite{Han1978} relates to Shannon entropy:
\begin{equation*}
  \DTC{\CR}(X_1; \dots; X_n) \coloneqq F\big( X_{\start{n}} \big) - \sum_{i = 1}^n X_{\start{n} \setminus i}.\CR(X_i).
\end{equation*}
Conditional $\CR$-dual total correlation $Y.\DTC{\CR}\big(X_1, \dots, X_n \big)$ corresponds by Hu's Theorem~\ref{thm:hu_kuo_ting_generalized} to a set of atoms, which coincide with those that vanish based on a conditional mutual $\CR$-independence $\bigindepF{\CR}_{i = 1}^n X_i \ | \ Y$. 
Again, we were able to use subset determination, Theorem~\ref{thm:subset_determination}, to prove the characterization.

In Appendix~\ref{sec:total_correlation_charac}, we also investigated $\CR$-total correlation, which generalizes the classical total correlation from~\cite{Watanabe1960}. 
We showed that it can also usually be used for a characterization of mutual $\CR$-independences. 
However, since it double-counts some atoms, the characterization only works for torsion-free groups $\gro{G}$, and we provide a counter example in the case of torsion, see Example~\ref{exa:torsion_example}.
This is not a strong restriction --- all groups that are used in practice seem to be derived from the real numbers $\R$, and therefore inherit the property to be torsion-free. 

Subset determination implies that $\CR\big( X_{\start{n}} \big)$ determines the whole $\CR$-diagram for the variables $X_1, \dots, X_n$.
In Appendix~\ref{sec:classification_CR}, this led to a classification of functions $\CR: \mon{M} \to \gro{G}$ satisfying the chain rule, in the case that $\mon{M}$ has a ``top element'' --- similar to how $X_{\start{n}}$ is ``on top'' of all elements $X_I$ for $I \subseteq \start{n}$.
We could show that such functions $\CR$ correspond precisely to elements in $\gro{G}$ that are annihilated by the top element, which we also used in our constructions for Example~\ref{exa:torsion_example}.
In Remark~\ref{rem:cohomological_interpretation}, we explained a cohomological interpretation of these findings:
the first Hochschild cohomology group of $\mon{M}$ with coefficients in $\gro{G}$ disappears.
This is in contrast to information cohomology~\citep{Baudot2015a,Vigneaux2019a}, where the joint locality property ensures that the first cohomology group is nontrivial; in fact, it is generated by Shannon entropy!

Finally, in Appendix~\ref{sec:generalizing_yeung_2019}, we explain that all results from~\cite{Yeung2019} except those that depend on an order relation also generalize to our setting.
Overall, this means that we have generalized most of the results in~\cite{Yeung1991,Kawabata1992,Yeung2002a,Yeung2019} from Shannon entropy to general functions $\CR$ satisfying the chain rule.

\subsection{Conceivable Extensions of the Theory and Open Questions}

\subsubsection*{\texorpdfstring{$K$}{K}-Independence and Kolmogorov Complexity}

We have set up the theory in such a way that it could in principle be extended beyond $\CR$-independence, since Section~\ref{sec:mrfs_in_separoids} introduces conditional (mutual) independences, Markov random fields, and Markov chains in the general context of separoids.
In this context, we also characterized Markov random fields --- defined in terms of the global Markov property --- by the \emph{cutset property} (Proposition~\ref{pro:equivalence_global_markov}), and showed that Markov chains can be equivalently described as Markov random fields corresponding to a path-shaped graph (Proposition~\ref{pro:characterization_markov_chain}).

We can think of one concrete way to potentially go beyond $\CR$-independence.
Namely,~\cite{Lang2022} also studied the case of functions $K: \mon{M} \times \mon{M} \to \gro{G}$ satisfying the chain rule absent of any monoid action of $\mon{M}$ on $\gro{G}$: $K(XY) = K(X) + K(Y \mid X)$,
where $K(Z) \coloneqq K(Z \mid \one)$.
One could then define the $K$-independence by $\IndepF{K}{X}{Y}{Z}$ if $K(X;Y \mid Z) = 0$.

One can then ask questions similar to the ones answered in this work: In what generality is $(\mon{M}, \indep_{K})$ a separoid? 
And when it is, does it allow to characterize conditional mutual $K$-independences and $K$-Markov random fields in terms of $K$-diagrams?
Note that $K$-diagrams do indeed exist, as proven in Corollary 3.3 in~\cite{Lang2022}.
However, we expect $\indep_{K}$ to not always satisfy the separoid axioms since our proof of the separoid axioms in Proposition~\ref{pro:F-separoid} made use of subset determination, Theorem~\ref{thm:subset_determination}, which in turn builds on the monoid action which is not available in the setting using $K$.

In a specific case of interest, $K$ would be a version of Kolmogorov complexity~\citep{Li1997}, especially Chaitin's prefix-free Kolmogorov complexity~\citep{Chaitin1987}, restricted to sets of strings that satisfy approximate independence relations. 
The fact that conditional mutual Kolmogorov complexity $K(x; y \mid z)$ is approximately non-negative implies that the separoid rules can be proved in this case, see~\cite{Steudel2010}.

\subsubsection*{\texorpdfstring{$\CR$}--Bayesian Networks}

Additionally, it would be interesting to attempt to go beyond Markov random fields by studying Bayesian networks.
A priori, they can be defined in general separoids using the d-separation criterion for a directed acyclic graph~\citep{Pearl1985}.
This can then be applied to the case of a separation $\indep_{\CR}$ coming from a function $\CR$ satisfying the chain rule.

When applying such a theory to the probabilistic case, however, one will likely encounter problems:
$\CR$-independences can only model conditionally stable cases, see Proposition~\ref{pro:weird_conditioning_trick}.
However, the independences in Bayesian networks are, due to colliders, not preserved under conditioning, see Remark~\ref{rem:independence_propagates}.
Therefore, this property is not stable, Definition~\ref{def:stable}.
When restricting information functions to such sets of probability mass functions, one thus loses the monoid action that makes use of conditioning.
Since the monoid action is crucially used in our main technique --- subset determination, Theorem~\ref{thm:subset_determination} --- we expect such a theory to take a different route from the one for $\CR$--Markov random fields.
Nevertheless,~\cite{Steudel2010}, Theorem 1, provides an information characterization of causal Bayesian networks for submodular information functions, which can be used as an inspiration.

\subsubsection*{\texorpdfstring{$O$}--Information and \texorpdfstring{$S$}--Information}

In~\cite{Rosas2019} and~\cite{Medina-Mardones2021}, the $O$-information and $S$-information were defined as the difference and sum of the classical total correlation and dual total correlation, respectively.
We have showed in this work that $\CR$-(dual )total correlation preserves its value in studying conditional mutual $\CR$-independences, and we therefore expect that $\CR$--$O$-information and $\CR$--$S$-information might provide useful insights into general characterizations of high-order interdependencies.

\subsubsection*{(Cluster) Cross-Entropies and Kullback-Leibler Divergence}

We think it is worthwhile to study the case where $\CR$ is cross-entropy or Kullback-Leibler divergence in greater detail.
After all, much of machine learning and deep learning involves the minimization of a cross-entropy, or equivalently Kullback-Leibler divergence~\citep{Bishop2007,Bishop2023}.
This becomes especially interesting for graphical methods, including diffusion models~\citep{Dickstein2015} that form the basis for widespread text-to-image generation methods like Dalle~\citep{Ramesh2021}, Imagen~\citep{Chitwan2022}, and stable diffusion~\citep{Stable_Diffusion2021}.
Diffusion models involve a decomposition of a joint Kullback-Leibler divergence over a Markov chain as in Figure~\ref{fig:four_and_five_circles}.
It could be valuable to analyze the loss functions of more such graphical methods using Kullback-Leibler and cross-entropy-diagrams. 
Similarly, the (cluster) cross-entropies (see Section~\ref{sec:functions_satisfying_chain_rule}) used in adaptive cluster expansion~\citep{Cocco2012} of Ising models, which form a Markov random field, deserve a further study.

\subsection*{Conclusion}

In this work we showed that it is possible to use the framework of general $\CR$-diagrams to study generalized notions of independences and Markov random fields.
In the process, we generalized well-known notions like the (dual) total correlation and developed new methods such as subset determination.
We were able to apply the general theory to the probabilistic case, and in particular to Kullback-Leibler diagrams on Markov chains. 
We think it is worthwhile to look into research areas such as machine learning, where graphical models and information functions such as cross-entropy are widespread, and to apply our generalized information-theoretic insights to such settings.

\newpage

\appendix

\section{Cohomological Characterization of Functions Satisfying the Chain Rule}\label{sec:classification_CR}

In this appendix we expand on Remark~\ref{rem:connection_to_F_classification} and show how functions $\CR$ satisfying the chain rule can be classified. 
We will also interpret that result in terms of the vanishing of a cohomology group of degree 1.
In this whole section, let $\mon{M}$ be a commutative, idempotent monoid acting additively on an abelian group $\gro{G}$.

\begin{definition}
  [Bounded Monoid, Top Element]
  \label{def:top_element}
  $\mon{M}$ is called \emph{bounded} if it contains a \emph{top element}, i.e., an element $\top \in \mon{M}$
  such that $X \cdot \top = \top$ for all $X \in \mon{M}$.

  With the relation $\precsim$ on $\mon{M}$ defined by $X \precsim Y$ if and only if $X \cdot Y = Y$, we see that a top element is equivalently described as the greatest element in $\mon{M}$.
\end{definition} 

Clearly, top elements are unique.

\begin{notation}
  We denote by
  \begin{equation*}
    \CRS(\mon{M}, \gro{G}) \coloneqq \Big\lbrace \CR: \mon{M} \to \gro{G} \  \big| \  \forall X, Y \in \mon{M}: \ \CR(XY) = \CR(X) + X.\CR(Y)\Big\rbrace
  \end{equation*}
  the set of functions satisfying the chain rule Equation~\eqref{eq:cocycle_equationn}.
\end{notation}

Note that $\CRS(\mon{M}, \gro{G})$ is itself an abelian group with addition
\begin{equation*}
  (\CR+\CR')(X) \coloneqq \CR(X) + \CR'(X).
\end{equation*}
Additionally, it carries an induced and well-defined additive monoid action $.: \mon{M} \times \CRS(\mon{M}, \gro{G}) \to \CRS(\mon{M}, \gro{G}) $ given by
\begin{equation*}
  (X.\CR)(Y) \coloneqq X.\CR(Y).
\end{equation*}

\begin{notation}
  Let $X \in \mon{M}$ be any element.
  We denote by
  \begin{equation*}
    \gro{G}^X \coloneqq \Big\lbrace g \in \gro{G} \  \big| \  X.g = 0 \Big\rbrace
  \end{equation*}
  the elements in $\gro{G}$ that are annihilated by $X$.
  This is a subgroup of $\gro{G}$.
  Furthermore, the action of $\mon{M}$ on $\gro{G}$ restricts to a well-defined additive monoid action $\mon{M} \times \gro{G}^X \to \gro{G}^X$.
\end{notation}

\begin{definition}
  [Module Homomorphism, Module Isomorphism]
  \label{def:module_homomorphism}
  Let $\gro{G}, \gro{H}$ both be groups carrying an additive monoid action from $\mon{M}$, which by abuse of notation we denote with the same symbol: 
  \begin{align*}
    .: \mon{M} \times \gro{G} \to \gro{G}, \quad .: \mon{M} \times \gro{H} \to \gro{H}. 
  \end{align*}
  A function $\Phi: \gro{G} \to \gro{H}$ is called a \emph{module homomorphism} if
  \begin{itemize}
    \item $\Phi$ is a group homomorphism: $\Phi(g+g') = \Phi(g) + \Phi(g')$ for all $g, g' \in \gro{G}$;
    \item $\Phi$ commutes with the monoid actions: $\Phi(X.g) = X.\Phi(g)$ for all $X \in \mon{M}$, $g \in \gro{G}$.
  \end{itemize}
  A \emph{module isomorphism} is a bijective module homomorphism.
\end{definition}

The following proposition shows that the functions $\CR: \mon{M} \to \gro{G}$ that satisfy the chain rule are, for bounded $\mon{M}$ with top element $\top$, essentially the same as the elements in $\gro{G}$ that are annihilated by $\top$.

\begin{proposition}
  \label{pro:description_cocycles}
  Let $\mon{M}$ be a bounded, commutative, idempotent monoid with top element $\top$, $G$ an abelian group, and $.: \mon{M} \times \gro{G} \to \gro{G}$ an additive monoid action.
  Define the pair of functions
  \begin{equation*}
    \begin{tikzcd}
      \CRS(\mon{M}, \gro{G}) \ar[rr, bend left, "\Phi"] & & \gro{G}^{\top} \ar[ll, bend left, "\Psi"]
    \end{tikzcd}
  \end{equation*}
  as follows: 
  \begin{align*}
    &  \forall \CR \in \CRS(\mon{M},\gro{G}):  \quad  \Phi(\CR) \coloneqq  \CR(\top); \\
    &  \forall g \in \gro{G}^{\top}: \quad  \Psi(g): \mon{M} \to \gro{G}, \ \  \big[\Psi(g)\big](X) \coloneqq g - X.g.
  \end{align*}
  Then $\Phi$ and $\Psi$ are mutually inverse module isomorphisms.
\end{proposition}

\begin{proof}
  For $\Phi$, we need to check well-definedness, i.e., that $\CR(\top) \in \gro{G}^{\top}$ for all $\CR \in \CRS(\mon{M}, \gro{G})$.
  Indeed, we have
  \begin{equation*}
    \CR(\top) = \CR(\top \top) = \CR(\top) + \top.\CR(\top)
  \end{equation*}
  by the chain rule, from which $\top.\CR(\top) = 0$ follows.
  It is clear that $\Phi$ is a module homomorphism.

  For $\Psi$, we also need to check that it is well-defined, i.e., that $\Psi(g)$ satisfies the chain rule for all $g \in \gro{G}^{\top}$.
  Indeed, we have
  \begin{align*}
    \big[ \Psi(g)\big](XY) & = g - (XY).g\\
    & = g - X.g + X.g - X.(Y.g) \\
    & = \big[\Psi(g) \big](X) + X.\big[ g - Y.g\big] \\
    & = \big[ \Psi(g)\big](X) + X.\big[\Psi(g)\big](Y).\footnotemark
  \end{align*}
  \footnotetext{Actually, the assumption $g \in \gro{G}^\top$ was not used in this computation. It works for all $g \in \gro{G}$, which we use in Remark~\ref{rem:cohomological_interpretation}}That $\Psi$ is a module homomorphism is clear.

  It remains to show that $\Psi$ and $\Phi$ are inverse to each other.
  By definition of $\gro{G}^{\top}$, we have for all $g \in \gro{G}^{\top}$:
  \begin{align*}
    \big( \Phi \circ \Psi\big)(g) & = \Phi\big(\Psi(g)\big) = \big[\Psi(g)  \big](\top) = g - \top.g = g.
  \end{align*}
  In the other direction, let $\CR \in \CRS(\mon{M}, \gro{G})$ arbitrary.
  For all $X \in \mon{M}$, we have $X \cdot \top = \top$ by definition of $\top$ and therefore
  \begin{equation*}
    \CR(\top) = \CR(X \cdot \top) = \CR(X) + X.\CR(\top)
  \end{equation*}
  by the chain rule.
  It follows
  \begin{equation*}
    \Big[\big( \Psi \circ \Phi \big)(\CR) \Big] (X) = \big[ \Psi\big(\Phi(\CR)\big)\big](X) 
    = \big[ \Psi\big(\CR(\top)\big)\big] (X) 
    = \CR(\top) - X.\CR(\top) 
    = \CR(X).
  \end{equation*}
  This means $\big( \Psi \circ \Phi\big)(\CR) = \CR$, and we are done.
\end{proof}

\begin{remark}
  [Cohomological Interpretation]
  \label{rem:cohomological_interpretation}
  For the interested reader, we interpret the preceding result in cohomological terms. 
  Namely, consider the following cochain complex:
  \begin{equation*}
    \begin{tikzcd}
      \gro{G} = \Maps(\mon{M}^0, \gro{G}) \ar[r, "\Psi"] & \Maps\big( \mon{M}, \gro{G} \big) \ar[r, "\delta"] & \Maps\big( \mon{M}^2, \gro{G} \big) \ar[r] & \dots
  \end{tikzcd}
  \end{equation*}
  Here, define $\Psi$ as in Proposition~\ref{pro:description_cocycles} by $\big[\Psi(g)\big](X) \coloneqq g - X.g$ and $\delta$ by
  \begin{equation*}
    \big[\delta(\CR)\big](X; Y) \coloneqq X.\CR(Y) - \CR(XY) + \CR(X).
  \end{equation*}
  These are Hochschild coboundary maps, as originally defined in~\cite{Hochschild1945}, and also used for information cohomology in~\cite{Baudot2015a, Vigneaux2019a}.
  Now, the fact that $\Psi$ is well-defined in Proposition~\ref{pro:description_cocycles}, i.e.\ that it only maps to functions satisfying the chain rule,\footnote{Which, as we remarked, holds for $\Psi$ defined on all of $\gro{G}$.} can now be expressed equivalently by saying that $\delta \circ \Psi = 0$, a crucial property for cochain complexes.
  The reason for this is that $\mathrm{Ker}(\delta) = \CRS(\mon{M}, \gro{G})$ is precisely the set of functions satisfying the chain rule.
  This allows to define the first cohomology group of the complex, given by $\mathrm{Ker}(\delta) / \IM(\Psi)$.

  Now, the fact that $\Psi$ is surjective in the preceding proposition --- meaning that it hits every function satisfying the chain rule --- can now be expressed with $\IM(\Psi) = \mathrm{Ker}(\delta)$, meaning the first cohomology group vanishes: $\mathrm{Ker}(\delta)/\IM(\Psi) = 0$. 
  This is shown by explicitly constructing the inverse $\Phi$.

  Finally, the preceding proposition also shows that $\Psi$ is injective when restricting to $\gro{G}^{\top}$.

  Now, consider the case that $\mon{M}$ is the monoid of equivalence classes of random variables on $\samp$ and $\gro{G} = \Meas\big( \Sim{\samp}, \R \big)$.
  Then the preceding discussion would suggest that Shannon entropy $\Ent$ disappears in the first cohomology group.
  At first sight, one might think this contradicts the fundamental result of information cohomology, which shows that the first cohomology group is non-trivial and generated by Shannon entropy~\citep{Baudot2015a}.
  The difference is explained by noting that~\cite{Baudot2015a} also require their cochains to satisfy the \emph{joint locality} property, which means that their cochains of degree zero are necessarily constant.
\end{remark}

\section{Conditional Mutual \texorpdfstring{$\CR$}--Independences and \texorpdfstring{$\CR$}--Total Correlation}\label{sec:total_correlation_charac}

In this appendix, we provide a characterization of conditional mutual $\CR$-independences using $\CR$--total correlation.
Different from the characterization using $\CR$--dual total correlation in Theorem~\ref{thm:charac_using_dual}, this characterization will only work when imposing an additional assumption on $\gro{G}$, namely that it is torsion-free.
The reason is that $\CR$--total correlation double counts atoms in the $\CR$-diagram; from a vanishing of the $\CR$--total correlation, one can then a priori only conclude that a \emph{multiple} of each of the contained atoms disappears, and one needs $\gro{G}$ to be torsion-free to be able to reduce the coefficient to $1$.

As in the main text, we fix a commutative, idempotent monoid $\mon{M}$ acting additively on an abelian group $\gro{G}$ and a function $\CR: \mon{M} \to \gro{G}$ satisfying the chain rule Equation~\eqref{eq:cocycle_equationn}.

\begin{proposition}
  \label{pro:recursive_formula_total_correlation}
  Let $X_1, \dots, X_n \in \mon{M}$.
  Then for all $\emptyset \neq I \subseteq \start{n}$, we have the following identities relating $\CR$--total correlation to higher $\CR$-interactions:
  \begin{enumerate}
    \item $\TC{\CR}\big( \bigscolon_{i \in I} X_i \big)  = \sum_{\emptyset \neq L \subseteq I} \big( \num{L} - 1\big) \cdot X_{I \setminus L}.\CR\big(\bigscolon_{l \in L} X_l \big)$;
    \item $\big( \num{I} - 1\big) \cdot \CR\big(\bigscolon_{i \in I} X_i \big) = \sum_{\emptyset \neq L \subseteq I} (-1)^{\num{I} - \num{L}} \cdot X_{I \setminus L}. \TC{\CR}\big( \bigscolon_{l \in L} X_l \big)$.
  \end{enumerate}
\end{proposition}

\begin{proof}
  To prove part 1, let the $X_i$, $i \in I$ be the full set of elements in Theorem~\ref{thm:hu_kuo_ting_generalized}.
  Then we obtain
  \begin{align*}
    \TC{\CR}\big( \bigscolon_{i \in I} X_i \big) & = \sum_{i \in I} \CR(X_i) - \CR(X_I)  \\
    & = \sum_{i \in I} \set{\CR}\big( \set{X}_i\big) - \set{\CR}\big(  \set{X}_I\big) \\
    & = \sum_{i \in I} \sum_{L \subseteq I, i \in L} \set{\CR}(\atom{L}) - \sum_{\emptyset \neq L \subseteq I} \set{\CR}(\atom{L}) \\
    & = \sum_{\emptyset \neq L \subseteq I} \sum_{i \in L} \set{\CR}(\atom{L}) - \sum_{\emptyset \neq L \subseteq I} \set{\CR}(\atom{L}) \\
    & = \sum_{\emptyset \neq L \subseteq I} \big( \num{L} - 1\big) \cdot \set{\CR}(\atom{L}).
  \end{align*}
  Now, note that $\set{\CR}(\atom{L}) = X_{I \setminus L}.\CR\big(\bigscolon_{l \in L} X_l \big)$ by Lemma~\ref{lem:how_atoms_look_like}. 
  The result follows.

  To prove part 2, we evaluate the right-hand-side using part 1 on each summand:
  \begin{align*}\label{eq:sum_of_TC}
     \sum_{\emptyset \neq L \subseteq I} (-1)^{\num{I} - \num{L}} & \cdot X_{I \setminus L}. \TC{\CR}\big( \bigscolon_{l \in L} X_l \big) \\
    & =  \sum_{\emptyset \neq L \subseteq I} (-1)^{\num{I} - \num{L}} \cdot X_{I \setminus L}.\Bigg( \sum_{\emptyset \neq K \subseteq L} \big( \num{K} - 1\big) \cdot X_{L \setminus K}.\CR\big(\bigscolon_{k \in K} X_k \big)\Bigg) \\
     & = \sum_{\emptyset \neq L \subseteq I} \sum_{\emptyset \neq K \subseteq L} (-1)^{\num{I} - \num{L}} \cdot \big(\num{K} - 1 \big) \cdot X_{I \setminus K}.\CR\big(\bigscolon_{k \in K} X_k \big) \\
     & = \sum_{\emptyset \neq K \subseteq I} (-1)^{\num{I}} \cdot \big( \num{K} - 1\big) \cdot
     \Bigg(\sum_{L: \ K \subseteq L \subseteq I} (-1)^{\num{L}} \Bigg) \cdot X_{I \setminus K}.\CR\big(\bigscolon_{k \in K} X_k \big).\\
     & = \sum_{\emptyset \neq K \subseteq I} (-1)^{\num{I}} \cdot \big( \num{K} - 1\big) \cdot (-1)^{\num{K}} \cdot \indic{I = K} \cdot X_{I \setminus K}.\CR\big(\bigscolon_{k \in K} X_k \big) \\
    & = (-1)^{2 \cdot \num{I}} \cdot \big( \num{I} - 1\big) \cdot X_{I \setminus I}.\CR\big(\bigscolon_{i \in I} X_i \big) \\
    & = \big( \num{I} - 1\big) \cdot \CR\big(\bigscolon_{i \in I} X_i \big).
  \end{align*}
  In the fourth step, we used Lemma~\ref{lem:pure_combinatorics}.
\end{proof}

\begin{definition}
  [Torsion-Free Abelian Group]
  \label{def:torsion_free}
  Let $\gro{G}$ be an abelian group.
  Then $\gro{G}$ is called \emph{torsion-free} if $0 \in \gro{G}$ is the only element of finite order.
  In other words, for all $0 \neq g \in \gro{G}$ and all $0 < k \in \N$, we have $k \cdot g \neq 0$.
  Equivalently, for all $k \in \N$ and $g \in \gro{G}$, if $k \cdot g = 0$ then $k = 0$ or $g = 0$.
\end{definition}

\begin{example}
  The groups $\Z$, $\R$, $\Meas\big( \Sim{\samp}, \R \big)$ and all other groups encountered in Section~\ref{sec:specializing_to_probabilistic} are torsion-free.
  $\Z/k\Z$ for $k \geq 1$ is a prototypical example of an abelian group with torsion:
  $k \cdot 1 = 0$.
\end{example}

\begin{theorem}
  \label{thm:inf_charac_mutual_independence}
  Let $\mon{M}$ be a commutative, idempotent monoid acting additively on a \emph{torsion-free} abelian group $\gro{G}$.
  Let $\CR: \mon{M} \to \gro{G}$ be a function satisfying the chain rule Equation~\eqref{eq:cocycle_equationn}.
  Let $X_1, \dots, X_n, Y \in \mon{M}$. 
  Then the following properties are equivalent:
  \begin{enumerate}
    \item $\bigindepF{\CR}_{i = 1}^{n} X_i \ | \ Y$;
    \item $Y.\TC{\CR}\big( \bigscolon_{i \in \start{n}} X_i \big) = 0$;
    \item $(Y X_{\start{n} \setminus I}).\TC{\CR}\big( \bigscolon_{i \in I} X_i \big) = 0$ for all $\emptyset \neq I \subseteq \start{n}$.
  \end{enumerate}
\end{theorem}

\begin{proof}
  Assume 1. 
  To prove 2, we use that by Proposition~\ref{pro:mutual_independence_characterization}, we have $\IndepF{\CR}{X_i}{X_{\start{i-1}}}{Y}$
  for all $i = 1, \dots, n$.
  By Proposition~\ref{pro:pairwise_independence_characterization}, this implies ${Y.\CR\big(X_{\start{i}}\big) = Y.\CR\big(X_{\start{i-1}}\big) + Y.\CR(X_{i})}$.
  Inductively, we obtain ${Y.\CR\big(X_{\start{n}}\big) = \sum_{i = 1}^{n} Y.\CR(X_{i})}$
  and therefore $Y.\TC{\CR}\big(\bigscolon_{i \in \start{n}} X_i\big) = 0$.

  Now assume 2.
  To prove 3, we note
  \begin{align*}
    0 & = X_{\start{n} \setminus I}.\Big(Y.\TC{\CR}\big(\bigscolon_{i \in \start{n}} X_i \big) \Big) \\
    & = Y. \Bigg( X_{\start{n} \setminus I}. \bigg( \sum_{i = 1}^{n} \CR(X_i) - \CR\big( X_{\start{n}}\big)\bigg) \Bigg) \\
    & = Y.\Bigg( \sum_{i = 1}^{n} X_{\start{n} \setminus I}.\CR(X_i) - X_{\start{n} \setminus I}.\CR\big( X_{\start{n}}\big)\Bigg) \\
    & \overset{(\star)}{=} Y. \Bigg( \sum_{i \in I} X_{\start{n} \setminus I}.\CR(X_i) - X_{\start{n} \setminus I}.\CR( X_I)\Bigg) \\
    & = \big( YX_{\start{n} \setminus I}\big). \Bigg( \sum_{i \in I} \CR(X_i) - \CR(X_I)\Bigg) \\
    & = \big( YX_{\start{n} \setminus I}\big). \TC{\CR}\big( \bigscolon_{i \in I} X_i \big).
  \end{align*}
  Step $(\star)$ can easily be seen using Hu's Theorem~\ref{thm:hu_kuo_ting_generalized}.
  That was to show.

  Finally, assume 3.
  For proving 1, we can alternatively prove the equivalent property 3 of Theorem~\ref{thm:charac_using_dual}.
  We note that for all $\emptyset \neq I \subseteq \start{n}$, Proposition~\ref{pro:recursive_formula_total_correlation} implies:
  \begin{align*}
    \big( \num{I} - 1\big) \cdot \big( YX_{\start{n} \setminus I}\big).\CR\big(\bigscolon_{i \in I} X_i \big) 
    & = \big( YX_{\start{n} \setminus I}\big) . \Big( \big( \num{I} - 1\big) \cdot \CR\big( \bigscolon_{i \in I} X_i\big)\Big) \\
    & = \big( YX_{\start{n} \setminus I}\big). \Bigg( \sum_{\emptyset \neq L \subseteq I} (-1)^{\num{I} - \num{L}} \cdot X_{I \setminus L}.\TC{\CR}\big( \bigscolon_{l \in L} X_l \big)\Bigg) \\
    & = \sum_{\emptyset \neq L \subseteq I} (-1)^{\num{I} - \num{L}} \cdot \big( YX_{\start{n} \setminus L}\big).\TC{\CR}\big( \bigscolon_{l \in L} X_l \big) \\
    & = 0.
  \end{align*}
  If $\num{I} = 1$ then the factor $\num{I} - 1$ vanishes and what we showed is vacuous.
  If $\num{I} \geq 2$, then what we showed implies
  \begin{equation*}
    \big( YX_{\start{n} \setminus I}\big).\CR\big(\bigscolon_{i \in I} X_i \big) = 0
  \end{equation*}
  since $\gro{G}$ is torsion-free, and we are done.
\end{proof}

We now show in an example that we cannot omit the assumption that $\gro{G}$ is torsion-free in the preceding theorem.
  More precisely, if $\gro{G}$ is not torsion-free, then property 2 does not necessarily imply property 1 anymore. 

\begin{example}\label{exa:torsion_example}
  Let $\mon{M} = \big( \Z/2\Z, \cdot\big)$, where $\Z/2\Z = \{0, 1\}$.
  This is a commutative, idempotent monoid with the following multiplication rules:
  \begin{equation*}
    1 \cdot 1 = 1; \quad 1 \cdot 0 = 0; \quad 0 \cdot 1 = 0; \quad 0 \cdot 0 = 0.
  \end{equation*}
  Furthermore, define $\gro{G} \coloneqq \big( \Z/2\Z, +\big)$, which is an abelian group with the following addition rules:
  \begin{equation*}
    0 + 0 = 0; \quad 0 + 1 = 1; \quad 1 + 0 = 1; \quad 1 + 1 = 0.
  \end{equation*}
  The last rule implies that $\gro{G}$ has torsion.
  We define the action of $M$ on $\gro{G}$ by $X.g \coloneqq X \cdot g$
  for all $X \in \mon{M}$ and $g \in \gro{G}$.
  That is, the action is simply the usual multiplication in the ring $\big( \Z/2\Z, +, \cdot\big)$ and therefore clearly an additive monoid action.
  Finally, define
  \begin{equation*}
    \CR: \mon{M} \to \gro{G}, \quad \CR(X) \coloneqq 1 - X.
  \end{equation*}
  We have
  \begin{equation*}
    \CR(XY) = 1 - XY = 1 - X + X - XY = 1 - X + X \cdot (1 - Y) = \CR(X) + X.\CR(Y),
  \end{equation*}
  implying that $\CR$ satisfies the chain rule.
  Thus, $\CR$ satisfies all the properties of Hu's Theorem~\ref{thm:hu_kuo_ting_generalized}.
  
  Now, set $X_1 \coloneqq X_2 \coloneqq X_3 \coloneqq 0 \in M$.
  Then we have
  \begin{equation*}
    \TC{\CR}(X_1;X_2;X_3) = \CR(0) + \CR(0) + \CR(0) - \CR(0 \cdot 0 \cdot 0) = 1 + 1 = 0.
  \end{equation*}
  However, we have
  \begin{equation*}
    \CR(X_3; X_1X_2) = \CR(0; 0) = \CR(0) = 1.
  \end{equation*}
  This means that the $\CR$-independence $\big(X_1X_2\big) \indep_{\CR} X_3$ does \emph{not} hold, implying that $X_1, X_2, X_3$ are \emph{not} mutually independent.
  Thus we cannot omit the assumption of torsion-freeness in the Theorem.

  Note that in light of Proposition~\ref{pro:description_cocycles}, the construction of $\CR$ can be explained as follows:
  $\mon{M}$ is a bounded monoid with top element $0 \in \mon{M}$.
  Then, $\gro{G}^{\top} = \gro{G}$, and so $1 \in \gro{G}$ is one of the annihilated elements.
  We then see:
  \begin{equation*}
    \big[ \Psi(1) \big](X) = 1 - X.1 = 1 - X = \CR(X),
  \end{equation*}
  showing that $\CR = \Psi(1)$.
  The only other function satisfying the chain rule, $\Psi(0)$, is trivial, and does therefore not give rise to a counter example.
\end{example}

\section{General Consequences of Section~\ref{sec:yeung2002}}\label{sec:generalizing_yeung_2019}

In this appendix, we take a look at~\cite{Yeung2019} and explain which results generalize to our setting. 
As it turns out, all these results generalize without problem with the same proofs, or otherwise cannot even be formulated in our case since the formulations are based on inequalities, which does not make sense in the context of general abelian groups. 

Fix a commutative, idempotent monoid $\mon{M}$ acting on an abelian group $\gro{G}$ and a function $\CR: \mon{M} \to \gro{G}$ satisfying the chain rule.
We fix elements $X_1, \dots, X_n \in \mon{M}$, giving rise to $\set{\CR}: 2^{\set{X}} \to \gro{G}$ by Hu's Theorem~\ref{thm:hu_kuo_ting_generalized}, and a graph $\Gr{G} = (\Ver{V}, \Ed{E})$ with $\Ver{V} = \start{n}$.

\subsubsection*{Smallest Graph Representations}

\begin{terminology}
  If $X_1, \dots, X_n$ form an $\CR$-Markov random field with respect to $\Gr{G}$, then we also say that $\Gr{G}$ is a (graph) \emph{representation} for $X_1, \dots, X_n$ with respect to $\CR$. 
\end{terminology}

\begin{definition}
  [Subgraph]
  \label{def:subgraph}
  If $\Gr{G}' = (\Ver{V}', \Ed{E}')$ is a second graph with $\Ver{V}' \subseteq \Ver{V}$ and $\Ed{E}' \subseteq \Ed{E}$, then we say that $\Gr{G}'$ is a \emph{subgraph} of $\Gr{G}$.
\end{definition}

\begin{definition}
  [Smallest Graph Representation]
  \label{def:smallest_graph_rep}
  $\Gr{G}'$ is said to be the \emph{smallest graph representation} for $X_1, \dots, X_n$ (with respect to $\indep_{\CR}$) if $\Gr{G}'$ is a representation for $X_1, \dots, X_n$ and a subgraph of all other representations for $X_1, \dots, X_n$.
\end{definition}

\begin{theorem}
  [Smallest Graph Representation]
  \label{thm:smallest_graph_representation}
  Define $\widehat{\Gr{G}} \coloneqq \big(\start{n}, \widehat{\Ed{E}}\big)$ with $\{i, j\} \in \widehat{\Ed{E}}$ if and only if
  \begin{equation}\label{eq:edge_condition}
    X_{\start{n} \setminus \{i, j\}}.\CR(X_i; X_j) = \set{\CR}\big(\atom{\{i, j\}}\big) \neq 0.
  \end{equation}
  If $X_1, \dots, X_n$ has a smallest graph representation, then it equals $\widehat{\Gr{G}}$.
\end{theorem}

\begin{proof}
  This is simply~\cite{Yeung2019}, Theorem 3, generalized to our setting. 
  Note that the first equality in Equation~\eqref{eq:edge_condition} is simply Lemma~\ref{lem:how_atoms_look_like} and always holds.
  The original proof can be copied word for word.
  Only once, an ``inequality sign'' needs to be replaced by an ``unequal sign''.
\end{proof}

\subsubsection*{Smallest Graphs for Subfields of Markov Random Fields}

\begin{definition}
  [Marginalization of Graph]
  \label{def:marginalization}
  Let $\Ver{V}' \subseteq \Ver{V}$.
  Then the Marginalization of $\Gr{G}$ for $\Ver{V}'$ is defined as $\Gr{G}^*(\Ver{V}') \coloneqq (\Ver{V}', \Ed{E}')$ with
  $\{i, j\} \in \Ed{E}'$ if and only if there is a walk from $i$ to $j$ in $\Gr{G}$ with all intermediate vertices in $\Ver{V} \setminus \Ver{V}'$.
\end{definition}

\begin{notation}
  Let $\Gr{G}' = (\Ver{V}', \Ed{E}')$ be a graph with $\Ver{V}' \subseteq \Ver{V} = \start{n}$. 
  Then we write $\Gr{G} \Longrightarrow \Gr{G}'$ if for all $Y_1, \dots, Y_n$ that form an $\CR$-Markov random field with respect to $\Gr{G}$, $Y_i, i \in \Ver{V}'$ form an $\CR$-Markov random field with respect to $\Gr{G}'$. 
  Note that ``$\Gr{G} \Longrightarrow \Gr{G}'$'' is not a statement about graphs alone, as it depends on $\CR: \mon{M} \to \gro{G}$.
\end{notation}

\begin{theorem}
  [Graphs for Subfields]
  \label{thm:subfield_graphs}
  Assume that there is at least one element $Z \in \mon{M}$ with $\CR(Z) \neq 0$.
  Let $\Ver{V}' \subseteq \Ver{V}$ be a subset.
  Then $\Gr{G}^*(\Ver{V}')$ is the smallest graph $\Gr{G}'$ such that $\Gr{G} \Longrightarrow \Gr{G}'$.
\end{theorem}

\begin{proof}
  This is precisely~\cite{Yeung2019}, Theorem 8. 
  The proof is exactly the same, except that the condition $\Ent(Z) > 0$ from the original paper is replaced by $\CR(Z) \neq 0$, and the constant random variable is replaced by $\one \in \mon{M}$.
\end{proof}

\begin{remark}
  \label{rem:philosophy}
  If we drop the condition that there is $Z \in \mon{M}$ with $\CR(Z) \neq 0$, then the conclusion is wrong.
  Indeed, if $\CR$ is trivial, then all $\CR$-independences always hold, meaning that the graph $\Gr{G}' \coloneqq (\Ver{V}', \emptyset)$ with empty edge set would be the smallest graph with $\Gr{G} \Longrightarrow \Gr{G}'$.
  The existence of $Z$ with $\CR(Z) \neq 0$ ensures that we can for all edges in $\Gr{G}^*\big(\Ver{V}'\big)$ construct an $\CR$-Markov random field for $\Gr{G}$ that is \emph{faithful} to that edge by having a corresponding dependence.
\end{remark}

\subsubsection*{Rewriting the Values of Atoms in $\CR$-Markov Random Fields}

Section 5 in~\cite{Yeung2019} discusses a characterization for forests of paths.
Namely, $\Gr{G} = (\Ver{V}, \Ed{E})$ is a forest of paths if and only if for all probability mass functions $P$ and all random variables $X_1, \dots, X_n$ that form a Markov random field with respect to $\Gr{G}$ and $P$, the corresponding $\Ent$-diagram is nonnegative.
These results cannot be stated in our generalized setting since $\CR$ takes values in an unspecified abelian group $\gro{G}$ that may not have an associated order relation to begin with.
However, that section still contains an interesting result worth stating in our generalized setting:

\begin{theorem}
  [Atoms in Markov Random Fields]
  \label{thm:atoms_in_mrfs}
  Assume that $X_1, \dots, X_n$ form an $\CR$-Markov random field with respect to $\Gr{G}$.
  Let $\atom{\Ver{I}}$ be a connected atom, where $\Ver{I} \subseteq \start{n}$ has at least two elements.
  Define
  \begin{equation*}
    \Ver{B} \coloneqq \Big\lbrace i \in \Ver{I} \ \big| \  \atom{\Ver{I} \setminus i} \text{ is connected}  \Big\rbrace.
  \end{equation*}
  Then one has
  \begin{equation*}
    \set{\CR}(\atom{\Ver{I}}) = X_{\Ver{V} \setminus \Ver{I}}.\CR\big( \bigscolon_{i \in \Ver{I}} X_i \big) = X_{\Ver{V} \setminus \Ver{I}}.\CR\big( \bigscolon_{i \in \Ver{B}} X_i \big).
  \end{equation*}
\end{theorem}

\begin{proof}
  Note that the first equality is simply Lemma~\ref{lem:how_atoms_look_like}.
  The second equality is~\cite{Yeung2019}, Theorem 11.
  The original proof generalizes word for word to our setting.
\end{proof}

\begin{remark}
  \label{rem:about_atoms_mrfs}
  The preceding theorem deals with connected atoms.
  For disconnected atoms $\atom{\Ver{I}}$, Theorem~\ref{thm:mrf_characterization} shows $\set{\CR}(\atom{\Ver{I}}) = 0$.

  Furthermore, note that Theorem~\ref{thm:atoms_in_mrfs} generalizes a special case of Corollary~\ref{cor:information_terms_markov_chains}:
  that corollary implies that for an $\CR$-Markov chain $X_1, \dots, X_n$ and an ``interval'' $I = \range{i}{j} \subseteq \start{n}$ for $i < j$, one has
  \begin{equation*}
    X_{\start{n} \setminus I}.\CR\big( \bigscolon_{i \in I} X_i\big) = X_{\start{n} \setminus I}.\CR(X_i; X_j).
  \end{equation*}
\end{remark}

\subsubsection*{Trees and Drawing $\CR$-Diagrams for $\CR$-Markov Random Fields}

The remaining results in~\cite{Yeung2019} are essentially graph theoretic in nature and not dependent on $\CR$.
In Theorem 9, they discuss when the marginalization of a graph is a tree. 
Finally, in Section 6, they discuss an algorithm for drawing an $\Ent$-diagram, with two requirements relating to the graph $\Gr{G}$: connected atoms are depicted as nonempty and disconnected atoms as empty.
These diagrams are then suitable to express any $\CR$-diagram of any $\CR$-Markov random field corresponding to $\Gr{G}$, as Theorem~\ref{thm:mrf_characterization} shows.
For the case of path-shaped graphs corresponding to Markov chains, we saw examples of such depictions in Figure~\ref{fig:four_and_five_circles}.
The algorithm in Section 6 of~\cite{Yeung2019} works for any graph $\Gr{G}$.

\section{Slices of \texorpdfstring{$\Ent$}--Diagrams}\label{sec:slices}

Our generalizations of the main results from~\cite{Yeung2002a}, Theorems~\ref{thm:FCMIs_characterization} and~\ref{thm:mrf_characterization}, work in the framework of $\CR$-independence, where $\CR$ satisfies the full set of assumptions in Hu's Theorem~\ref{thm:hu_kuo_ting_generalized}.
In contrast, the original results were formulated for the $P$-independence with respect to a fixed probability mass function $P$, thus engaging only with one ``slice'' of the $\Ent$-diagrams, as we explained in Equation~\eqref{eq:slice_function}.
A priori, it may seem unclear whether these original results can actually be deduced from our supposed generalizations.
In this section, we prove that this is indeed possible.

\begin{theorem}
  [Characterization of $P$-FCMIs]
  \label{thm:charac_probabilistic_fcmis}
  Let $\Ent: \mon{M} \to \Meas\big( \Sim{\samp}, \R \big)$ be the Shannon entropy function.
  Let $X_1, \dots, X_n$ be random variables on $\samp$ and $P \in \Sim{\samp}$ a probability mass function, giving rise to $\set{I}^P: 2^{\set{X}} \to \R$ by Equation~\eqref{eq:slice_function}.
  Let $\FCMI{K} = \big(J, L_i, 1 \leq i \leq q \big)$ be a conditional partition of $\start{n}$ with $q \geq 2$.
  Then the following two statements are equivalent:
  \begin{itemize}
    \item $\FCMI{K}$ induces a $P$-FCMI: $\bigindepF{P}_{i = 1}^{q} X_{L_i} \ | \ X_J$;
    \item For all $\atom{W} \in \IM(\FCMI{K})$: $\set{I}^{P}(\atom{W}) = 0$.
  \end{itemize}
\end{theorem}

\begin{proof}
  Let $\Prty{R}$ be the following property on elements $P' \in \Sim{\samp}$:
  \begin{equation*}
    \Prty{R}(P') \quad : \Longleftrightarrow \quad \bigindepF{P'}_{i = 1}^{q} X_{L_i} \ | \ X_J.
  \end{equation*}
  Then by Proposition~\ref{pro:fcmi_stable_conditioning}, and using that it is also well-defined and measurable, this is stable property.
  Consequently, we can define the restricted entropy function $\WP{\Ent}{\Prty{R}}: \WP{\mon{M}}{\Prty{R}} \to \Meas\big( \WP{\Sim{\samp}}{\Prty{R}}, \R \big)$, which satisfies the chain rule according to Proposition~\ref{pro:property_restriction_chain_rule}. 
  Now, we claim that the following conditional mutual independence holds:
  \begin{equation}\label{eq:inferred_mutual_independence}
    \bigindepF{\WP{\Ent}{\Prty{R}}}_{i = 1}^{q} X_{L_i} \ | \ X_J.
  \end{equation}
  By definition, conditional mutual independences are given by a set of pairwise independences, and these in turn by the vanishing of conditional mutual information.
  Thus, the claim reduces to the following:
  for all $i \in \start{q}$, we have
  \begin{equation*}
    X_J.\WP{\Ent}{\Prty{R}}\big( X_{L_i}; X_{L_{\setminus i}}\big) = 0,
  \end{equation*}
  where $L_{\setminus i} \coloneqq \bigcup_{k \neq i} L_k$.
  But this is clear since for all $P'$ with $\Prty{R}(P')$, we have $\IndepF{P'}{X_{L_i}}{X_{L_{\setminus i}}}{X_J}$
  and thus, as is well known, the vanishing of the corresponding conditional mutual information follows:
  \begin{equation*}
    \Big[ X_J.\WP{\Ent}{\Prty{R}}\big( X_{L_i}; X_{L_{\setminus i}}\big) \Big]\big(P'\big) = \Big[ X_J.\Ent\big( X_{L_i}; X_{L_{\setminus i}}\big) \Big]\big(P'\big) = 0.
  \end{equation*}
  This proves the claim, Equation~\eqref{eq:inferred_mutual_independence}.

  Now, assume that the first statement holds, i.e., $\bigindepF{P}_{i = 1}^{q} X_{L_i} \ | \ X_J$.
  Then $\Prty{R}(P)$ holds.
  From Equation~\eqref{eq:inferred_mutual_independence}, we obtain by Theorem~\ref{thm:FCMIs_characterization} that for all $\atom{W} \in \IM(\FCMI{K})$: $\set{\WP{\Ent}{\Prty{R}}}(\atom{W}) = 0$.
  It follows
  \begin{equation*}
    \set{\Ent}^{P}(\atom{W}) = \big[ \set{\Ent}(\atom{W}) \big](P) = \big[ \set{\WP{\Ent}{\Prty{R}}} (\atom{W})  \big](P) = 0,
  \end{equation*}
 which proves one direction.

 For the other direction, assume that $\set{I}^{P}(\atom{W}) = 0$ for all $\atom{W} \in \IM(\FCMI{K})$.
 We want to show $\bigindepF{P}_{i = 1}^{q} X_{L_i} \ | \ X_J$.
 As is well-known, this amounts to showing the following for all $i \in \start{q}$:
 \begin{equation*}
   \Big[ X_J.\Ent\big( X_{L_i}; X_{L_{ \setminus i}} \big) \Big](P) = 0.
 \end{equation*}
 By Hu's Theorem~\ref{thm:hu_kuo_ting_generalized}, we thus need to show $\set{\Ent}^P\big( \set{X}_{L_i} \cap \set{X}_{L_{\setminus i}} \setminus \set{X}_J \big) = 0$.
 But note that, clearly, we have $\set{X}_{L_i} \cap \set{X}_{L_{\setminus i}} \setminus \set{X}_J \subseteq \IM(\FCMI{K})$. 
 Thus, the claim follows from the assumption that $\set{I}^P(\atom{W}) = 0$ for all $\atom{W} \in \IM(\FCMI{K})$, together with the fact that $\set{I}^P$ is a signed measure and therefore additive over disjoint unions.
\end{proof}

Now we prove Yeung's main result, Theorem~\ref{thm:charac_proba_mrfs}, from~\citet{Yeung2002a}:

\begin{proof}[Proof of Theorem~\ref{thm:charac_proba_mrfs}]
  Let the property $\Prty{R}$ be given by $\Prty{R} \coloneqq \Prty{MRF}\big( \Gr{G}; X_1, \dots, X_n \big)$, see the discussion after Equation~\eqref{eq:mrf_property_yeah}.
  Note that in our application of that definition, we have $r = 0$.
  Now, assume the first statement, i.e., $\Prty{R}(P)$ holds.
  We also know from Theorem~\ref{thm:get_mrf_for_free_with_restriction} that $X_1, \dots, X_n$ form an $\WP{\Ent}{\Prty{R}}$-Markov random field with respect to $\Gr{G}$. 
  Then Theorem~\ref{thm:mrf_characterization} implies that $\set{\WP{\Ent}{\Prty{R}}}(\atom{\Ver{W}}) = 0$ for all disconnected atoms $\atom{\Ver{W}}$.
  For these atoms, we then obtain
  \begin{equation*}
    \set{\Ent}^{P}(\atom{\Ver{W}}) = \Big[\set{\Ent}(\atom{\Ver{W}}) \Big](P) = \Big[\set{\WP{\Ent}{\Prty{R}}}(\atom{\Ver{W}}) \Big](P) = 0,
  \end{equation*}
  which proves one direction.

  For the other direction, assume that $\set{\Ent}^P(\atom{\Ver{W}}) = 0$ for all disconnected atoms $\atom{\Ver{W}}$.
  By Proposition~\ref{pro:equivalence_global_markov}, we need to show that $X_1, \dots, X_n$ satisfies the cutset property with respect to $\Gr{G}$ and $\indep_P$.
  Thus, let $\Ver{U}$ be a cutset of $\Gr{G}$ and $\FCMI{K} \coloneqq \big( \Ver{U}, \Ver{V}_i(\Ver{U}), 1 \leq i \leq \ncomp{\Ver{U}} \big)$ the corresponding conditional partition.
  We need to show $\bigindepF{P}_{i = 1}^{\ncomp{\Ver{U}}} X_{\Ver{V}_i(\Ver{U})} \ | \ X_{\Ver{U}}$.
  By Theorem~\ref{thm:charac_probabilistic_fcmis}, we need to show that $\set{\Ent}^P(\atom{\Ver{W}}) = 0$ for all atoms $\atom{\Ver{W}} \in \IM(\FCMI{K})$.
  All such $\atom{\Ver{W}}$ are disconnected by Lemma~\ref{lem:atoms_image_disconnected}, and so the result follows from the assumption. 
\end{proof}

%\newpage

\phantomsection

\addcontentsline{toc}{section}{References}

\bibliographystyle{plainnat}
\bibliography{library}

\end{document}

%% file: math_commands.tex
\usepackage{amsmath, amssymb, amsfonts,bm, amsthm, thmtools, dsfont, mathtools}

\newcommand{\N}{\mathds{N}}
\newcommand{\Z}{\mathds{Z}}

\newcommand{\R}{\mathds{R}}

\declaretheorem[numberwithin=section]{theorem}
\declaretheorem[style=plain, name=Proposition, sibling=theorem]{proposition}
\declaretheorem[style=plain, name=Lemma, sibling=theorem]{lemma}
\declaretheorem[style=plain, name=Corollary, sibling=theorem]{corollary}
\declaretheorem[style=plain, name=Definition, sibling=theorem]{definition}
\declaretheorem[style=plain, name=Example, sibling=theorem]{example}
\declaretheorem[style=plain, name=Notation, sibling=theorem]{notation}
\declaretheorem[style=plain, name=Remark, sibling=theorem]{remark}
\declaretheorem[style=plain, name=Terminology, sibling=theorem]{terminology}

\newcommand{\samp}{\Omega}            % The sample space of an RV
\newcommand{\vs}[1]{E_{#1}}           % The space where an RV takes values.
\newcommand{\ec}[1]{\left[#1\right]}  % Equivalence Class
\newcommand{\one}{\boldsymbol{1}}     % Neutral element monoid
\newcommand{\Asterisk}{\mathop{\scalebox{1.5}{\raisebox{-0.2ex}{$\ast$}}}}   % Large 6-part star
\newcommand{\indep}{\perp\!\!\!\perp} % Independent sign
\newcommand{\Indep}[3]{#1 \indep #2 \  |  \ #3}   % Unspecified Independence of three objects
\newcommand{\IndepF}[4]{#2 \underset{#1}{\indep} #3  \ |  \ #4}  % Specified independence of three objects 
\DeclareMathOperator{\Maps}{Maps}     % Function space
\DeclareMathOperator{\Meas}{Meas}     % Space of measurable functions
\DeclareMathOperator{\CRS}{CR}        % Space of functions satisfying the chain rule.
\newcommand{\mon}[1]{#1}             % Monoid
\newcommand{\gro}[1]{#1}             % Group
\newcommand{\meas}{\mu}              % Generic measure
\newcommand{\set}[1]{\widetilde{#1}} % The set corresponding to an (abstract) variable, or to a CR function
\newcommand{\atom}[1]{p_{#1}}        % the atom of a set/index.
\newcommand{\start}[1]{[#1]}         % The first #1 elements
\newcommand{\range}[2]{\left[#1:#2\right]}    % Numbers from left to right side
\newcommand{\CR}{F}                  % Generic function satisfying the chain rule
\newcommand{\num}[1]{\left| #1\right|}    % The number of elements in a set

\newcommand{\law}[2]{{#2}_{#1}}      % The distributional law of distribution #2 with RV #1
\newcommand{\cond}[2]{{#1}|_{#2}}    % The distribution #1 conditioned on #1
\newcommand{\Sim}[1]{\Delta(#1)}   % The simplex of distributions over #1, a sample space.
\newcommand{\WP}[2]{{#1}_{#2}}      % WP is #1 WITH PROPERTY #2. Used for saying ``markov chain''.
\newcommand{\KL}{D}                  % Symbol for the KL divergence.
\newcommand{\CE}{C} % Symbol for the cross entropy
\newcommand{\Ent}{I}                  % Symbol for Shannon entropy
\newcommand{\AC}[1]{\widetilde{#1}}    % The set of absolutely continuous things
\newcommand{\sep}{\|}                   % A separator between distributions in divergences
\newcommand{\Prty}[1]{\mathcal{#1}}    % A property.

\newcommand{\TC}[1]{TC_{#1}}           % Total Correlation
\newcommand{\DTC}[1]{DTC_{#1}}           % Dual Total Correlation
\newcommand{\indic}[1]{\mathds{1}_{#1}}   % indicator symbol.
\newcommand{\FCMI}[1]{#1}               % Symbol for an FCMI
\DeclareMathOperator{\IM}{Im}          % Symbol for the image of an FCMI
\newcommand{\Gr}[1]{\mathcal{#1}}      % Graph
\newcommand{\Ver}[1]{\mathcal{#1}}       % Vertex set
\newcommand{\Ed}[1]{\mathcal{#1}}         % Edge set
\newcommand{\Grdif}[2]{{#1}^{\setminus #2}}         % Graph difference
\newcommand{\ncomp}[1]{s\left( #1 \right)}       % number components
\newcommand{\edge}{\scalebox{1.6}[1]{$-$}}

\newlength{\depthofsumsign}
\newcommand{\bigindep}[1][1.4]{% only for \displaystyle
    \mathop{%
        \raisebox
            {-#1\depthofsumsign+1\depthofsumsign}
            {\scalebox
                {#1}
                {$\displaystyle\indep$}%
            }
    }
}
\newcommand{\bigindepF}[1]{#1: \bigindep}

\newcommand{\bigscolon}[1][1.6]{% only for \displaystyle
    \mathop{%
            \depthofsumsign+1\depthofsumsign}
            {\scalebox
                {#1}
                {$\displaystyle ;$}%
            }
}

\newcommand{\loss}{\mathcal{L}}

%% file: library.bib
@ARTICLE{Yeung1991,
  author={Yeung, Raymond W.},
  journal={IEEE Transactions on Information Theory}, 
  title={A new outlook on {S}hannon's information measures}, 
  year={1991},
  volume={37},
  number={3},
  pages={466-474},
  doi={10.1109/18.79902}
}

@book{Chaitin1987,
author = {Chaitin, Gregory. J.},
doi = {10.1017/CBO9780511608858},
isbn = {9780521343060},
month = {Oct.},
publisher = {Cambridge University Press},
title = {Algorithmic Information Theory},
year = {1987}
}

@book{Li1997,
  author    = {Ming Li and Paul Vit{\'a}nyi},
  title     = {An Introduction to {K}olmogorov Complexity and Its Applications},
  edition   = {4th},
  year      = {2019},
  publisher = {Springer Cham},
  isbn      = {9783030112974},
  doi       = {10.1007/978-3-030-11298-1}
}

@article{Steudel2010,
archivePrefix = {arXiv},
arXivId = {1002.4020},
author = {Steudel, Bastian and Janzing, Dominik and Sch{\"{o}}lkopf, Bernhard},
eprint = {1002.4020},
isbn = {9780982252925},
journal = {Conference on Learning Theory},
pages = {464--476},
title = {Causal {M}arkov condition for submodular information measures},
year = {2010},
url = {https://www.learningtheory.org/colt2010/conference-website/papers/COLT2010proceedings.pdf}
}

@book{Cover2005,
author = {Cover, Thomas M. and Thomas, Joy A.},
booktitle = {Elements of Information Theory},
howpublished = {Hardcover},
isbn = {9780471241959},
pages = {1--748},
publisher = {Wiley-Interscience},
title = {Elements of Information Theory},
year = {2006}
}

@article{Watanabe1960,
author = {Watanabe, Satosi},
doi = {10.1147/rd.41.0066},
journal = {IBM Journal of Research and Development},
number = {1},
pages = {66--82},
title = {Information Theoretical Analysis of Multivariate Correlation},
volume = {4},
year = {1960}
}

@book{Yeung2008,
author = {Yeung, Raymond W.},
edition = {1st},
isbn = {0387792333},
issn = {0009-4978},
pages = {46--5079--46--5079},
publisher = {Springer Publishing Company, Incorporated},
title = {Information Theory and Network Coding},
volume = {46},
year = {2008}
}

@article{Yeung2002a,
author = {Yeung, Raymond W. and Lee, Tony T. and Ye, Zhongxing},
doi = {10.1109/TIT.2002.1013139},
issn = {0018-9448},
journal = {IEEE Transactions on Information Theory},
number = {7},
pages = {1996--2011},
publisher = {IEEE},
title = {Information-Theoretic Characterizations of Conditional Mutual Independence and {M}arkov Random Fields},
volume = {48},
year = {2002}
}

@article{Han1978,
author = {Han, Te S.},
title = {Nonnegative entropy measures of multivariate symmetric correlations},
doi = {10.1016/S0019-9958(78)90275-9},
issn = {0019-9958},
journal = {Information and Control},
number = {2},
pages = {133--156},
volume = {36},
year = {1978}
}

@article{Yeung2019,
archivePrefix = {arXiv},
arxivId = {1608.03697},
author = {Yeung, Raymond W. and Al-Bashabsheh, Ali and Chen, Chao and Chen, Qi and Moulin, Pierre},
doi = {10.1109/TIT.2018.2866564},
eprint = {1608.03697},
issn = {00189448},
journal = {IEEE Transactions on Information Theory},
number = {3},
pages = {1493--1511},
title = {On Information-Theoretic Characterizations of Markov Random Fields and Subfields},
volume = {65},
year = {2019}
}

@article{Hu1962,
author = {Hu, Kuo T.},
doi = {10.1137/1107041},
journal = {Theory of Probability and Its Applications},
number = {4},
pages = {439--447},
title = {On the Amount of Information},
volume = {7},
year = {1962}
}

@article{Hochschild1945,
author = {Hochschild, Gerhard},
issn = {0003-486X},
journal = {Annals of Mathematics},
number = {1},
pages = {58--67},
publisher = {Annals of Mathematics},
title = {On the Cohomology Groups of an Associative Algebra},
volume = {46},
year = {1945},
doi = {10.2307/1969145}
}

@article{Rosas2019,
archivePrefix = {arXiv},
arXivId = {1902.11239},
author = {Rosas, Fernando E. and Mediano, Pedro A.M. and Gastpar, Michael and Jensen, Henrik J.},
doi = {10.1103/PhysRevE.100.032305},
eprint = {1902.11239},
issn = {2470-0053},
journal = {Physical Review E},
number = {3},
pages = {1--17},
pmid = {31640038},
title = {Quantifying High-Order Interdependencies via Multivariate Extensions of the Mutual Information},
volume = {100},
year = {2019}
}

@article{Dawid2001,
author = {Dawid, Alexander P.},
doi = {10.1023/A:1016734104787},
issn = {1573-7470},
journal = {Annals of Mathematics and Artificial Intelligence},
number = {1},
pages = {335--372},
title = {Separoids: A Mathematical Framework for Conditional Independence and Irrelevance},
volume = {32},
year = {2001}
}

@Article{Baudot2015a,
AUTHOR = {Baudot, Pierre and Bennequin, Daniel},
TITLE = {The Homological Nature of Entropy},
JOURNAL = {Entropy},
VOLUME = {17},
YEAR = {2015},
NUMBER = {5},
PAGES = {3253--3318},
URL = {https://www.mdpi.com/1099-4300/17/5/3253},
ISSN = {1099-4300},
DOI = {10.3390/e17053253}
}

@article{Kawabata1992,
author = {Kawabata, Tsutomu and Yeung, Raymond W.},
doi = {10.1109/18.135658},
journal = {IEEE Transactions on Information Theory},
number = {3},
pages = {1146--1149},
title = {The structure of the {I}-measure of a {M}arkov chain},
volume = {38},
year = {1992}
}

@Article{Baudot2019,
AUTHOR = {Baudot, Pierre and Tapia, Monica and Bennequin, Daniel and Goaillard, Jean-Marc},
TITLE = {Topological Information Data Analysis},
JOURNAL = {Entropy},
VOLUME = {21},
YEAR = {2019},
NUMBER = {9},
ARTICLE-NUMBER = {869},
URL = {https://www.mdpi.com/1099-4300/21/9/869},
ISSN = {1099-4300},
DOI = {10.3390/e21090869}
}

@phdthesis{Vigneaux2019a,
  TITLE = {Topology of statistical systems : A cohomological approach to information theory},
  AUTHOR = {Vigneaux, Juan P.},
  URL = {https://theses.hal.science/tel-02951504},
  NUMBER = {2019USPCC070},
  SCHOOL = {{Universit{\'e} Sorbonne Paris Cit{\'e}}},
  YEAR = {2019},
  MONTH = {June},
  TYPE = {{PhD} Dissertation},
  HAL_ID = {tel-02951504},
  HAL_VERSION = {v1},
}

@misc{Forre2021a,
      title={Transitional Conditional Independence}, 
      author={Patrick Forré},
      year={2021},
      eprint={2104.11547},
      archivePrefix={arXiv},
      primaryClass={math.ST},
      url={https://arxiv.org/abs/2104.11547}, 
      howpublished = {arXiv preprint arXiv:2104.11547}
}

@ARTICLE{Cocco2012,
       author = {{Cocco}, Simosa and {Monasson}, R\'emi},
        title = {Adaptive Cluster Expansion for the Inverse {I}sing Problem: Convergence, Algorithm and Tests},
      journal = {Journal of Statistical Physics},
         year = 2012,
        month = {Apr.},
       volume = {147},
       number = {2},
        pages = {252-314},
          doi = {10.1007/s10955-012-0463-4},
archivePrefix = {arXiv},
       eprint = {1110.5416},
 primaryClass = {cond-mat.dis-nn},
       adsurl = {https://ui.adsabs.harvard.edu/abs/2012JSP...147..252C}
}

@ARTICLE{Amari2001,
  author={Amari, Shun'ichi},
  journal={IEEE Transactions on Information Theory}, 
  title={Information Geometry on Hierarchy of Probability Distributions}, 
  year={2001},
  volume={47},
  number={5},
  pages={1701-1711},
  doi={10.1109/18.930911}}

@book{Bishop2007,
  asin = {0387310738},
  author = {Bishop, Christopher M.},
  description = {Amazon.com: Pattern Recognition and Machine Learning (Information Science and Statistics): Christopher M. Bishop: Books},
  dewey = {006.4},
  ean = {9780387310732},
  edition = {1st},
  isbn = {0387310738},
  publisher = {Springer},
  title = {Pattern Recognition and Machine Learning},
  year = 2007
}

@book{Bishop2023,
  author = {Bishop, Christopher M. and Bishop, Hugh},
  doi = {10.1007/978-3-031-45468-4},
  edition = {1st},
  isbn = {9783031454684},
  language = {english},
  title = {Deep Learning: Foundations and Concepts},
  year = 2023,
  publisher = {Springer}
}

@InProceedings{Dickstein2015,
  title = 	 {Deep Unsupervised Learning using Nonequilibrium Thermodynamics},
  author = 	 {Sohl-Dickstein, Jascha and Weiss, Eric and Maheswaranathan, Niru and Ganguli, Surya},
  booktitle = 	 {Proceedings of the 32nd International Conference on Machine Learning},
  pages = 	 {2256--2265},
  year = 	 {2015},
  editor = 	 {Bach, Francis and Blei, David},
  volume = 	 {37},
  series = 	 {Proceedings of Machine Learning Research},
  address = 	 {Lille, France},
  month = 	 {July},
  publisher =    {PMLR},
  url = 	 {https://proceedings.mlr.press/v37/sohl-dickstein15.html}
}

@InProceedings{Ramesh2021,
  title = 	 {Zero-Shot Text-to-Image Generation},
  author =       {Ramesh, Aditya and Pavlov, Mikhail and Goh, Gabriel and Gray, Scott and Voss, Chelsea and Radford, Alec and Chen, Mark and Sutskever, Ilya},
  booktitle = 	 {Proceedings of the 38th International Conference on Machine Learning},
  pages = 	 {8821--8831},
  year = 	 {2021},
  editor = 	 {Meila, Marina and Zhang, Tong},
  volume = 	 {139},
  series = 	 {Proceedings of Machine Learning Research},
  month = 	 {July},
  publisher =    {PMLR},
  url = 	 {https://proceedings.mlr.press/v139/ramesh21a.html}
}

@inproceedings{Chitwan2022,
 author = {Saharia, Chitwan and Chan, William and Saxena, Saurabh and Li, Lala and Whang, Jay and Denton, Emily L. and Ghasemipour, Kamyar and Gontijo Lopes, Raphael and Karagol Ayan, Burcu and Salimans, Tim and Ho, Jonathan and Fleet, David J. and Norouzi, Mohammad},
 booktitle = {Advances in Neural Information Processing Systems},
 editor = {S. Koyejo and S. Mohamed and A. Agarwal and D. Belgrave and K. Cho and A. Oh},
 pages = {36479--36494},
 publisher = {Curran Associates, Inc.},
 title = {Photorealistic Text-to-Image Diffusion Models with Deep Language Understanding},
 url = {https://proceedings.neurips.cc/paper_files/paper/2022/file/ec795aeadae0b7d230fa35cbaf04c041-Paper-Conference.pdf},
 volume = {35},
 year = {2022}
}

@INPROCEEDINGS{Stable_Diffusion2021,
  author={Rombach, Robin and Blattmann, Andreas and Lorenz, Dominik and Esser, Patrick and Ommer, Björn},
  booktitle={2022 IEEE/CVF Conference on Computer Vision and Pattern Recognition (CVPR)}, 
  title={High-Resolution Image Synthesis with Latent Diffusion Models}, 
  year={2022},
  pages={10674-10685},
  doi={10.1109/CVPR52688.2022.01042}
}

@article{Medina-Mardones2021,
archivePrefix = {arXiv},
arXivId = {2010.01117},
author = {Medina-Mardones, Anibal M. and Rosas, Fernando E. and Rodr{\'{i}}guez, Sebasti{\'{a}}n E. and Cofr{\'{e}}, Rodrigo},
doi = {10.1088/2632-072X/abf231},
eprint = {2010.01117},
issn = {2632-072X},
journal = {Journal of Physics: Complexity},
number = {3},
title = {Hyperharmonic analysis for the study of high-order information-theoretic signals},
volume = {2},
year = {2021}
}

@InProceedings{Giry1982,
author="Giry, Mich{\`e}le",
editor="Banaschewski, B.",
title="A categorical approach to probability theory",
booktitle="Categorical Aspects of Topology and Analysis",
year="1982",
publisher="Springer Berlin Heidelberg",
pages="68--85",
isbn="9783540390411",
doi="10.1007/BFb0092872"
}

@article{Lang2022,
    title = {Information Decomposition Diagrams Applied beyond {S}hannon Entropy: A Generalization of {H}u's {T}heorem},
    author     = {Leon Lang and Pierre Baudot and Rick Quax and Patrick Forré},
    url        = {https://compositionality.episciences.org/14181},
    doi        = {10.46298/compositionality-7-1},
    journal    = {Compositionality},
    issn       = {2631-4444},
    volume     = {Volume 7},
    eid        = 1,
    year       = {2025},
    month      = {Jan},
}

@misc{Hammersley1971,
author = {Hammersley, John M. and Clifford, Peter},
note = {Unpublished Manuscript},
pages = {1--26},
title = {Markov Fields on Finite Graphs and Lattices},
year = {1971},
url = {https://www.statslab.cam.ac.uk/~grg/books/hammfest/hamm-cliff.pdf}
}

@book{Preston1976,
  author = {Christopher Preston},
  title        = {Random Fields},
  series       = {Lecture Notes in Mathematics},
  volume       = {534},
  year         = {1976},
  publisher    = {Springer},
  isbn         = {9783540078524},
  pages        = {212},
  language     = {English}
}

@book{Spitzer1971,
author = {Spitzer, Frank L.},
publisher = {Mathematical Association of America},
title = {Random Fields and Interacting Particle Systems},
year = {1971}
}

@inproceedings{Pearl1985,
  author       = {Judea Pearl},
  title        = {Bayesian Networks: A Model of Self-Activated Memory for Evidential Reasoning},
  booktitle    = {Proceedings of the Seventh Annual Conference of the Cognitive Science Society},
  year         = {1985},
  address      = {Irvine, CA},
  pages        = {329--334},
  url          = {http://ftp.cs.ucla.edu/tech-report/198_-reports/850017.pdf}
}
